\documentclass{article}

\usepackage[english]{babel}
\usepackage{times}
\usepackage[T1]{fontenc}
\usepackage{comment}
\usepackage{natbib}
\usepackage{pifont}
\usepackage{tikz}
\usetikzlibrary{graphs, arrows, calc, positioning, matrix}

\usepackage{amsthm}
\usepackage{amsmath,amssymb,enumerate,rotate,multirow}
\usepackage{bm,latexsym,mathrsfs,bbm,amsfonts}
\usepackage{color}
\usepackage{verbatim}
\usepackage{subfig}
\usepackage{alphalph}

\usepackage{morefloats}
\usepackage{makecell}
\usepackage{pifont}
\usepackage{diagbox}

\usepackage{authblk}

\usepackage[toc]{appendix}

\usepackage{thmtools,thm-restate}

\newtheorem{theorem}{Theorem}

\theoremstyle{definition}
\newtheorem{definition}{Definition}

\newcommand{\iid}{\stackrel{\rm iid}{\sim}}

\begin{document}

\title{Repulsion, Chaos and Equilibrium in Mixture Models}

\date{}

\author[1,2]{Andrea Cremaschi}
\author[3]{Timothy M. Wertz}
\author[1,2,4]{Maria De Iorio}

\affil[1]{Singapore Institute for Clinical Sciences (SICS), Agency for Science, Technology and Research (A*STAR),  Singapore, Republic of Singapore}
\affil[2]{Department of Paediatrics, Yong Loo Lin School of Medicine, National University of Singapore (NUS), Singapore, Republic of Singapore}
\affil[3]{Department of Mathematics, National University of Singapore (NUS), Singapore, Republic of Singapore}
\affil[4]{Department of Statistical Science, University College London (UCL), London, UK}

\maketitle

\section*{Abstract}
Mixture models are commonly used in applications with heterogeneity and overdispersion in the population, as they allow the identification of subpopulations. In the Bayesian framework, this entails the specification of suitable prior distributions for the weights and location parameters of the mixture. Widely used are Bayesian semi-parametric models based on mixtures with infinite or random number of components, such as Dirichlet process mixtures (Lo, 1984) or mixtures with random number of components (Miller and Harrison, 2018). Key in this context is the choice of the kernel for cluster identification. Despite their popularity, the flexibility of these models and prior distributions often does not translate into interpretability of the identified clusters. To overcome this issue, clustering methods based on repulsive mixtures have been recently proposed (Quinlan et al., 2021). The basic idea is to include a repulsive term in the prior distribution of the atoms of the mixture, which favours mixture locations far apart. This approach is increasingly popular and allows one to produce well-separated clusters, thus facilitating the interpretation of the results. However, the resulting models are usually not easy to handle due to the introduction of unknown normalising constants. Exploiting results from statistical mechanics, we propose in this work a novel class of repulsive prior distributions based on Gibbs measures. Specifically, we use Gibbs measures associated to joint distributions of eigenvalues of random matrices, which naturally possess a repulsive property. The proposed framework greatly simplifies the computations needed for the use of repulsive mixtures due to the availability of the normalising constant in closed form. We investigate theoretical properties of such class of prior distributions, and illustrate the novel class of priors and their properties, as well as their clustering performance, on benchmark datasets.

\section{Mixture models with repulsive component}
Mixture models are a very powerful and natural statistical tool to model data from heterogeneous populations. In a mixture model, observations are assumed to have arisen from one of $M$ (finite or infinite) groups, each group being suitably modelled by a density, typically from a  parametric family. The density of each group is referred to as a component of the mixture, and is weighed by the relative frequency (weight) of the group in the population. This model offers a conceptually simple way of relaxing distributional assumptions and a convenient and flexible method to approximate distributions that cannot be modelled satisfactorily by a standard parametric family. Moreover, it provides a framework by which observations may be clustered together into groups for discrimination or classification. For a comprehensive review of mixture models and their applications see \cite{mclachlan2000finite, fruhwirth2006finite} and \cite{fruhwirth2019handbook}. A mixture model for a vector of $d$-dimensional observations $\bm y_1, \dots, \bm y_n$ is usually defined as:
\begin{align}\label{eq:mixture1}
	\bm y_i \mid \bm w, \bm \theta, M \sim \sum_{j = 1}^M w_j f(\bm y_i \mid \theta_j) \quad i = 1, \dots, n
\end{align}
where the function $f(\bm y \mid \bm \theta)$, referred to as the \textit{kernel}, represents the chosen sampling model for the observations (often a parametric distribution such as the Gaussian distribution), $\bm w = (w_1, \dots, w_M)$ is a vector of normalised weights and $\bm \theta = (\theta_1, \dots, \theta_M)$ is an array of kernel-specific parameters. The number of components (sub-populations) in the mixture is equal to $M$, which can be either fixed or random. In this work, we consider the latter case.

An important feature of mixture models is their ability to identify sub-populations by allowing for clustering of the subjects. Conditionally on cluster allocation and model parameters, observations are independent and identically  distributed within the groups and independent between groups. Indeed, model \eqref{eq:mixture1} can be re-written introducing a vector of latent allocation variables $\bm c = \left(c_1, \dots, c_n\right)$ indicating the allocation of observations to a mixture component:
\begin{align}\label{eq:mixture2}
	&\bm y_i \mid c_i, \bm \theta, M \sim f(\bm y_i \mid \theta_{c_i}) \quad i = 1, \dots, n \nonumber \\
	&c_1, \dots, c_n \mid \bm w, M \iid \text{Multinomial}(1, \bm w) 
\end{align}
where $\text{Multinomial}(1, \bm w)$ denotes the multinomial distribution of size 1 and probability vector $\bm w$. The model is completed by specifying  prior distributions on the remaining parameters: 
\begin{align}\label{eq:mixprior}
	&\theta_1, \dots, \theta_M \mid M \sim P_0(\bm \theta) \nonumber\\
	&\bm w \mid M \sim \pi_{\bm w}\\
	&M \sim \pi_M \nonumber 
\end{align}
Thus, observations are partitioned into clusters such that $\theta_{c_i} = \theta_{c_j}$ iff $i$ and $j$ belong to the same component. The number $K$ of unique values in the vector $\theta_{c_1}, \dots, \theta_{c_n}$ represents the number of clusters. It is important to highlight the distinction between $M$ and $K$: $M$ refers to the data-generation process and denotes the number of components in a mixture, i.e. of possible clusters/sub-populations, while the number of clusters, $K$, represents number of allocated components, i.e. components to which at least one observation has been assigned \citep[see][]{argiento2022infinity}. In general, both $M$ and $K$ are unknown and object of posterior inference in the study of finite mixtures (i.e., when $M < +\infty$). Still, even when $M$ is fixed in a finite mixture model, i.e. the number of components in the population is fixed, we need to estimate $K$, the actual number of clusters in the sample (allocated components) \citep[see][]{rousseau2011asymptotic}.
On the other hand, in different settings such as Bayesian nonparametrics, $M = + \infty$ and the object of interest is only $K$. Clustering is of importance in many applications where a more parsimonious representation of the data is desired, or where the identification of subpopulations is relevant (e.g., patients risk groups). 
The specific choice of the kernel, as well as the prior distribution on $M$, the parameters $\bm \theta$ and the weights $\bm w$ plays a crucial role in defining the clustering output. Various features of mixture models have been carefully investigated in the literature, together with associated computational schemes \citep{mclachlan2000finite,fruhwirth2019handbook, fruhwirth2019here, argiento2022infinity}.

In this paper, we focus on the specification of a prior distribution $P_0$ for the location parameters $\bm \theta$. Typically, the location parameters $\theta_j$ are assumed i.i.d. from $P_0$. However, such assumption can be too restrictive in several applications where it is desirable to introduce dependence among the locations to improve interpretability of the results. A popular example of this approach is the specification of \textit{repulsive mixtures}, which have recently attracted increasing interest in the literature on model-based clustering \citep[see, for example,][]{petralia2012repulsive, bianchini2020determinantal, xie2020bayesian, quinlan2021class}. The rationale behind this approach is purely empirical and based on the notion of distance between clusters, reflecting the requirement of more separated clusters to improve interpretability, a property referred to as \textit{repulsion}. As pointed out in \cite{hennig2013find}, the properties of the clustering method used should be a reflection of the definition of clusters given by the user, rather than based on the assumption that an underlying true partition of the data exists. Following this idea, the specification of repulsive mixtures reflects the need to improve the interpretation of the resulting partition by enhancing separation between observations in different clusters. Indeed, \cite{quinlan2017parsimonious} argue that repulsion promotes the reduction of redundant mixture components (or singletons) without substantially sacrificing goodness-of-fit, favouring a-priori subjects to be allocated to a few well-separated clusters. This strategy offers a compromise between the desire to remove redundant (or singleton) clusters that often appear when modelling location parameters of mixture components independently, and the forced parsimony induced by hard types of repulsion. An alternative definition of repulsive mixture is provided by \cite{malsiner2017identifying}, whose approach encourages nearby components to merge into groups at a first 
hierarchical level and then to enforce between-group separation at the second level. A similar idea has been employed in \cite{natarajan2021cohesion} in the context of distance-based clustering, where the repulsive term appears at the likelihood level. Finally, we note that \cite{fuquene2019choosing} propose the use on Non-Local priors (NLP) to select the number of components, characterised by improved parsimony obtained through the inclusion of a penalty term, and leading  to well-separated components with non-negligible weight, interpretable as distinct subpopulations. Despite similarities between NLPs and repulsive over-fitted mixtures, the former approach requires not only a repulsive force between the locations, but also penalising low weight components, which leads to better model performance \citep{fuquene2019choosing}. Still, they fit their model for different  numbers of components and compare them  through model choice criteria based on estimates of the marginal likelihood, without performing full posterior inference on $M$.

Two popular approaches to introduce repulsion among the elements of $\bm \theta$ are determinantal point processes (DPP) and the inclusion of a repulsive term in the specification of $P_0$, borrowing ideas from Gibbs point processes (GPP). There is an interesting connection between GPPs and DPPs \citep{georgii2005conditional,lavancier2015determinantal} as DPPs can be considered as a subclass of GPPs, at least when they are defined on a bounded region. Since this link is of limited interested in what follows, we will not discuss it further.
	
DPPs were firstly introduced in \cite{macchi1975coincidence} as \textit{fermion processes}, since they are used to describe the behaviour of systems of fermions, subatomic particles exhibiting an ``antibunching'' effect, i.e. their configuration tends to be well separated. A DPP is defined as a point process where the joint distribution of the points (i.e., the particle configuration) is expressed in terms the determinant of a positive semidefinite matrix. The repulsion property of the DPPs derives from the characterisation of the determinant as the hyper-volume of the parallelepiped spanned by the columns of the corresponding matrix. As such, the probability of a configuration grows as the columns are farther apart from each other in $\mathbb{R}^d$ equipped with the Euclidean norm. See \cite{macchi1975coincidence, borodin2005eynard, hough2009zeros, kulesza2012determinantal, lavancier2015determinantal} and references therein for theoretical and computational details on DPPs. To the best of our knowledge, \cite{kwok2012priors} are the first to employ a DPP as a prior distribution for the locations in mixture models, highlighting the repulsive property of the DPP. In their work, inference is performed via a maximum a-posteriori estimation. \cite{affandi2013approximate} adopt a DPP in a mixture model with a fixed number of components (the $K$-DPP), and propose two sampling schemes for posterior inference. More recently, \cite{xu2016bayesian, bianchini2020determinantal} and \cite{beraha2022mcmc} propose the use of DPPs in Bayesian mixture models with a random number of components $M$. Posterior inference in \cite{xu2016bayesian} and \cite{bianchini2020determinantal} is performed through the labour-intensive reversible jump algorithm \citep{green1995reversible}, while \cite{beraha2022mcmc} propose an algorithm which exploits the construction by \cite{argiento2022infinity}, and implement a sampling scheme based on the Metropolis-Hastings birth-and-death algorithm by \cite{geyer1994simulation}. In general, these methods do not scale well with the dimension of the parameters of interest due to the inherent double-intractability of the posterior \citep{murray2006mcmc}, and require the implementation of tailored algorithms due to the presence of a prior on the number of components $M$.

Another common strategy for repulsive mixtures is to directly specify a repulsion term in the prior density function corresponding to $P_0$. The main idea behind this approach is to start with the usual independent prior on the locations, e.g. a product of Gaussian distributions, and include, in a fairly \textit{ad-hoc} manner, a multiplicative factor which represents a \textit{penalty} term often defined on the basis of the pairwise distances between the parameters of interest \citep{petralia2012repulsive, xie2020bayesian, quinlan2021class}. Borrowing ideas from statistical mechanics, in analogy  with the behaviour of gas particles interacting with each other, the penalty term describes the repulsion among the location parameters of the mixture. This approach, arguably the most common in practice, presents computational and theoretical drawbacks. It is discussed in detail in Section \ref{sec:gpp} as it is one of the main motivations of this work. In Figure~\ref{fig:PointProcesses_examples} we show two-dimensional realisations from an independent Gaussian prior, a DPP and the repulsive prior of \cite{quinlan2021class}. 
 
In this work, we take a different approach and specify tractable joint distributions, still within the class of GPPs, presenting connections with statistical mechanics and the mathematical theory of gases. We specify a novel class of repulsive distributions for the location parameters of mixture models with random number of components, based on GPPs and specifically on the joint distribution of the eigenvalues of random matrices, once again providing an  interpretation of the concept of repulsion in terms of interacting particles. These distributions are linked to the  joint Gibbs canonical distributions used to model \textit{Coulomb gases}, also called log-gases \citep{dyson1962statistical, forrester2010log}. 

The paper is structured as follows. In Section~\ref{sec:statmec} we recall basic concepts of statistical mechanics, while in Section~\ref{sec:gpp} we describe the link between Gibbs Point Processes and repulsive mixtures. In Section~\ref{sec:repulsive_priors} we introduce the novel class of repulsive prior distributions, based on the joint law of eigenvalues of random matrices, whose theoretical properties are investigated in Section~\ref{sec:gibbs-measures-and-chaos}. In Section~\ref{sec:LDP} we prove that the proposed class of prior distributions satisfies the large deviation principle. In Section~\ref{sec:model-and-algorithm} we specify the full mixture model with random number of components and a repulsive prior on the locations. We demonstrate the approach in simulations and on a real data application in Section~\ref{sec:examples}. Finally, we conclude the paper in Section~\ref{sec:concl}.

\begin{figure}[ht]
	\centering
	\subfloat[]{\includegraphics[width=0.35\linewidth]{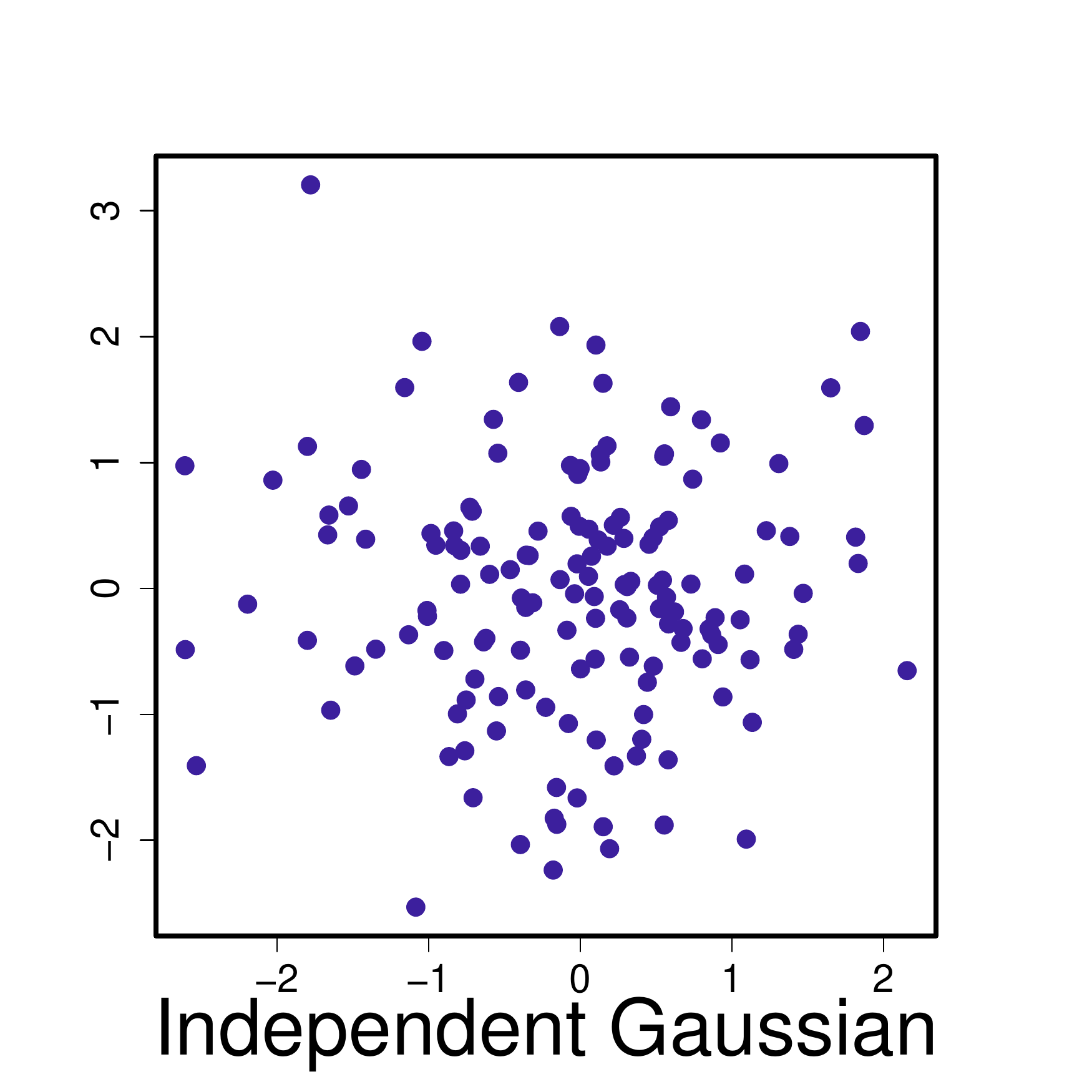}}
	\subfloat[]{\includegraphics[width=0.35\linewidth]{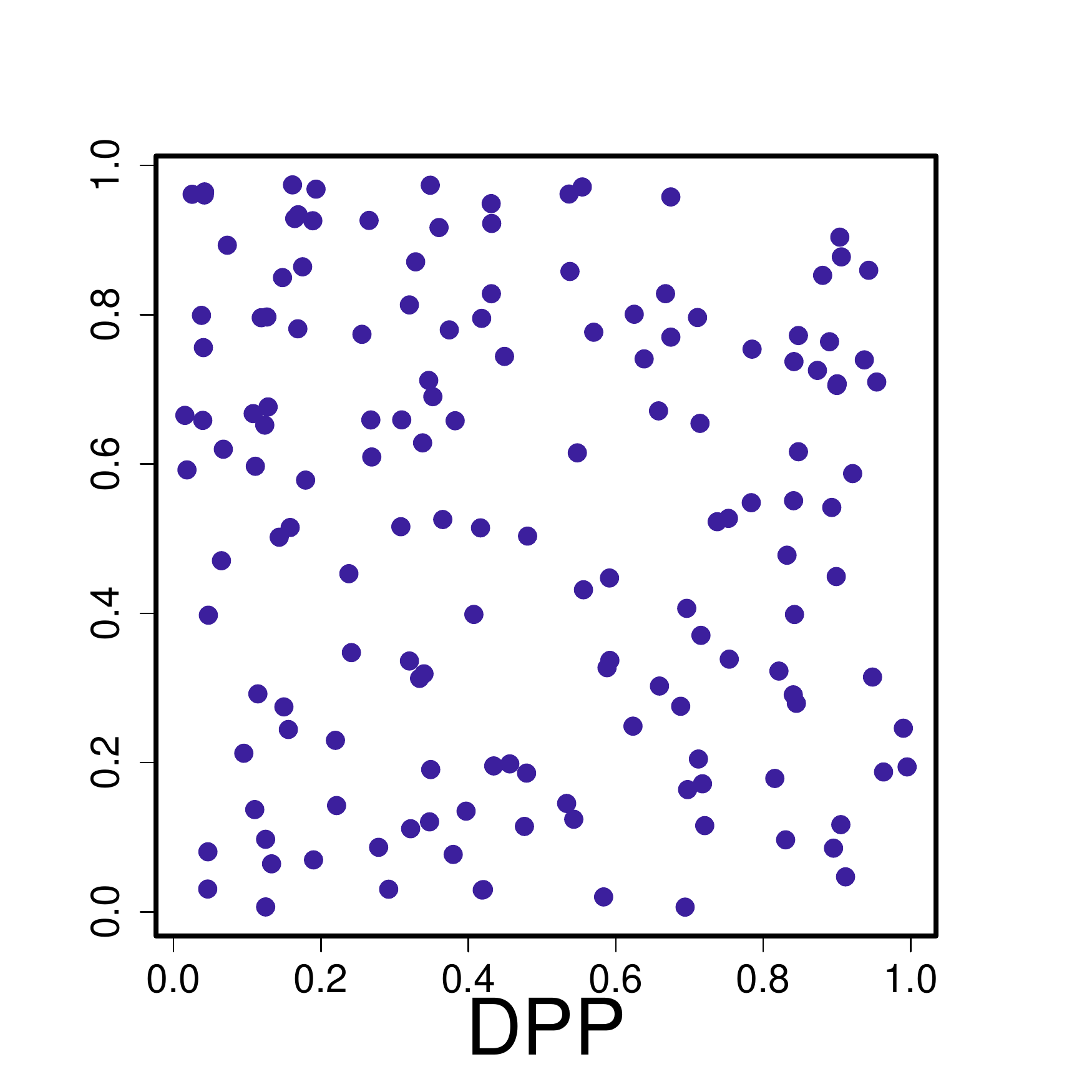}}
	\subfloat[]{\includegraphics[width=0.35\linewidth]{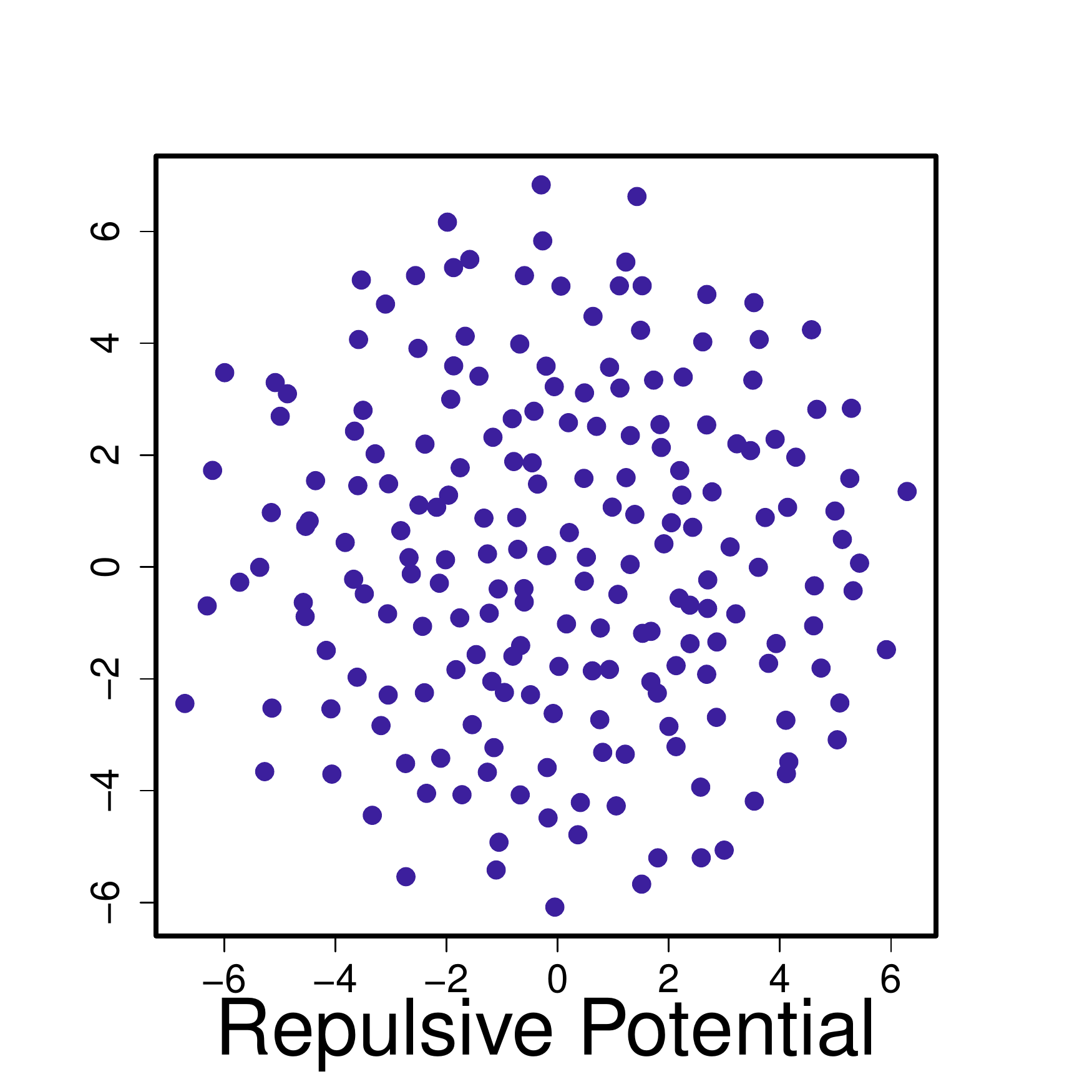}}
	\caption{Samples of size $N = 100$ points from three different distributions in $\mathbb{R}^2$: (a) independent Gaussian; (b) Determinantal Point Process; (c) Repulsive mixture as in \cite{quinlan2021class}.} 
	\label{fig:PointProcesses_examples}
\end{figure}

\section{General concepts of statistical mechanics}\label{sec:statmec}

In this Section, we introduce basic concepts from statistical mechanics necessary for the following development.

An individual configuration of the particles is referred to as \textit{microstate}, while the macroscopic observables of interest (e.g., energy of the system) are called \textit{macrostates}. Each macrostate can be associated to several microstates, since different particle configurations can lead to the same macroscopic quantity of interest. An \textit{ensemble} is a set of microstates of a system that are consistent with a given macrostate. In this framework, statistical mechanics examines an ensemble of microstates corresponding to a given macrostate by providing their probability distribution. Therefore, the macroscopic properties of a system are derived from probability distributions describing the interactions among the particles. Note that these probability distributions are \textit{conditional} on a particular macrostate. The systems studied in statistical mechanics often involve particles interacting with an external environment, called a \textit{reservoir}. In particular, the microstates are influenced by the macroscopic features of the surrounding environment, thus characterising their probability distribution.
The three main ensembles, corresponding to different assumptions on the system conditions, are \citep{landau1968statistical, mandl1991statistical}:
\begin{itemize}
    \item the \textit{microcanonical} ensemble: assumes an isolated system in which the energy is pre-specified. There are no exchanges, either of energy or particles, with the surrounding environment.
	
    \item the \textit{canonical} ensemble: assumes a system interacting with an external reservoir. The system is kept at a fixed temperature, while the energy is allowed to vary with the particle configuration. The total energy of the combined system (reservoir and particles) is fixed. The probability of each configuration is given by the Boltzmann distribution.
	
    \item the \textit{grand canonical} ensemble: assumes a system where both energy and particles are allowed to fluctuate between the system and the reservoir. The probability distribution describing the particle configuration presents an additional term involving the number of particles in the system. Both the total energy and number of particles of the combined systems are fixed. The distribution of the particle configuration is known as the Gibbs distribution.
\end{itemize}
Statistical mechanics often focuses on a system in \textit{equilibrium}, i.e. for which the probability distribution over the ensemble does not have an explicit time dependence \citep{tuckerman2010statistical}. We assume that this is the case throughout this work. From a mathematical point of view, this implies that the probability distribution of the particles configuration in a particular ensemble is a function of the total energy as expressed by the \textit{Hamiltonian}, an operator containing kinetic and potential energy terms describing the system \citep{tuckerman2010statistical}. 

Let $\bm \theta = \left(\theta_1, \dots, \theta_M\right)$ be a set of $M$ random variables  describing the position of the particles on the $d$-dimensional lattice $\mathbb{Z}^d$ and let $\Omega$ be the set of possible configurations of such particles.
The internal energy of the system $\overline{E}$ can be computed via the Hamiltonian as a function of the microstate $\omega \in \Omega$ describing the pairwise interactions between the particles and the relationship with the reservoir. A general expression for the Hamiltonian of a system is:
\begin{equation}\label{eq:EnergyIntro}
	\mathcal{H}(\bm \theta \mid \zeta, h) = h \sum_{i = 1}^M \psi_1(\theta_i) + \zeta \sum_{1\leq i < j\leq M} \psi_2(\theta_i, \theta_j)
\end{equation}
where $h, \zeta \in \mathbb{R}$, while $\psi_1$ and $\psi_2 $ represent the \textit{self energy} and the \textit{interaction function}, respectively \citep{domb2000phase}. The latter specifies how the particles interact with each other, often through the use of pairwise terms. On the other hand, $\psi_1$ captures how the particles are affected by the presence of fields external to the system. An example of Hamiltonian is the one corresponding to the popular Ising model, used to study the behaviour of a system of ferromagnetic particles immersed in a magnetic field. In this case, the Hamiltonian is:
\[
\mathcal{H}(\bm \theta \mid \zeta, h) = h \sum_{i = 1}^M \theta_i + \zeta \sum_{i,j : |i - j| = 1} \theta_i \theta_j
\] 
where $\theta_i \in \{-1,1\}$, indicating the positive or negative spin of the particles, and $\zeta > 0$ is the \textit{inverse temperature}.
  
A fundamental concept in physics is that of Entropy, which plays a crucial role in determining the probability distribution of the configuration of particles, given different types of ensemble \citep{maxwell1860ii}. The Entropy of a system is given by \citep{boltzmann1866}:
\begin{equation}\label{eq:Boltzmann_Entropy}
	S = \kappa_{B} \ln W
\end{equation}
where $\kappa_{B} \approx 1.38 \times 10^{-23}$J/K is the Boltzmann constant, and $W$ is the number of microstates associated with a given macrostate $\overline{E}$ (e.g., energy of the system). Note that $W$ depends on the type of ensemble assumed to describe the system of particles and on the assumption of equilibrium. Given $W$, the probability distributions of each of the three ensembles can be derived \citep{landau1968statistical,landau1980chapter, bowley1999introductory}:

\begin{itemize}
    \item Microcanonical ensemble. Boltzmann's \textit{postulate of equal a-priori probabilities} assumes that each configuration of the particles is equally probable. Therefore, since in this case this is an isolated system with fixed energy and number of particles, the probability of observing a microstate $\omega$ is given by the inverse of the number of microstates $W$, i.e. $p_{\omega} = 1/W$.
    \item Canonical ensemble. The Boltzmann distribution is given by:
	\begin{equation}\label{eq:Boltzmann_distr}
		p_{\omega} = \frac{e^{-\mathcal{H}(\bm \theta \mid \zeta, h) / (\kappa_B T)}}{Z^{c}_M}
	\end{equation}
	where the normalising constant $Z^{c}_M = \sum_{\omega \in \Omega} e^{-\mathcal{H}(\bm \theta \mid \zeta, h) / \kappa_B T}$ is the \textit{partition function}. Due to its role as unit conversion factor in the description of thermodynamic systems, the constant $\kappa_B $ is usually called \textit{inverse temperature} and $T$ is the temperature.
    \item Grand canonical ensemble. The additional assumption of exchange of particles between the system and the reservoir affects the expression of the probability distribution, by including an additional term reflecting the migration of particles between the two sub-systems:
	\begin{equation}\label{eq:Gibbs_distr}
		p_{\omega} = \frac{e^{\left(\mu N_{\omega} - \mathcal{H}(\bm \theta \mid \zeta, h) \right)/ (\kappa_B T)}}{Z^{gc}_M}
	\end{equation}
	where $\mu$ is the \textit{chemical potential}, $N_{\omega}$ is the number of particle in the microstate $\omega$ and $Z^{gc}_M = \sum_{\omega \in \Omega} e^{\left(\mu N_{\omega} -\mathcal{H}(\bm \theta \mid \zeta, h) \right) / \kappa_B T}$ is the partition function.
\end{itemize}

From an information-theoretic perspective, the above distributions have the property of maximising the Shannon entropy, a measure of the amount of information or uncertainty about the possible outcomes of a random variable \citep{shannon1948mathematical}. In statistical mechanics, this is referred to as the Gibbs entropy. Given the postulate of a-priori probabilities, the Boltzmann entropy is a special case of the Shannon entropy, where all the probabilities are equal. For more details on this topic see \cite{jaynes1965gibbs}.
The higher the value of the entropy for the distribution of microstates over an ensemble, the higher the uncertainty around the distribution of the particle configurations. In practice, choosing the distribution maximising the entropy of a system, i.e. the Boltzmann distribution, corresponds to choosing the flattest possible distribution over the microstates compatible with the available information, i.e. the macrostate $\overline{E}$. In a Bayesian setting, this concept is analogous to that of Gibbs posterior \citep[see, for instance,][and references therein]{jiang2008gibbs, rigon2020generalized}, where a likelihood-free approach is devised for the estimation of a set of parameters of interest, directly specifying their posterior distribution via a standard prior and a term depending on a loss function with desired properties, yielding an expression similar to Eq.~\eqref{eq:Boltzmann_distr}. \cite{bissiri2016general} show that this approach has good theoretical properties relating to a maximum-entropy principle. Finally, this approach is reminiscent of the ``product of approximate conditionals'' (PAC) likelihood \citep{cornuet2007note,li2003modeling}.  

This work focuses on the description and discussion of Bayesian mixing distributions characterised by a repulsive term. The latter presents an interesting analogy with the distributions arising in statistical mechanics under different ensemble assumptions. Indeed, there is a parallelism between particles in a thermodynamic system and location parameters in a mixture model with repulsion terms. The locations and the number of components of the mixture can be associated to the position and number of particles of a physical system, while the repulsion term in their prior distribution can be related to their interaction in a given configuration (i.e. the function $\psi_2$ in Eq.~\eqref{eq:EnergyIntro}). The three ensembles are recovered by making different assumptions on the prior for the locations. For instance, the microcanonical ensemble is recovered by assuming a mixture model with fixed number of components and no repulsion term with the locations i.i.d. draws from $P_0$, the canonical ensemble corresponds to a mixture model with fixed number of components and non-zero repulsion (introducing dependence among the locations), while the grand canonical ensemble additionally allows for a prior distribution on the number of components.

\section{Gibbs Point Processes and Repulsive Mixtures}\label{sec:gpp}

Gibbs Point processes \citep[][page 127]{daley2003introduction} are a fundamental class of point processes arising in statistical physics to describe forces acting on and between particles. Moreover, point processes can be seen as limits of ensembles \citep{Holcomb2015}. Gibbs processes are generated by interaction potentials, as described in the previous section. The total potential energy corresponding to a given configuration of particles is assumed to be decomposable into terms representing the interactions between the particles taken in pairs, triples, and so on; first-order terms representing the potential energies of the individual particles due to the action of an external force field may also be included. Given a \textit{potential} $\mathcal{H}$, a Gibbs process, i.e. homogeneous spatial point pattern \citep{daley2003introduction}, is directly specified by using the \textit{Boltzmann distribution} for the configuration of a set of particles (also referred to as Gibbs canonical distribution), with Janossy density: 
\begin{equation}\label{eq:Gibbs_Density_gen}
	p_M\left(\bm \theta \mid \zeta \right) = Z^{-1}_M(\zeta) \exp\left\{-\zeta \mathcal{H}(\bm \theta) \right\} 
\end{equation}
where $\bm \theta = (\theta_1,\ldots, \theta_M)$.

In Eq. \eqref{eq:Gibbs_Density_gen} $\zeta$ is a parameter controlling the amount of repulsion and $Z_M(\zeta)$ is a normalisation constant called \textit{partition function}, which plays an important role in the study of particle systems, as well as in the study of repulsive mixture models, as discussed later. Point processes, involved in the description of particle configurations, are usually characterised by interaction between the points. The interaction can be attractive or repulsive, depending on geometrical features, whereas the null interaction is associated to the well-known Poisson point process. Frequently, it is supposed that only the first- and second-order terms need to be included, so that the process is determined by the point pair potentials. In this case we have {\em repulsive interactions with pair potential}:
\begin{equation}\label{eq:pairwisepotential}
 	\mathcal{H}\left(\bm \theta\right) = \sum_{i=1}^M \psi_1(\theta_i) + \sum_{1\leq i < j\leq M} \psi_2(\theta_i, \theta_j)   
\end{equation}
for appropriate choice of $\psi_1$ and $\psi_2$. Commonly, three types of potentials are used to model the pairwise interactions \citep{daley2003introduction}:
\begin{align}\label{eq:potentials}
	& \phi_1(r) = - \log \left(1 - e^{-(r/\zeta)^2} \right) \nonumber \\
	& \phi_2(r) = \left(\zeta / r \right)^{m_1} - \left(\zeta / r \right)^{m_2}, \quad m_1 > m_2 \geq 0 \dots \\
	& \phi_3(r) = \infty \mathbb{I}_{[0,\zeta]}(r) \nonumber
\end{align}
where $r$ is the distance between two particles and $\zeta > 0$ is a tuning parameter. These three pairwise potentials are all functions of the pairwise interactions, captured by the distance between points $r$, and can be used to specify $\psi_2$. 

The works of \cite{petralia2012repulsive,xie2020bayesian} and \cite{quinlan2021class} are based on the above pairwise potentials and propose repulsive prior distributions of the form:
\begin{eqnarray*}
 	p_M\left(\bm \theta \mid \zeta\right)  & \propto & \left(\prod_{m=1}^M g(\theta_m)\right)  \left( \prod_{1\leq i < j\leq M} h_s \left(\lVert \theta_i - \theta_j \rVert\right)
 	\right) \\ 
 	& \propto &\exp \left\{\sum_{m=1}^M  \log g(\theta_m)  +\sum_{1\leq i < j\leq M} \phi_s \left(\lVert \theta_i - \theta_j \rVert\right) \right\}
 	\\ 
 &&	s = 1  \quad \mbox{ in \cite{quinlan2021class}}\\
 &&	s = 2  \quad \mbox{ in  \cite{petralia2012repulsive}} \\
 &&	s = 3  \quad \mbox{ in \cite{beraha2022mcmc}} 
\end{eqnarray*}
The repulsive prior is specified by setting $\psi_1 = \log p_{\bm \mu}$ as in \cite{quinlan2021class}, where $p_{\bm \mu}$ is then chosen to be a (often) Normal distribution with mean zero. This yields $\psi_1 = \sum_{1}^M \theta_i^2$ and that the repulsion term $\psi_2$ is a function of $(\theta_i - \theta_j)^2$ in \cite{quinlan2021class} and of the Euclidean distance between pairs of locations in \cite{petralia2012repulsive}, becoming special cases of repulsive interactions with pair potential. Both papers consider extensions. These approaches present serious drawbacks that will be discussed later. 

On the other hand, \cite{xie2020bayesian} use the same form for $\psi_1$, but the repulsive part of the prior, still belonging to the class of GPPs, does not simplify to a function of pairwise interactions. The fact that $\psi_2$ does not simplify as before forces them to make stronger assumptions on $\mathcal{H}$ to ensure the existence of the partition function $Z_M(\zeta)$. In particular, the authors show that the partition function $Z_M(\zeta)$ is bounded as a function of the number of components $M$ when the repulsion term is smaller than 1 and when square integrability of a function of the norm $\parallel \theta_i - \theta_j \parallel_2$ is assumed. In more details, the latter assumption states that $\int_{\mathbb{R}}\int_{\mathbb{R}}\left(\log\left(\hat{g}\left( \parallel \theta_i - \theta_j \parallel_2 \right)\right) \right)^2 d\theta_i d\theta_j < +\infty$, for $\hat{g} : \mathbb{R}^+ \rightarrow \left[ 0, 1 \right]$ a strictly monotonically increasing function s.t. $\hat{g}(0) = 0$. We point out that the first assumption is not necessary for the prior distributions presented in this work. Moreover, the authors use the second condition to prove a theorem similar to the large deviation principle discussed later.

When $\psi_2$ is chosen to be one of the functions in Eq.~\eqref{eq:potentials}, the Janossy density of the resulting GPP factorises. These potentials have an interpretation in statistical mechanics. The first type of potential arises in the study of the behaviour of gases \citep{ruelle1970statistical, prestion1976random}, and is reminiscent of earlier work involving the Morse potential, used to describe interatomic interactions \citep{morse1929diatomic}. The second type of potential is the Lennard-Jones potential \citep{jones1924determination}, which has been investigated in relation to the study of gases (such as argon) whose repulsive and attractive forces follow an inverse power of the distance between the particles. Potential $\phi_3$ is called Strauss potential \citep{strauss1975model}, and it has been firstly introduce to test the hypothesis that a collection of points in space is distributed uniformly. It is constructed by replacing the distance between two points $r$ by an indicator variable describing the proximity of the points in terms of a given radius. This potential is an example of ``hard-core'' potential, due to the abrupt change in repulsive force imposed on the particles, describing the situation in which the particles are hard spheres of radius $\zeta$, with their centres corresponding to the points in the configuration.

In conclusion, in Bayesian mixture models, the specification of a repulsive joint prior distribution often reduces to  multiplying a standard density (usually corresponding to the independence assumption) by a repulsive term, usually a function of the pairwise distances between the locations (see previous section). In principle, this approach allows the specification of a wide range of joint distributions for the parameters of interest, governed by some tuning parameters. Indeed, Gibbs processes are appealing in terms of flexibility and interpretability, but are typically intractable due to the normalising constant $Z_M(\zeta)$, i.e. the partition function. 
To deal with this issue, \cite{ogata1981estimation, ogata1984likelihood} provide an approximation of the joint probability distribution for maximum likelihood estimation, while, under some conditions, bounds on the partition function can be derived \citep{DobrushinShlosman1985, DobrushinShlosman1986}. In the context of finite mixture models (i.e., when $M$ is fixed), as in \cite{petralia2012repulsive} and \cite{quinlan2021class}, the knowledge of $Z_M(\zeta)$ is not needed to perform posterior inference, which is usually based on Metropolis-Hastings algorithms as the number of particles is fixed. When a prior distribution is assumed on the number of components of the mixture $M$, knowledge of the normalising constant $Z_M(\zeta)$ enables simpler computations, as opposed to the approach of \citep{beraha2022mcmc} which requires non-trivial adaptation of birth-and-death Metropolis-Hastings algorithms and the exchange algorithm of \cite{murray2006mcmc}. Alternatively, \cite{xie2020bayesian} construct an \textit{ad-hoc} prior for $M$ which contains the normalising constant $Z_M(\zeta)$ in its specification to remove such issues, since $Z_M(\zeta)$ cancels out when evaluating Metropolis-Hastings acceptance probabilities. Although theoretical results on the relationship between the number of components and $Z_M(\zeta)$ are shown by \cite{xie2020bayesian}, this choice of prior distribution for $M$ does not allow the inclusion of relevant a-priori information into the model, making interpretation and tuning of the hyper-parameters difficult. As already pointed out by \cite{murray2007advances, murray2012bayesian}, the prior choice of \cite{xie2020bayesian} tends to dominate posterior inference. Furthermore, the intractability of $Z_M(\zeta)$ does not allow to perform inference on parameters such as the strength of the repulsion. 
Another drawback of these previous approaches is the fact, already pointed out by \cite{quinlan2021class}, that the joint distribution for the location parameters specified in this way is not \textit{sample-size consistent}, i.e. the $(M-1)$-th dimensional distribution cannot be derived from appropriate marginalisation of the $M$-dimensional distribution, leading to theoretical and computational issues. For instance, it is not possible to specify a mixture model with random number of components in a standard way, and consequently posterior inference across dimensions is unfeasible.

\section{Repulsive priors obtained from random matrices}\label{sec:repulsive_priors}

In this work, we take a different approach and specify tractable joint distributions, still within the class of GPPs, 
exploiting  results from random matrix theory. Such distributions still present interesting connections with the mathematical theory of gases. Specifically, we consider the joint Gibbs canonical distributions used to model \textit{Coulomb gases}, also called log-gases \citep{dyson1962statistical, forrester2010log, mehta1963statistical, de1995statistical}. These distributions are obtained starting from a potential $\mathcal{H}$ with logarithmic pairwise interactions, following the same approach used to define the Boltzmann distribution in \eqref{eq:Gibbs_Density_gen}. 
Coulomb gases \citep{dyson1962statistical} provide an example of analytically tractable systems, with particles described as infinitely long parallel charged lines.
In the case of \textit{one component} Coulomb system, all $M$ particles are of like charge, $\sqrt{\zeta}$, say. For such gases, the interaction between the particles is described by a logarithmic function in the expression of the Hamiltonian:
\begin{equation}\label{eq:Coulomb_potential}
	\phi_{\zeta}(r_{ij}) = - \zeta \log \left( r_{ij} \right)
\end{equation}
where $r_{ij}$ is the distance between particles $i$ and $j$. The parameter $\zeta > 0$ characterises the strength of the interaction between particles. In this case, a general expression for the potential $\mathcal{H}$ for $M$ particles and the resulting Janossy density are:
\begin{align}\label{eq:pairwisepotential_Coulomb}
	&\mathcal{H}(\theta_1,\ldots, \theta_M \mid \zeta) = \sum_{i=1}^M \psi_1(\theta_i) + \zeta \sum_{i < j} \log(\lVert \theta_i - \theta_j \rVert) \\
	&p_M(\theta_1,\ldots,\theta_M \mid \zeta) = Z^{-1}_M(\zeta) \exp\left\{-\mathcal{H}(\theta_1,\ldots, \theta_M \mid \zeta) \right\} \nonumber 
\end{align}
Different expressions of the function $\psi_1$ lead to different joint distributions for the location parameters $\bm \theta$. Under specific choices (discussed later), these present tractable normalising constants (i.e., partition functions), obtained by normalising \eqref{eq:pairwisepotential_Coulomb}, making them an appealing class of distributions in statistical inference. In particular, for suitable choices of $\psi_1$, the joint distributions obtained by normalising \eqref{eq:pairwisepotential_Coulomb} coincide with those of the eigenvalues of Gaussian, Wishart and Beta random matrices \citep{mehta2004random, forrester2010log}. The link between Coulomb gases and random matrix theory, referred to as the Coulomb gas analogy, is widely recognised in the field of statistical mechanics.  

The Gibbs canonical distributions of the Coulomb gases also present a link with DPPs. In particular, for different choices of $\zeta$, $p_M$ in \eqref{eq:pairwisepotential_Coulomb} can be expressed as the determinant of specific random matrices \citep{mehta2004random}. These results are, in general, very technical. A tractable example is obtained for $\zeta = 2$, where a symmetric positive definite matrix can be constructed via a basis of orthogonal polynomials, whose joint law of eigenvalues coincides with its determinant. Suitable choices of the orthogonal polynomial basis yield the eigenvalue distributions for the Gaussian (Hermite polynomials), Wishart (Laguerre polynomials) or Beta (Jacobi polynomials) random matrices \citep{forrester2010log}, introduced in the following sections. The three families of distributions generated in this way inherit the names of Hermite, Laguerre and Jacobi ensembles, respectively. 
Moreover, in one and two dimensions, at inverse temperature $\zeta = 2$, the Coulomb system with
logarithmic interactions (a.k.a. Dyson log gas in 1D) is known to be a determinantal point process, meaning that its correlation functions are given by certain determinants \citep{ghosh2018point}. We refer the reader to the work by \cite{forrester2010log} for an extensive discussion on the topic.

Here, we introduce eigenvalue distributions of random matrices whose probability density function is proportional to:
\begin{equation}\label{eq:EigenDistr_General}
	\prod_{l=1}^M \eta\left(\theta_l \mid \zeta\right) \prod_{1\leq i < j\leq M}\mid \theta_i - \theta_j \mid^{\zeta}, \quad \zeta > 0
\end{equation}
where $\eta\left(\theta \mid \zeta\right)$ is a weight function characterising the resulting probability law. Some common weight functions used in random matrix theory are:
\begin{equation}\label{eq:weight_functions}
\eta(\theta \mid \zeta) = 
\left\{\begin{array}{l @{\quad} l r l}
	e^{- \frac{\zeta}{2} \theta^2} & \theta \in \mathbb{R} & \mbox{Hermite} \\
	\theta^{\frac{\alpha \zeta}{2}} e^{- \frac{\zeta}{2} \theta} & \theta > 0, \alpha \zeta / 2 > -1 & \mbox{Laguerre} \\
	\theta^{\alpha} (1 - \theta)^{\beta} & \theta \in (0,1), \alpha, \beta > -1 & \mbox{Jacobi} \\
\end{array}\right.
\end{equation}
The three types of weight function refer to the properties of the random matrices for which Eq.~\eqref{eq:EigenDistr_General} is the eigenvalue distribution and have an interpretation in statistical mechanics. In particular, as we will describe in the next Sections, each weight function refers to specific transformations of random matrices. In agreement with the theory describing the joint law of systems of particles, these sets of random matrices are referred to as \textit{ensembles}.  The parameter $\zeta$ plays a crucial role in the derivation of the eigenvalue distributions. In particular, each subset of random matrices corresponds to a specific value of $\zeta$. 

\subsection{Eigenvalue distributions derived from Gaussian ensembles}

In this Section, we describe the eigenvalue distribution for the set of random matrices belonging to the Gaussian ensemble. Let $M > 0$ and $\bm Z$ be a $M \times M$ matrix with i.i.d. entries $Z_{ij} \sim \mathcal{N}(0,1)$. Define the $M \times M$ matrix $\bm X = \left(\bm Z  + \bm Z'\right)/\sqrt{2}$. The resulting symmetric matrix $\bm X$ is called a real Wigner matrix and the joint law of its eigenvalues, denoted by $\bm \theta = (\theta_1, \dots, \theta_M)$, is given by:
\begin{equation}\label{eq:Eigen_GOE}
	p(\bm \theta \mid  M) = \mathcal{G}^{-1}_{M} \prod_{l = 1}^M e^{-\frac{\theta_l^2}{2}} \prod_{i < j} \mid \theta_i - \theta_j \mid
\end{equation}
where $\mathcal{G}_{M}$ is the normalising constant. The family of eigenvalue distributions originated by this random matrix construction is known as the Gaussian orthogonal ensemble. When the matrix $\bm Z$ has complex or quaternion Gaussian entries, the distribution of the eigenvalues of $\bm X$ changes. In general, the joint eigenvalue distribution for the Gaussian ensemble has the following expression:
\begin{equation}\label{eq:Eigen_Gaussian}
	p(\bm \theta \mid M, \zeta) = \mathcal{G}^{-1}_{M, \zeta} \prod_{l = 1}^M e^{- \frac{\zeta}{2} \theta_l^2} \prod_{i < j} \mid \theta_i - \theta_j \mid^{\zeta}
\end{equation}
where, when $\zeta = 1$, we recover the Gaussian orthogonal ensemble. Values of $\zeta$ equal to 2 and 4 correspond to $\bm Z$ with complex (Gaussian unitary ensemble) and quaternion (Gaussian symplectic ensemble) entries, respectively. In general, for $\zeta > 0$, Eq.\eqref{eq:Eigen_Gaussian} gives the eigenvalue distribution of $\bm X$ in the case of tri-diagonal random matrices with standard normal diagonal entries and $\chi$-squared off-diagonal entries with degrees of freedom equal to $(M-1)\zeta, (M-2)\zeta, \dots, \zeta$ \citep{forrester2010log}. 
Such distribution can be used to specify a joint prior distribution for the location parameters of a mixture model. 
The normalising constant has a closed form expression \citep{mehta1963statistical, mehta2004random}:
$$
\mathcal{G}_{M, \zeta} = \zeta^{-\frac{M}{2} - \zeta M(M-1)/4}\left( 2 \pi \right)^{\frac{M}{2}}\prod_{l = 0}^{M-1}\frac{\Gamma\left(1 + \left(j+1\right)\frac{\zeta}{2}\right)}{\Gamma\left(1 + \frac{\zeta}{2}\right)}
$$
which allows for more efficient computations. Note that, when $M = 1$, we recover the univariate Gaussian distribution with precision parameter $\zeta$.

\textit{Remark:} in Eq.~\eqref{eq:Eigen_Gaussian}, when $\zeta = 2$, the repulsion is a function of the Euclidean distance as proposed by \cite{quinlan2021class}. The difference is that the distance in \cite{quinlan2021class} appears in the exponential function, with the repulsion term penalising more heavily close locations that our penalty term. Still, the normalising constant is not available analytically, adding complexity to computations.

\subsection{Eigenvalue distributions derived from Laguerre ensembles}

Throughout, we assume $N > M$. Let $\bm Z$ be a $N \times M$ matrix with i.i.d. entries $Z_{ij} \sim \mathcal{N}(0,1)$. Let $\bm \Sigma$ be a deterministic positive-definite $M \times M$ matrix and let $\sqrt{\bm \Sigma}$ denote its unique positive-definite square root, such that $\bm \Sigma = \sqrt{\bm \Sigma}'\sqrt{\bm \Sigma}$. Define the $N \times M$ matrix $\bm X = \bm Z \sqrt{\bm \Sigma}'$. Then, the $M \times M$ matrix $\bm S = \bm X' \bm X$  has a Wishart distribution with $N$ degrees of freedom and scale matrix $\bm \Sigma$ and we write $\bm S \sim W_M\left(N, \bm \Sigma\right)$, with the following p.d.f.:
\begin{equation}\label{eq:Wishart_pdf}
\begin{gathered}
	p\left(\bm S \mid N, \bm \Sigma\right) = \\
 \frac{\mid \bm \Sigma \mid^{-\frac{N}{2}}}{2^{M\frac{N}{2}} \Gamma_M \left( \frac{N}{2} \right) } \mid \bm S \mid^{\frac{N - M - 1}{2}} \exp\left\{ \frac{tr \left(\bm \Sigma^{-1} \bm S \right)}{2} \right\} \mathbbm{1}_{\Omega^{\bm S}_M}(\bm S)
 \end{gathered}
\end{equation}
where $\Gamma_p(x)$ is the multidimensional gamma function of dimension $M$ and argument $x>0$ \citep{gupta2018matrix}, $tr(\cdot)$ and $\mid \cdot \mid$ indicate the trace and the determinant of a square matrix, respectively, while $\Omega^{\bm S}_M$ represents the cone of positive definite square matrices of dimension $M$. The joint law of the eigenvalues of $\bm S$, denoted by $\bm \theta = (\theta_1, \dots, \theta_M)$, is given by:
\begin{equation}\label{eq:Eigen_LOE}
\begin{gathered}
p(\bm \theta \mid N, M, \bm \Sigma) = \\
\mathcal{W}^{-1}_{N,M} |\bm \Sigma|^{-\frac{N}{2}} {}_0F_0\left( -\frac{1}{2}\bm \Sigma^{-1}, \bm S \right)\prod_{l = 1}^M \theta_l^{N-M-1/2} \prod_{i < j} \mid \theta_i - \theta_j \mid
\end{gathered}
\end{equation}
where ${}_0F_0(\cdot,\cdot)$ is the hypergeometric function of two matrix arguments \citep{james1964distributions} and $\mathcal{W}_{N,M}$ is the normalising constant. Note that $\theta_j>0$ since they are eigenvalues of a positive definite matrix. The set of random matrices for which Eq.~\eqref{eq:Eigen_LOE} is the joint eigenvalue distribution is known as the Laguerre orthogonal ensemble.

When $\bm \Sigma = \mathbb{I}_M$, i.e. the identity matrix of dimension $M$, the eigenvalue distribution reduces to:
\begin{equation}\label{eq:Eig_ID_Wishart}
p(\bm \theta \mid N, M) = \mathcal{W}^{-1}_{N,M} e^{-\frac{1}{2}\sum_{l = 1}^M \theta_l}  \prod_{l = 1}^M \theta_l^{N-M-1/2} \prod_{1\leq i < j\leq M} \mid \theta_i - \theta_j \mid
\end{equation}
In this case, the normalising constant can be computed exactly \citep{fisher1939sampling, hsu1939distribution, roy1939p}. When the matrix $\bm X$ has complex or quaternion entries, we can introduce the parameters $\zeta > 0$, such that $\bm \Sigma = \zeta \mathbb{I}_M$, and $\alpha = N - M + 1 - 2/\zeta$ and define the family of distributions referred to as the Laguerre $\zeta$-ensemble \citep{forrester2010log}:
\begin{equation}\label{eq:Eigen_Wishart}
	p(\bm \theta \mid N, M, \alpha, \zeta) = \mathcal{W}^{-1}_{\alpha, \zeta, M} e^{-\frac{\zeta}{2}\sum_{l = 1}^M \theta_l}\prod_{l = 1}^M \theta_l^{\alpha \frac{\zeta}{2}} \prod_{i < j} \mid \theta_i - \theta_j \mid^{\zeta}
\end{equation}
For $\zeta = 1$, this is exactly the joint law in Eq.~\eqref{eq:Eig_ID_Wishart}, also called the Laguerre orthogonal ensemble. The cases $\zeta = 2$ and $\zeta = 4$ correspond to the analogous distributions when the matrix $\bm X$ has complex (Laguerre unitary ensemble) and quaternion (Laguerre symplectic ensemble) entries, respectively \citep{edelman2005rmt}. For arbitrary $\zeta \not\in\{1,2,4\}$, the distribution \eqref{eq:Eigen_Wishart} may be viewed as the case when $\bm Z$ contains $\chi$-squared distributed entries \citep{forrester2010log}. 
Notice, when $M = 1$, we obtain the $\text{Gamma}\left(\alpha \zeta/2 + 1, \zeta/2\right)$ distribution with mean $\alpha + 2/\zeta$ and variance $\left(\alpha + 2/\zeta\right)2/\zeta$.

Note that the above law is normalisable for any pair $\zeta > 0$ and $\alpha > - 2/\zeta$, with normalising constant known in closed form \citep{forrester2010log}:
\begin{equation}\label{eq:NormConst_Wishart}
\begin{gathered}
 \mathcal{W}_{\alpha, \zeta, M} = \\
 \left( \zeta/2 \right)^{-M \left( 1 + \frac{\zeta}{2}\left(\alpha + M-1\right) \right)} \prod_{j = 0}^{M-1}\frac{\Gamma\left(1 + \frac{\zeta}{2}\left(j + 1\right)\right)\Gamma\left(1 + \alpha\frac{\zeta}{2} + j\frac{\zeta}{2}\right)}{\Gamma\left(1 + \frac{\zeta}{2}\right)}
 \end{gathered}
\end{equation}
The joint distribution is otherwise not well-defined.

\subsection{Eigenvalue distributions derived from Jacobi ensembles}

The Jacobi ensemble is related to the eigenvalues of the Beta random matrix \citep{forrester2010log} and has applications in physics \citep{livan2011moments, vivo2008transmission}. Following \cite{mitra1970}, let $\bm S_1 \sim W_M\left(N_1, \bm \Sigma\right), \bm S_2 \sim W_M\left(N_2, \bm \Sigma\right)$, with $\bm S_1$ and $\bm S_2$ independent. Let us define $\bm S_{12} = \bm S_1 + \bm S_2$, so that $\bm S_{12} \sim W_M\left(N_1+N_2, \bm \Sigma\right)$. We take the Cholesky factorisation $\bm S_{12} = \bm C' \bm C$, where $\bm C'$ is lower triangular with non-negative diagonal elements. Note that $\bm S_{12}$, and hence $\bm C'$, is invertible with probability 1. Letting $\bm L' = \left(\bm C'\right)^{-1}$, we define:
\begin{equation} \label{eq:Betamatrix}
	\bm U = \bm L' \bm S_1 \bm L,
\end{equation}
and we say that $\bm U$ follows the matrix-variate Beta distribution with parameters $\frac{N_1}{2},\frac{N_2}{2}$, with $N_1, N_2 > (M - 1)/2$. The latter condition is derived from the existence of the Wishart matrix $\bm S_{12}$. The p.d.f. of the matrix-variate Beta distribution is the following:
\begin{equation}\label{eq:BetaMatrix_pdf}
\begin{gathered}
	p\left(\bm U \mid N_1, N_2\right) = \\
 \frac{\Gamma_M \left( \frac{N_1 + N_2}{2} \right) }{\Gamma_M \left( \frac{N_1}{2} \right) \Gamma_M \left( \frac{N_2}{2} \right)} \mid \bm U \mid^{\frac{N_1 - M - 1}{2}} \mid \bm I_M - \bm U \mid^{\frac{N_2 - M - 1}{2}} \mathbbm{1}_{\Omega^{\bm U}_M}(\bm U)
 \end{gathered}
\end{equation}
where $\Omega^{\bm U}_M$ is the space of $M \times M$ symmetric matrices $U$ with the property that $\bm U$ and $\bm I_M - \bm U$ are both positive definite. Notice that the p.d.f. of $\bm U$ does not depend on the scale matrix $\bm \Sigma$. Consider the joint distribution of the eigenvalues of $\bm U$, $\bm \theta = (\theta_1, \dots, \theta_M)$, given by:
\begin{equation}\label{eq:Eigen_JOE}
	p(\bm \theta | \alpha, \beta, M) = \mathcal{B}^{-1}_{\alpha, \beta, M} \prod_{l = 1}^M \left(\theta_l\right)^{\alpha - 1}(1 - \theta_l)^{\beta-1} \prod_{i < j} |\theta_i - \theta_j|
\end{equation}
where $\alpha = (N_1 - M + 1)/2$ and $\beta = (N_2 - M + 1)/2$, highlighting that fact that the hyperparameters of this distribution also depend on the number of components $M$. As before, this can be viewed as a specific instance of a more general family of distributions indexed by the parameter $\zeta > 0$, leading to the $\zeta$-Jacobi ensemble: 
\begin{equation}\label{eq:Eigen_Beta}
	p(\bm \theta \mid \alpha, \beta, \zeta, M) = \mathcal{B}^{-1}_{\alpha, \beta, \zeta, M} \prod_{l = 1}^M \left(\theta_l\right)^{\alpha - 1}(1 - \theta_l)^{\beta - 1} \prod_{i < j} |\theta_i - \theta_j|^{2 \zeta}
\end{equation}
The normalising constant $\mathcal{B}_{\alpha, \beta, \zeta, M}$ is also known as the Selberg integral \citep{selberg1944berkninger}:
\begin{equation}\label{eq:NormConst_Beta}
	\mathcal{B}_{\alpha, \beta, \zeta, M} = \prod_{j = 0}^{M-1} \frac{\Gamma(\alpha + j\zeta) \Gamma(\beta + j\zeta) \Gamma(1 + (j+1)\zeta)}{\Gamma(\alpha + \beta + (M + j - 1)\zeta)\Gamma(1 + \zeta)}
\end{equation}
The conditions for the existence of the distribution are $\alpha, \beta > -1$ and $\zeta \geq 0$. 

\cite{pham2009multivariate} discuss several properties of the above distribution, alternatively called the multivariate Selberg Beta distribution, and present an example of marginal laws obtained in the bivariate case. An interesting result is that, for fixed values of $\alpha$, $\beta$, $\zeta$ and $M$, thanks to the symmetry of expression \eqref{eq:Eigen_Beta}, the univariate marginal distributions are all of the same type, although difficult to compute analytically. The authors also show that a generalisation of the Dirichlet distribution can be achieved, referred to as the Selberg Dirichlet distribution, by normalising a vector of random variables whose joint law is given by Eq. \eqref{eq:Eigen_Beta}. This joint distribution presents the same tractable features as the multivariate Selberg distributions (e.g., the joint distribution of the eigenvalues of the Beta random matrix), i.e. its normalising constant is known in closed form and presents the desired repulsive property. Furthermore, when $M = 1$, the distribution coincides with the Beta distribution with parameters $\alpha/2$ and $\beta/2$.

As we can observe from Eq. \eqref{eq:Eigen_Gaussian}, \eqref{eq:Eigen_Wishart} and \eqref{eq:Eigen_Beta}, the joint distributions of the eigenvalues of random matrices belonging to the Gaussian, Laguerre or Jacobi ensembles include a repulsive term depending on the pairwise differences between the eigenvalues, and on the choice of the parameters $(\alpha, \beta, \zeta)$. When these distributions are used as centring measure $P_0$ for the location parameters in a mixture model, the choice of the parameters $(\alpha, \beta, \zeta)$ is crucial in defining the amount of repulsion induced, and influence the effect on the resulting clustering structure. As the repulsion term appears in the same functional form, these prior distributions will lead to similar posterior inference and $\zeta$ plays the role of a calibration parameter.


\begin{table}[!ht]
	\centering
	\caption[]{Moments of eigenvalues of random matrices.}
	\label{tab:moments}
\resizebox{\linewidth}{!}{%
	\begin{tabular}{c|ccc}
		Matrix & Gaussian & Wishart & Beta \\
		Ensemble & Hermite & Laguerre & Jacobi \\\hline
		
		$\mathbb{E}\left( \theta_l \right)$ & $ - $ & $1 + (\alpha + M - 1)\zeta/2$ & $\frac{\alpha + (M-1)\zeta}{\alpha + \beta + 2(M-1)\zeta}$ \\
		
		$\mathbb{E}\left(\theta^2_l\right)$ & $\left(1 + \zeta / 2 (M-1) \right)/\zeta$ & $2 + (\alpha + 2M - 1)\zeta/2$ & $\frac{\left(\alpha + 1 + 2(M-1)\zeta\right)\mathbb{E}\left(\theta_l\right) - (M-1)\zeta \mathbb{E}\left(\theta_i \theta_j\right)}{\alpha + \beta + 1 + 2(M-1)\zeta}$ \\          
			

  	$\mathbb{E}\left(\prod\limits_{l=1}^k\theta_l\right)$ & $\left(-\frac{1}2\right)^{\frac{k}{2}}\frac{k!}{2^{\frac{k}{2}
   }\left(\frac{k}{2}\right)!}$, $k \mbox{ even}$ & $\prod\limits_{l=1}^k \left(1 + (\alpha + M - l)\zeta/2\right)$ & $\prod\limits_{l=1}^k \frac{\alpha + (M-l)\zeta}{\alpha + \beta + (2M-l-1)\zeta}$ \\


	\end{tabular}
 }
\end{table}

Contour plots of the joint distribution of the eigenvalues of two-dimensional matrices in the Gaussian ensemble are reported in Figure \ref{fig:Gaussian_pdf_prior}, for increasing values of $\zeta$. This parameter influences the repulsive structure of the joint distribution, as well as its variability.
The joint distribution of the eigenvalues of the two-dimensional Wishart matrix are reported in Figure \ref{fig:Wishart_pdf_prior}, for different choices of the parameters $(\alpha, \zeta)$. As expected, increasing the value of $\zeta$ increases the amount of repulsion imposed on the eigenvalues, producing two diverging modes. On the other hand, higher values of the parameter $\alpha$ increase the peakiness of the modes, reducing the variability.
Similar plots in the case of the matrix-variate Beta distribution are shown in Figure \ref{fig:Beta_pdf_prior}. The panels show the different roles of the parameters of this distribution. The parameter $\zeta$ controls the amount of repulsion, while the parameters $\alpha$ and $\beta$ act as shape parameters, similarly to a classical Beta or Dirichlet distribution, allowing for asymmetry. In Table~\ref{tab:moments} we report analytical moments of the eigenvalues for the three different ensembles. Results about the moments of the eigenvalue distributions can be found in \cite{ullah1986, mehta2004random, iguri2009selberg, pham2009multivariate}.

\begin{figure}[ht]
	\centering
	\subfloat[$\zeta = 0.001$]{\includegraphics[width=0.35\linewidth]{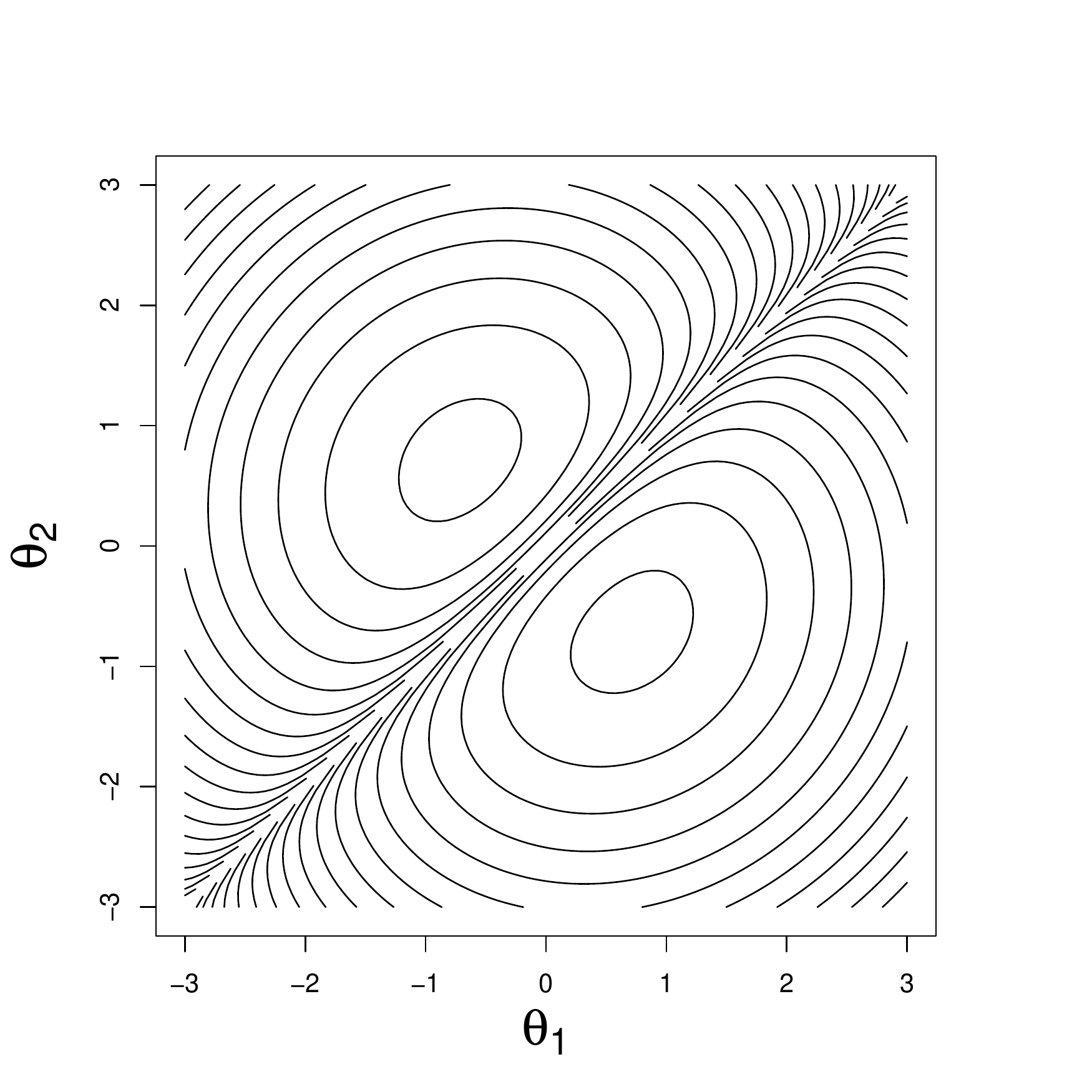}}
	\subfloat[$\zeta = 0.01$]{\includegraphics[width=0.35\linewidth]{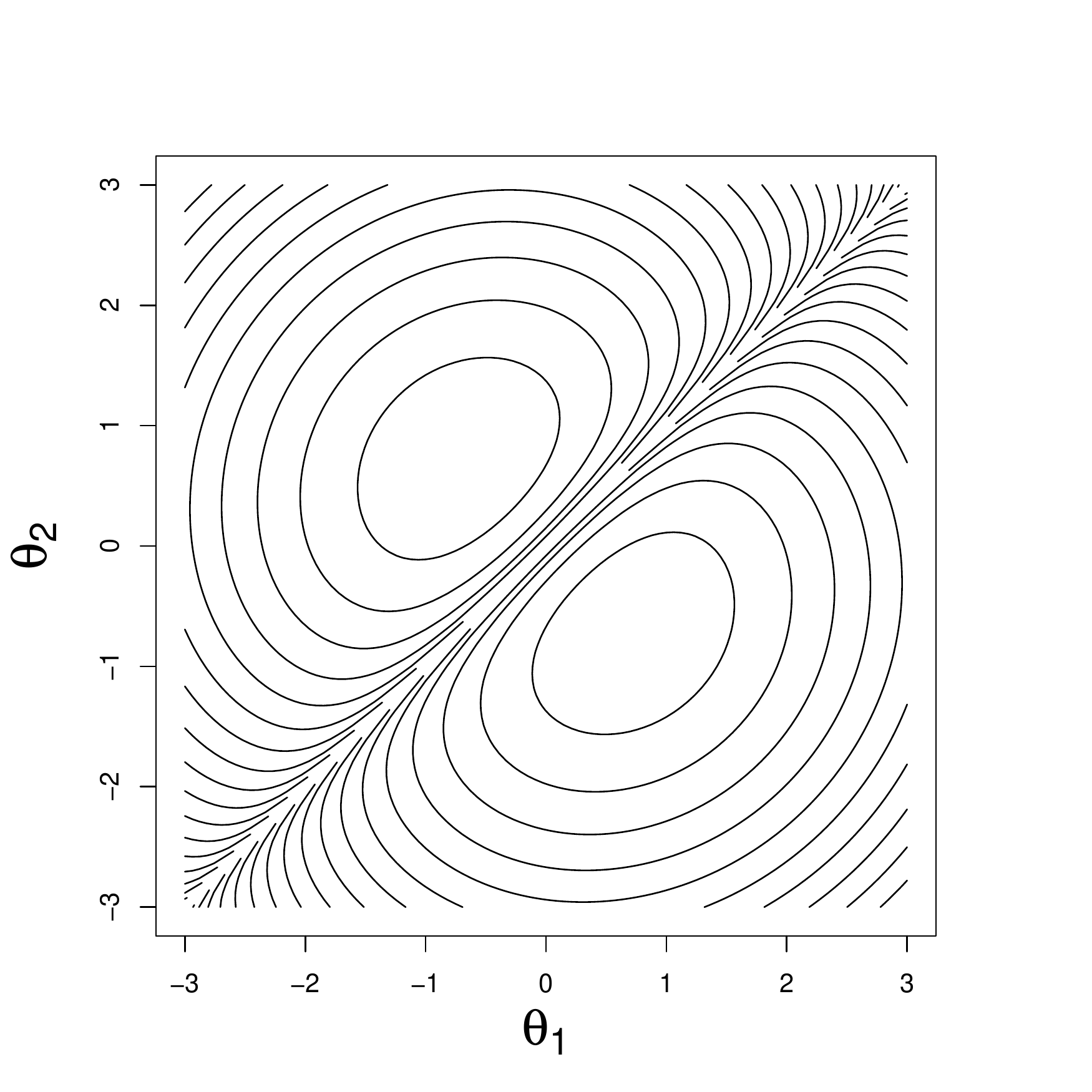}}
	\subfloat[$\zeta = 0.1$]{\includegraphics[width=0.35\linewidth]{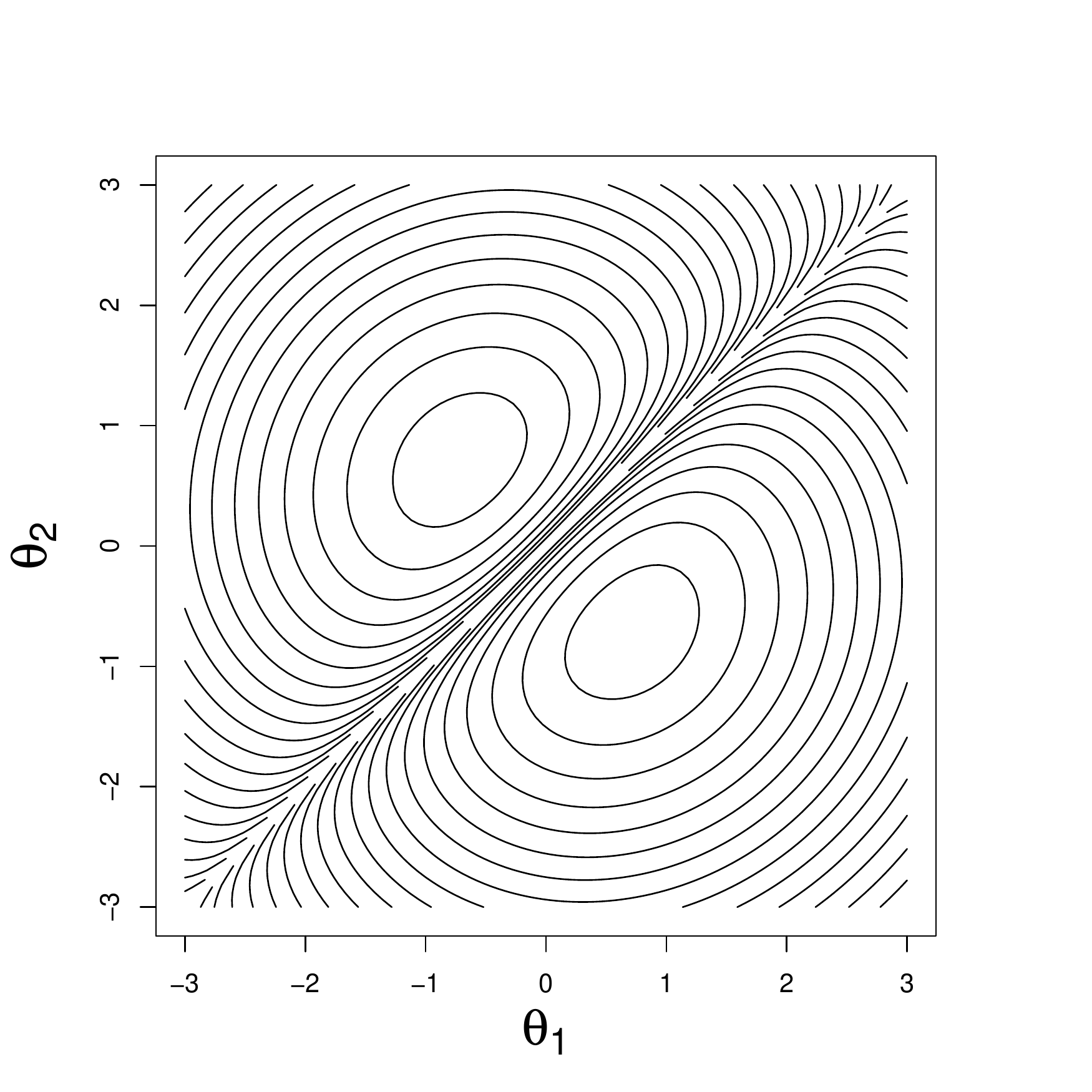}}
	
	\centering
	\subfloat[$\zeta = 0.5$]{\includegraphics[width=0.35\linewidth]{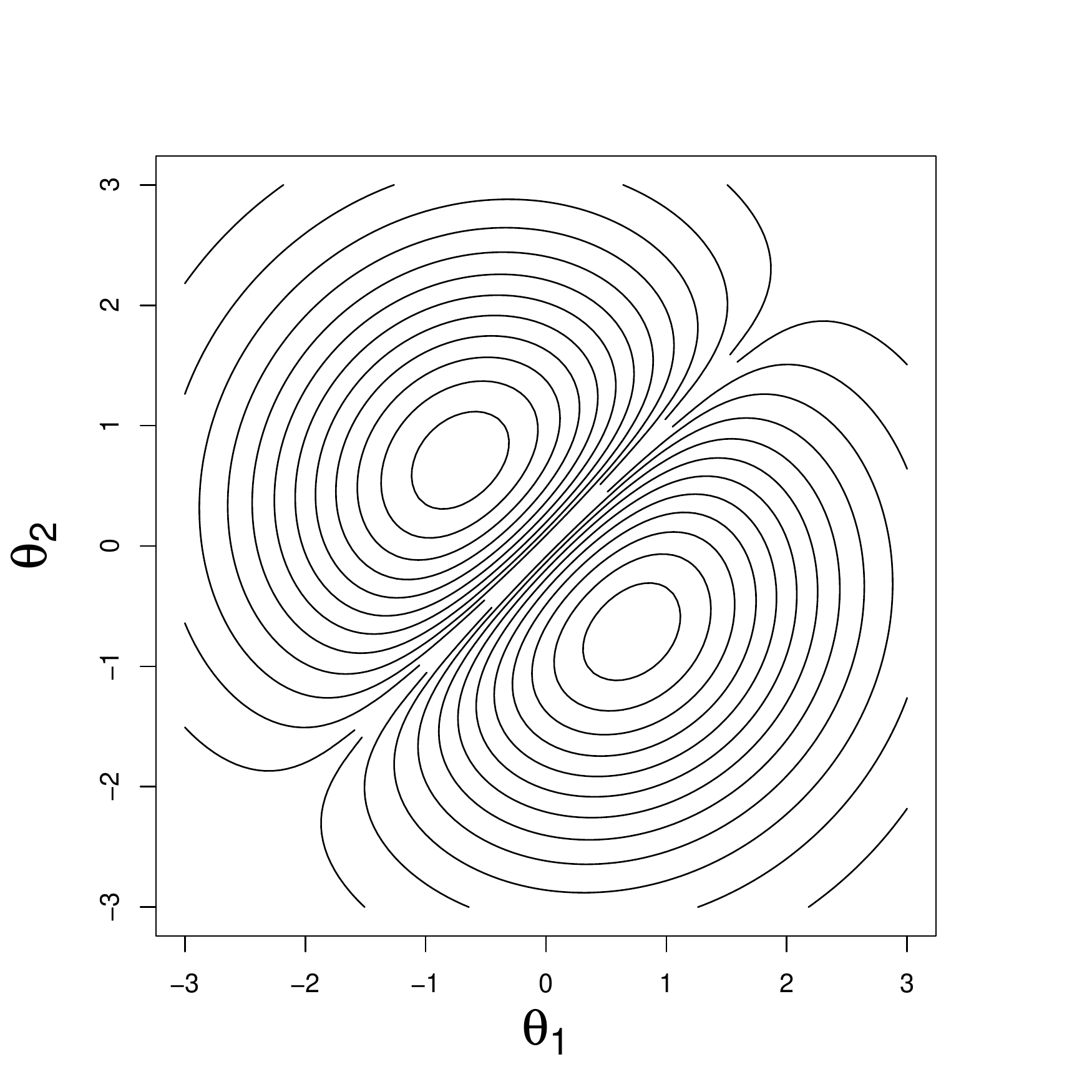}}
	\subfloat[$\zeta = 1$]{\includegraphics[width=0.35\linewidth]{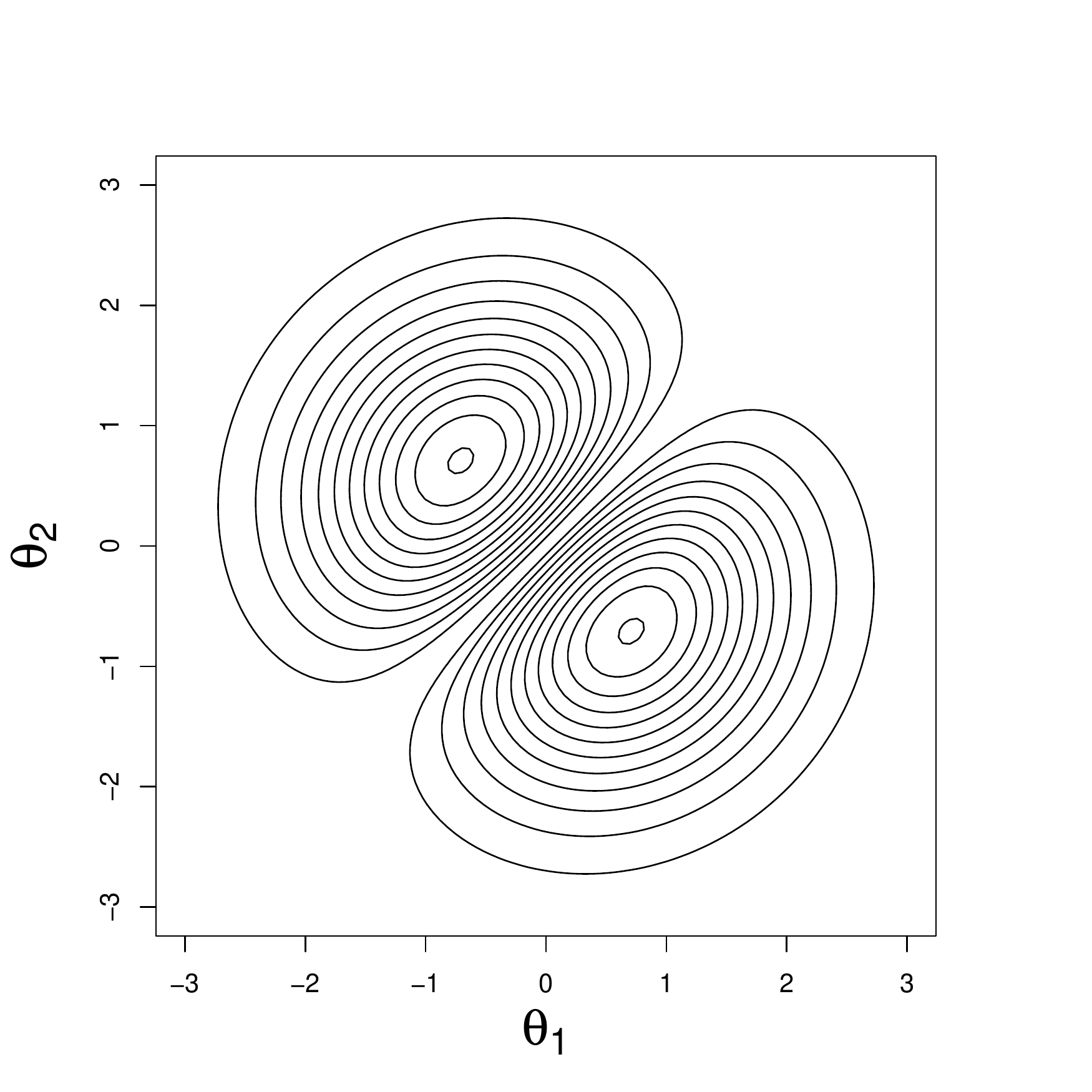}}
	\subfloat[$\zeta = 2.5$]{\includegraphics[width=0.35\linewidth]{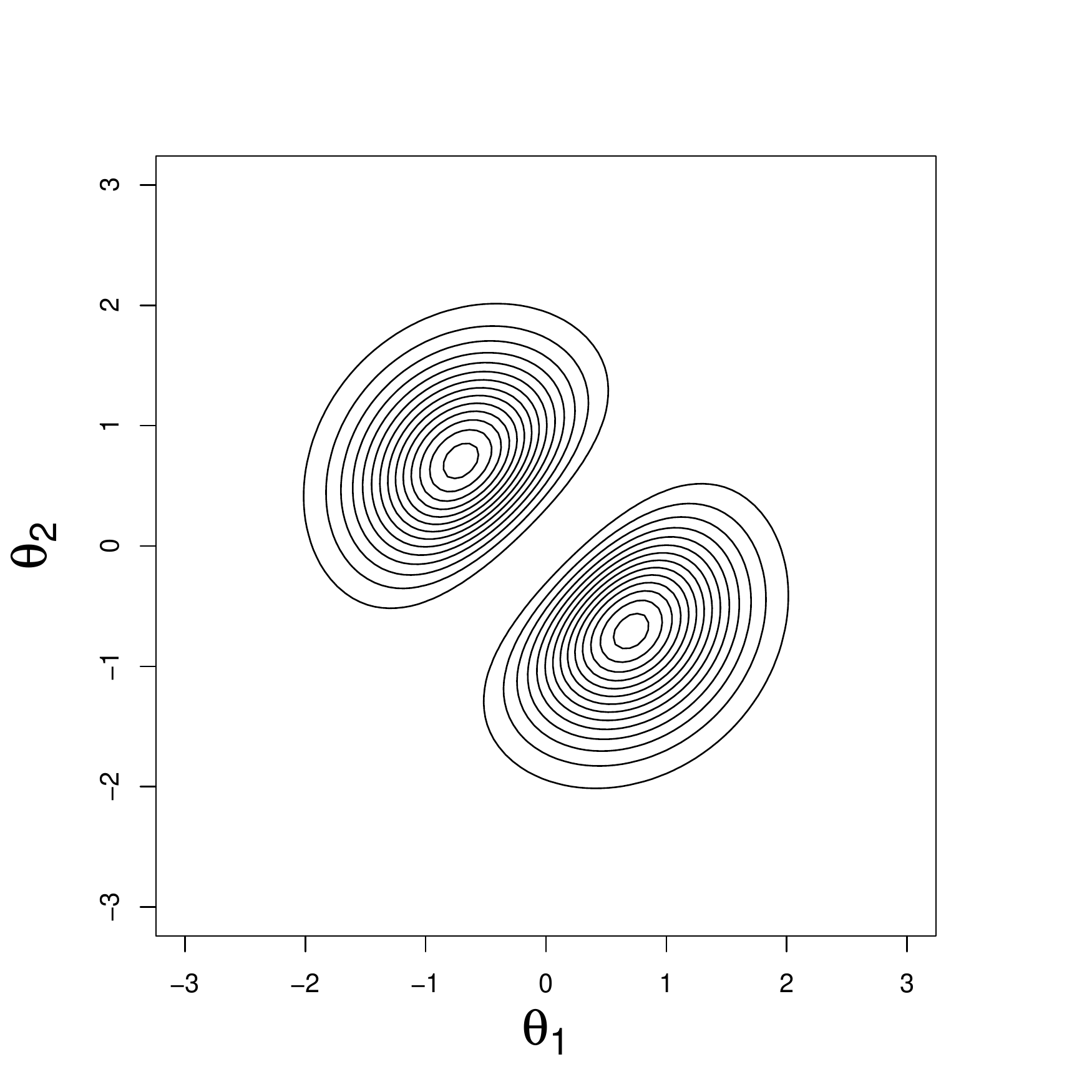}}
	
	\centering
	\subfloat[$\zeta = 5$]{\includegraphics[width=0.35\linewidth]{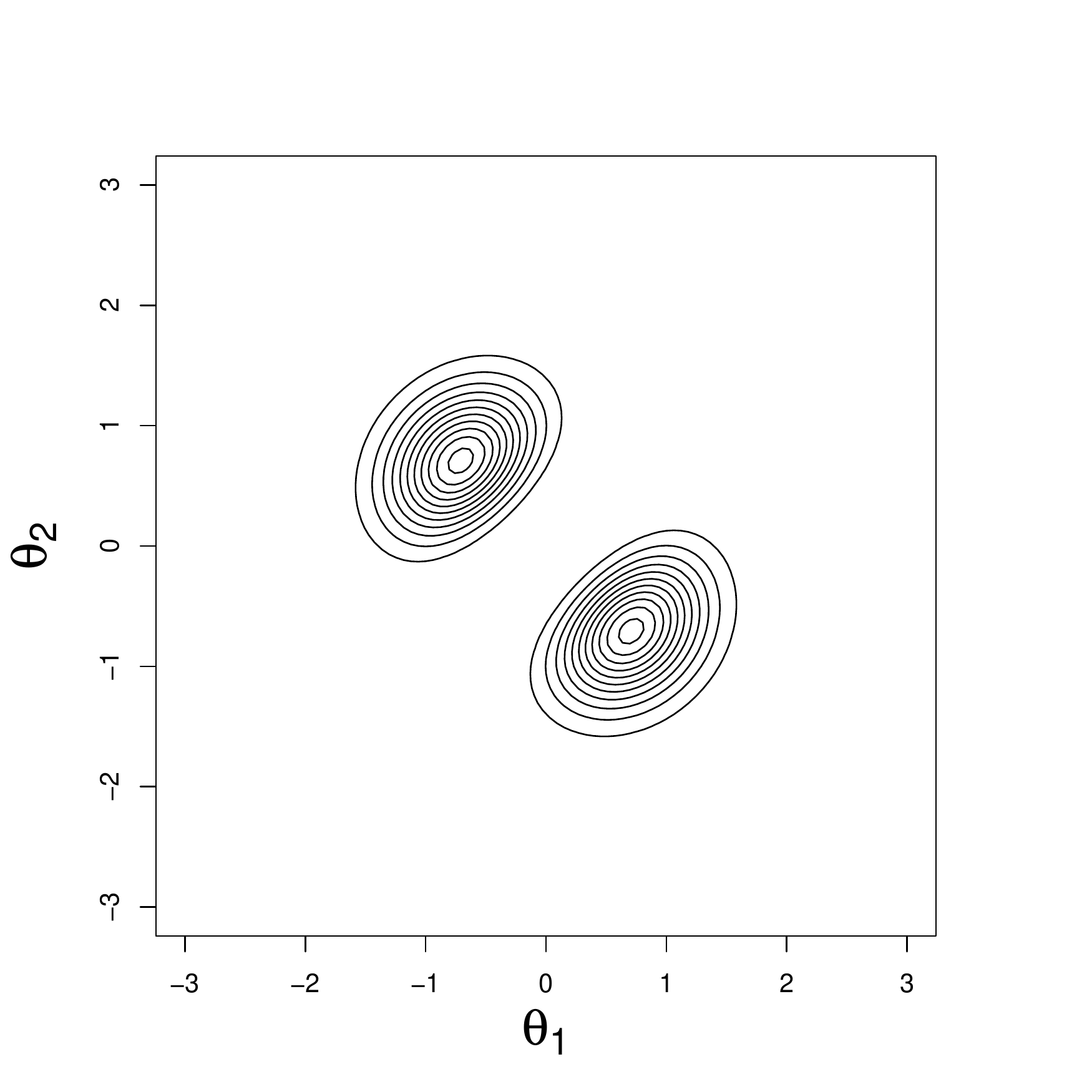}}
	\subfloat[$\zeta = 10$]{\includegraphics[width=0.35\linewidth]{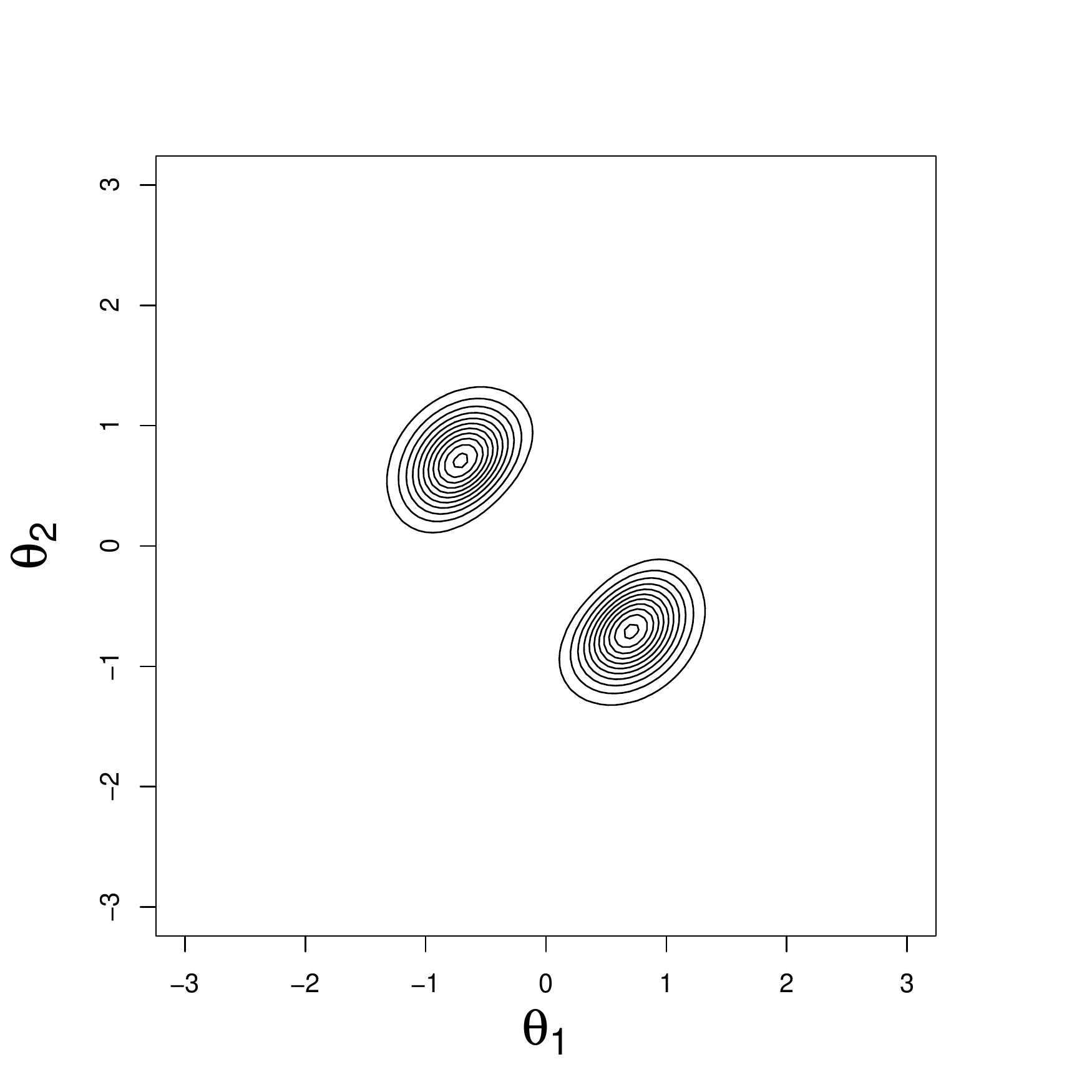}}
	\subfloat[$\zeta = 100$]{\includegraphics[width=0.35\linewidth]{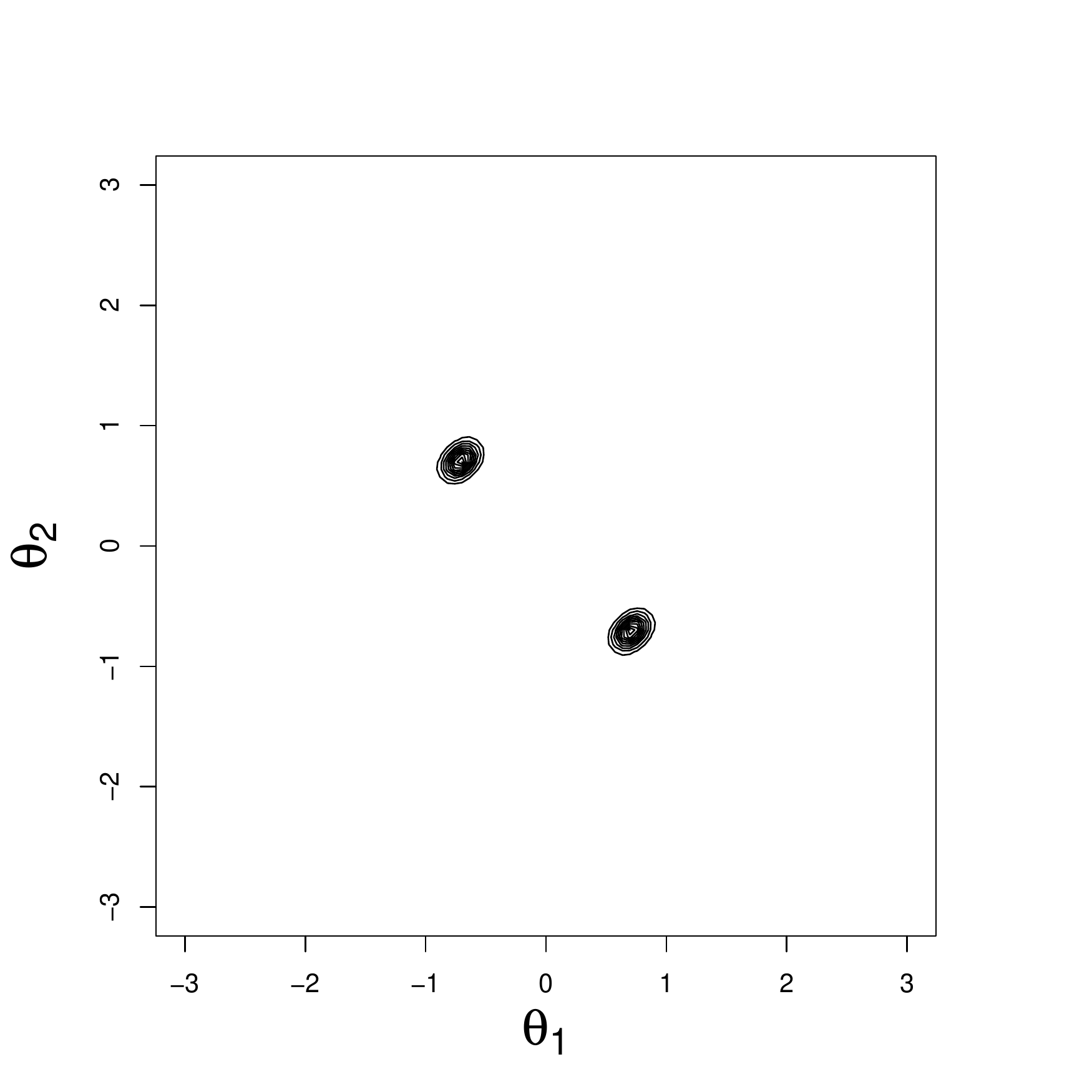}}
	\caption{Gaussian ensemble. Joint distribution of the eigenvalues of a 2-dimensional Gaussian matrix for different choices of the parameter $\zeta$.}
	\label{fig:Gaussian_pdf_prior}
\end{figure}

\begin{figure}[ht]
	\centering
	\subfloat[$(\alpha, \zeta) = (1,2)$]{\includegraphics[width=0.35\linewidth]{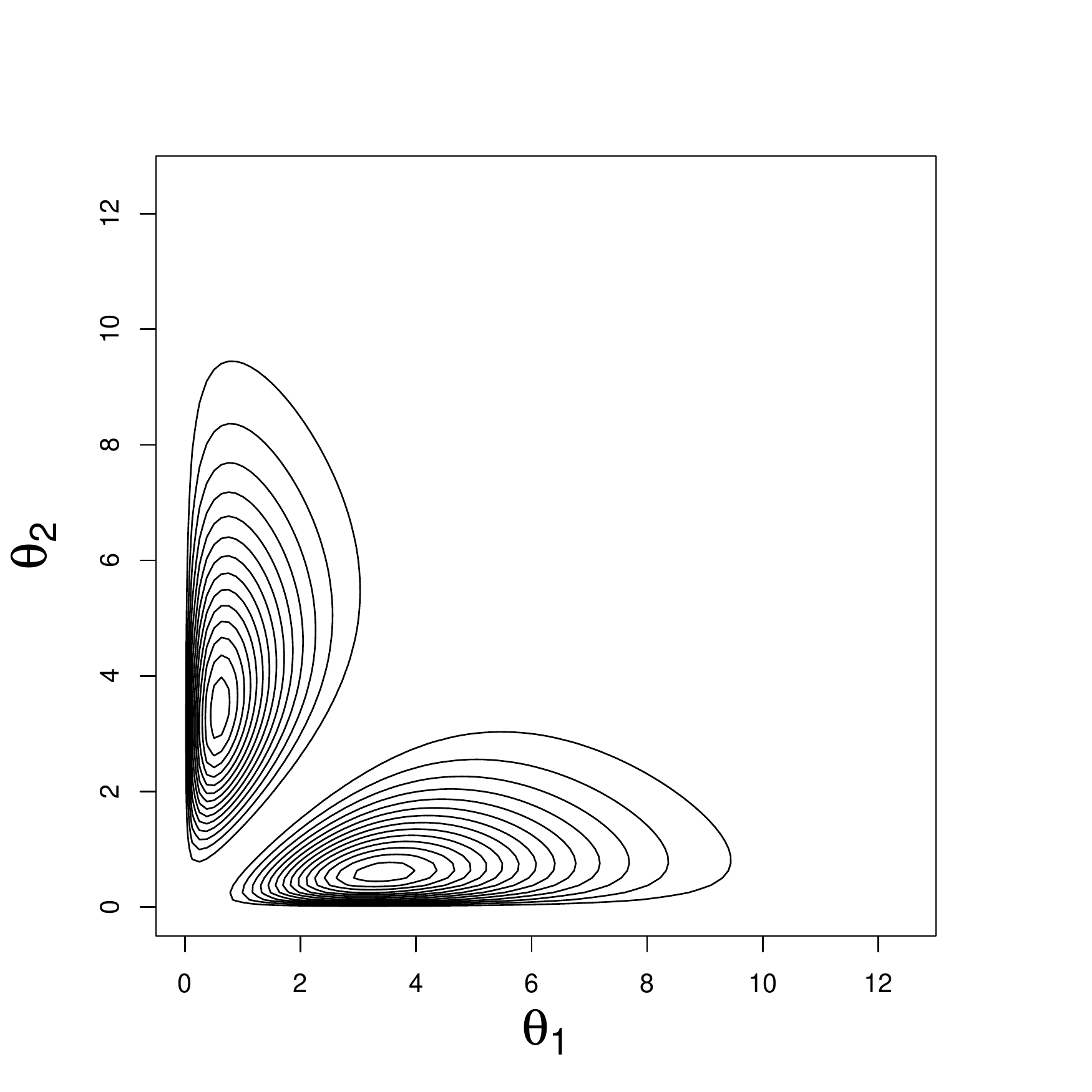}}
	\subfloat[$(\alpha, \zeta) = (1,2.01)$]{\includegraphics[width=0.35\linewidth]{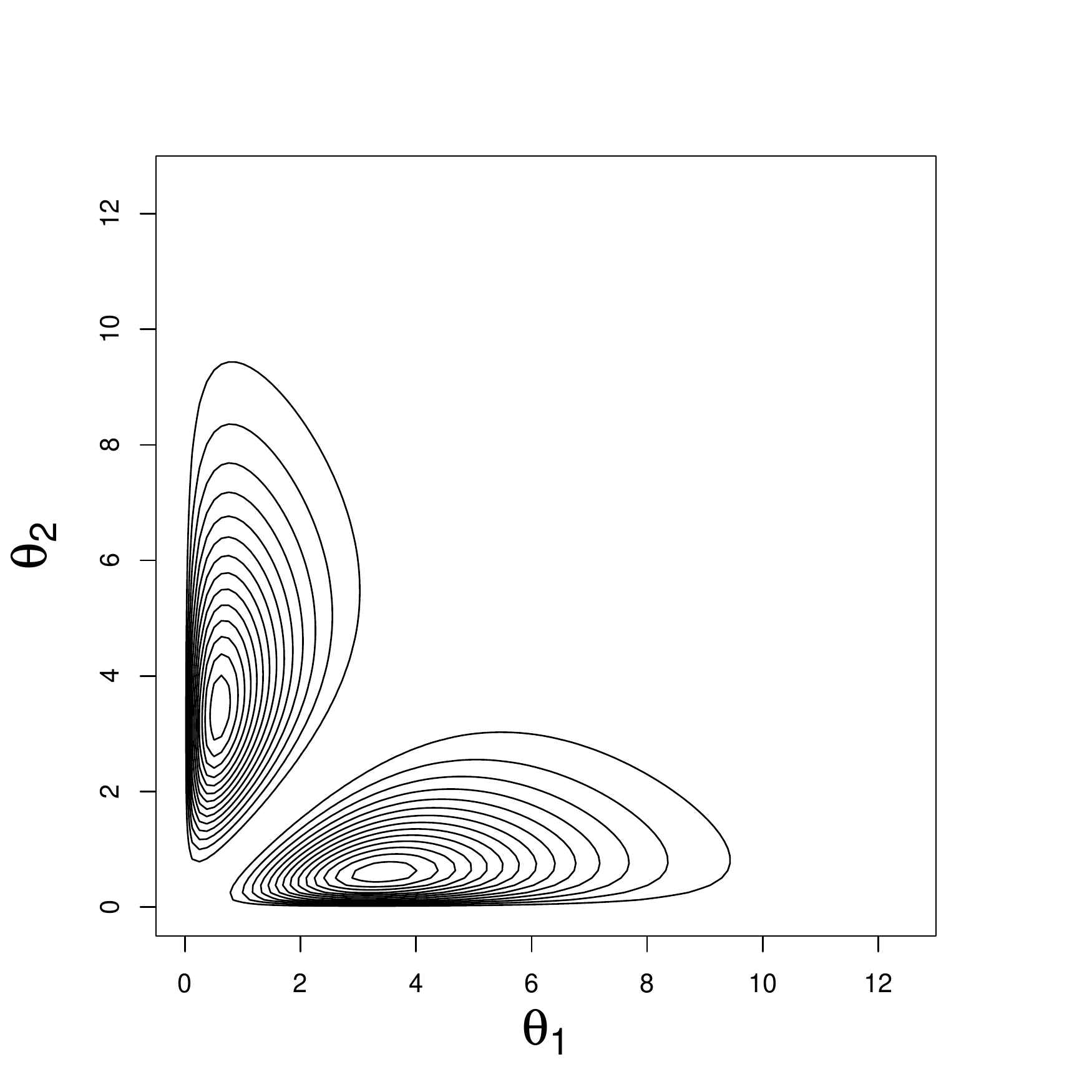}}
	\subfloat[$(\alpha, \zeta) = (1,2.1)$]{\includegraphics[width=0.35\linewidth]{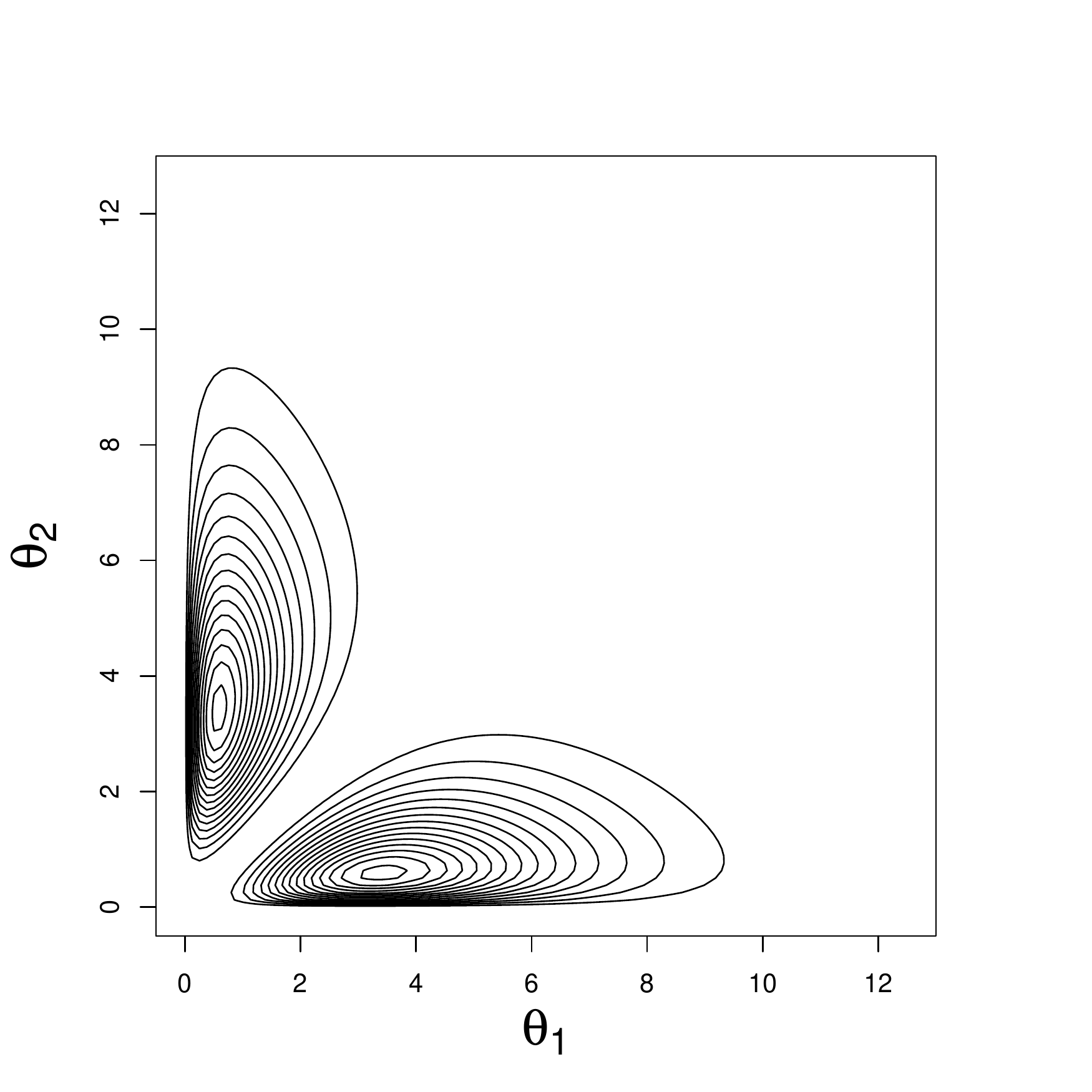}}
	
	\centering
	\subfloat[$(\alpha, \zeta) = (2,2)$]{\includegraphics[width=0.35\linewidth]{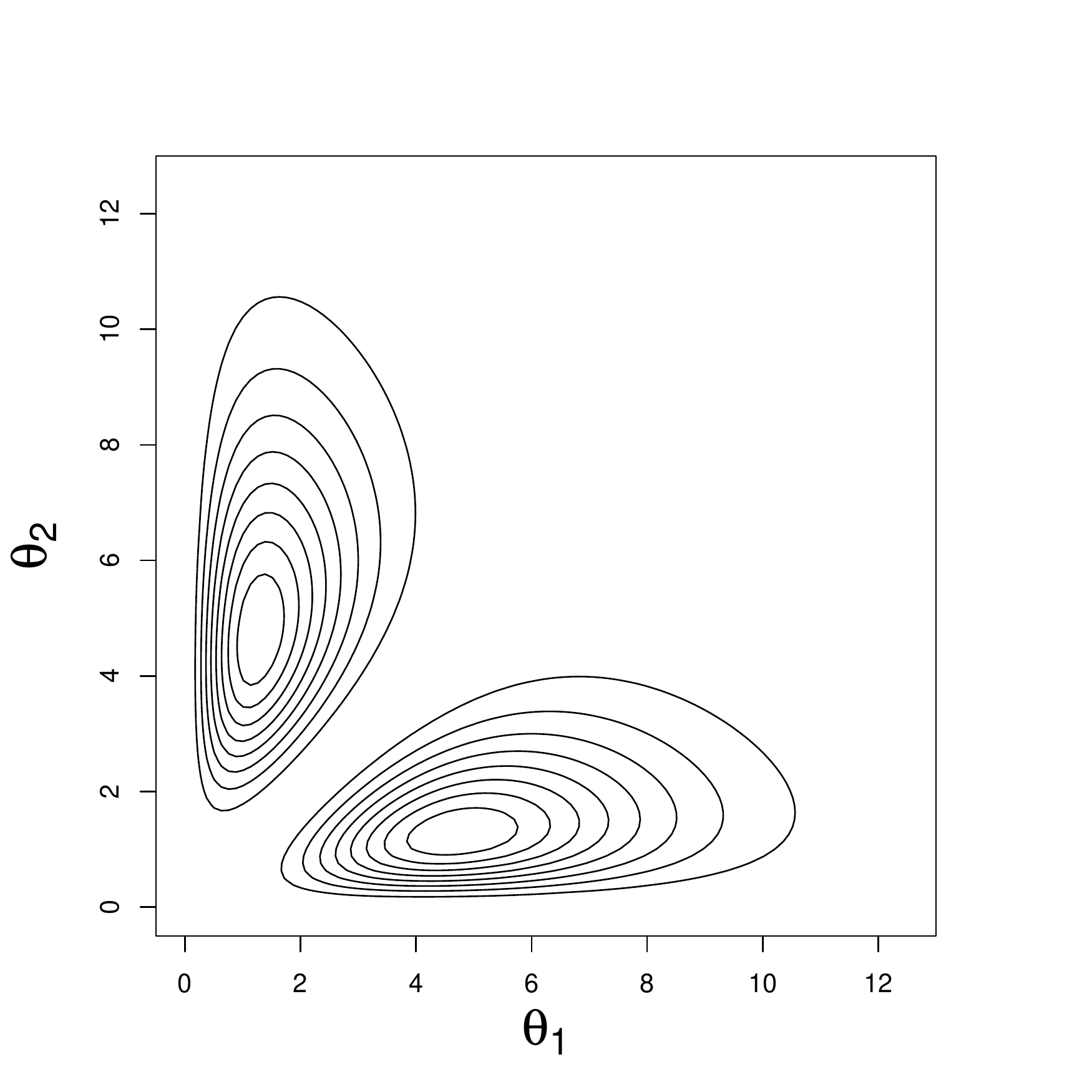}}
	\subfloat[$(\alpha, \zeta) = (2,2.01)$]{\includegraphics[width=0.35\linewidth]{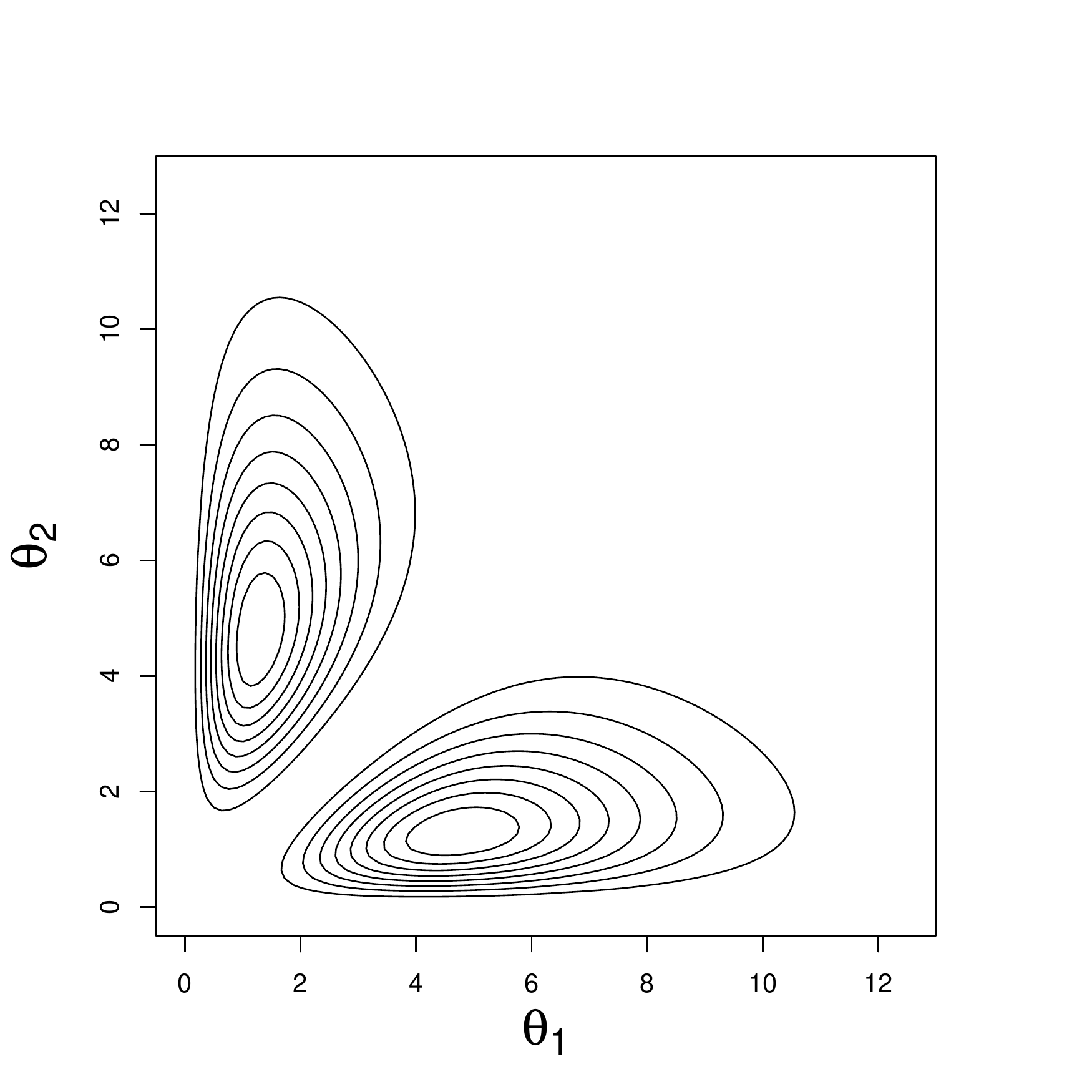}}
	\subfloat[$(\alpha, \zeta) = (2,2.1)$]{\includegraphics[width=0.35\linewidth]{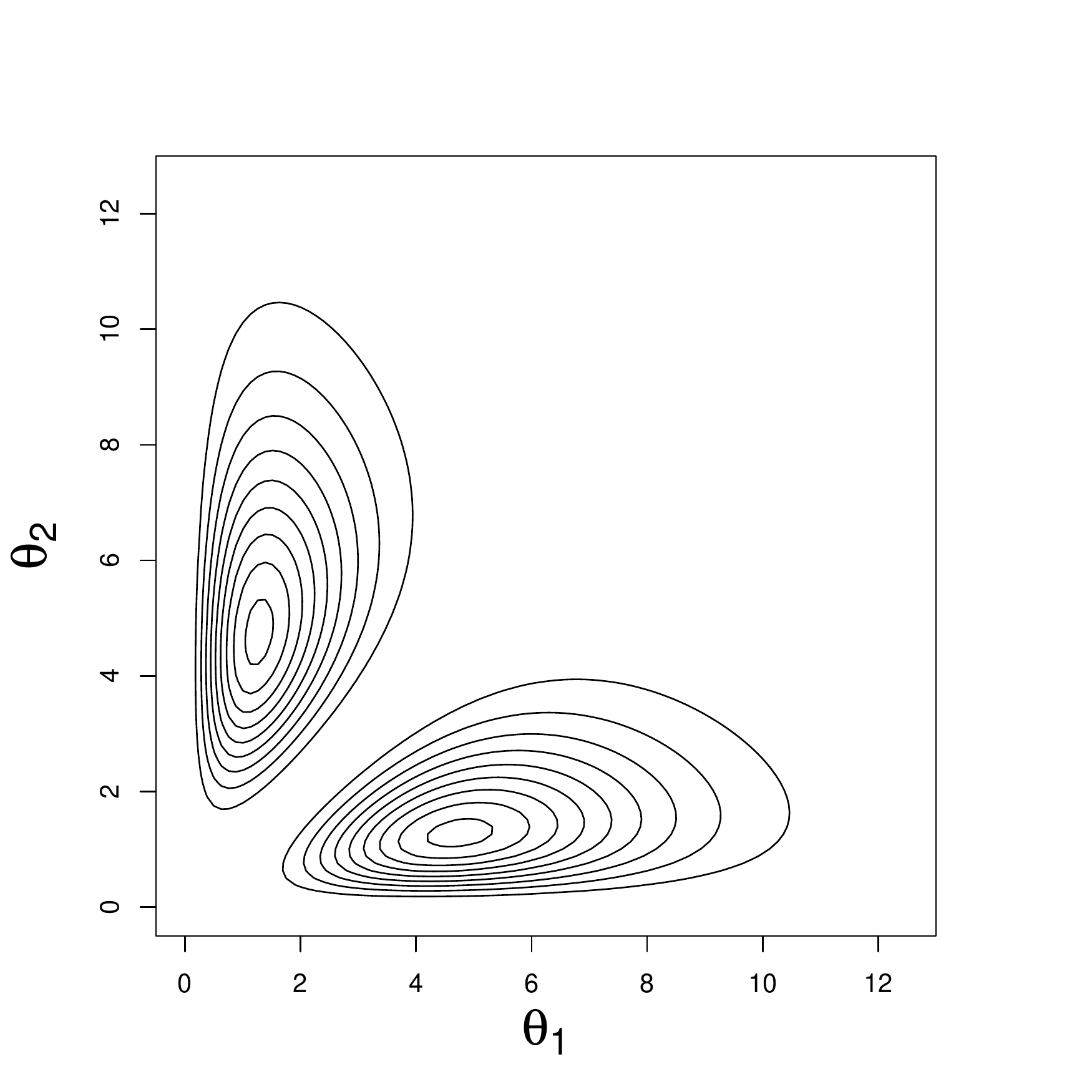}}
	
	\centering
	\subfloat[$(\alpha, \zeta) = (-0.5,1)$]{\includegraphics[width=0.35\linewidth]{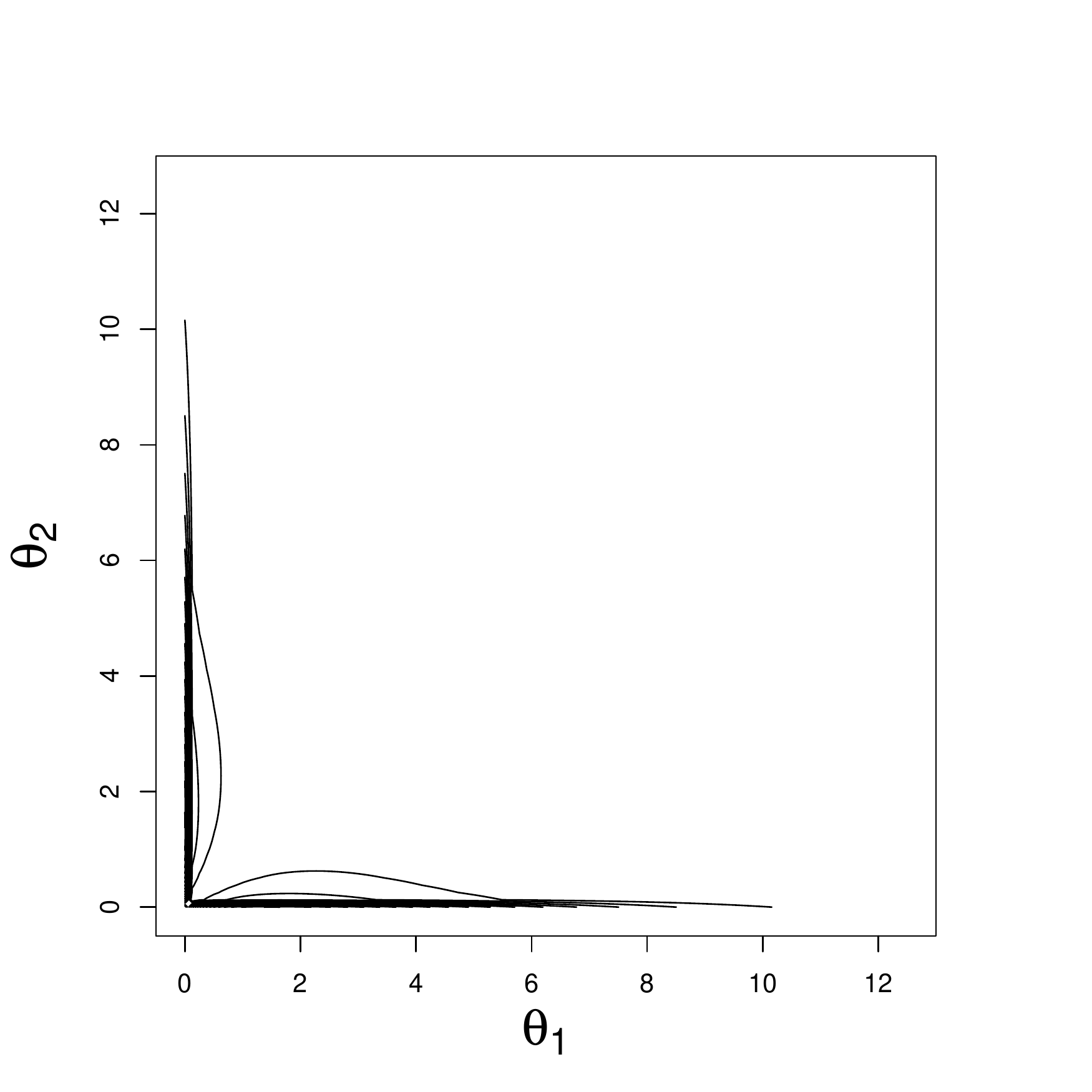}}
	\subfloat[$(\alpha, \zeta) = (3,5)$]{\includegraphics[width=0.35\linewidth]{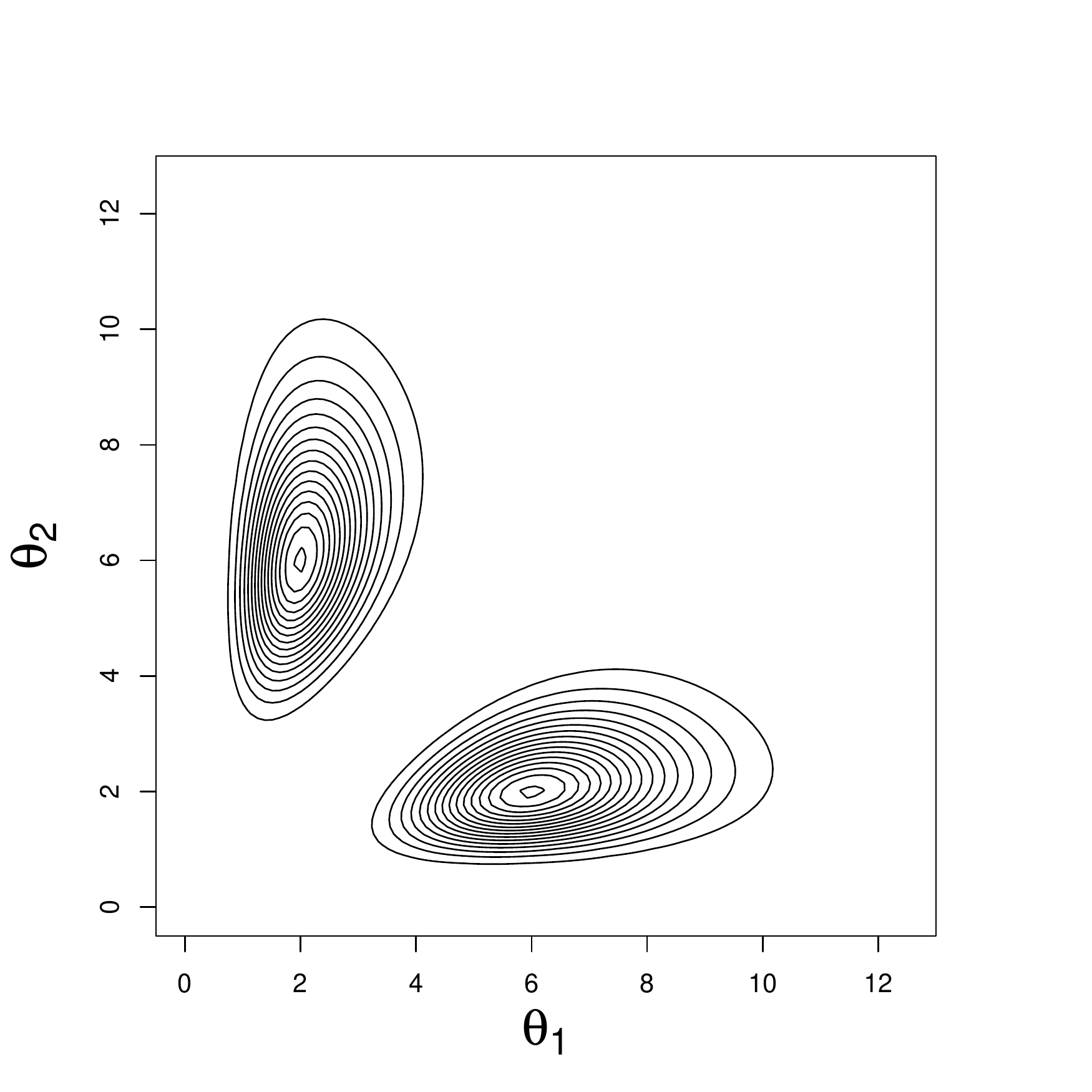}}
	\subfloat[$(\alpha, \zeta) = (4,10)$]{\includegraphics[width=0.35\linewidth]{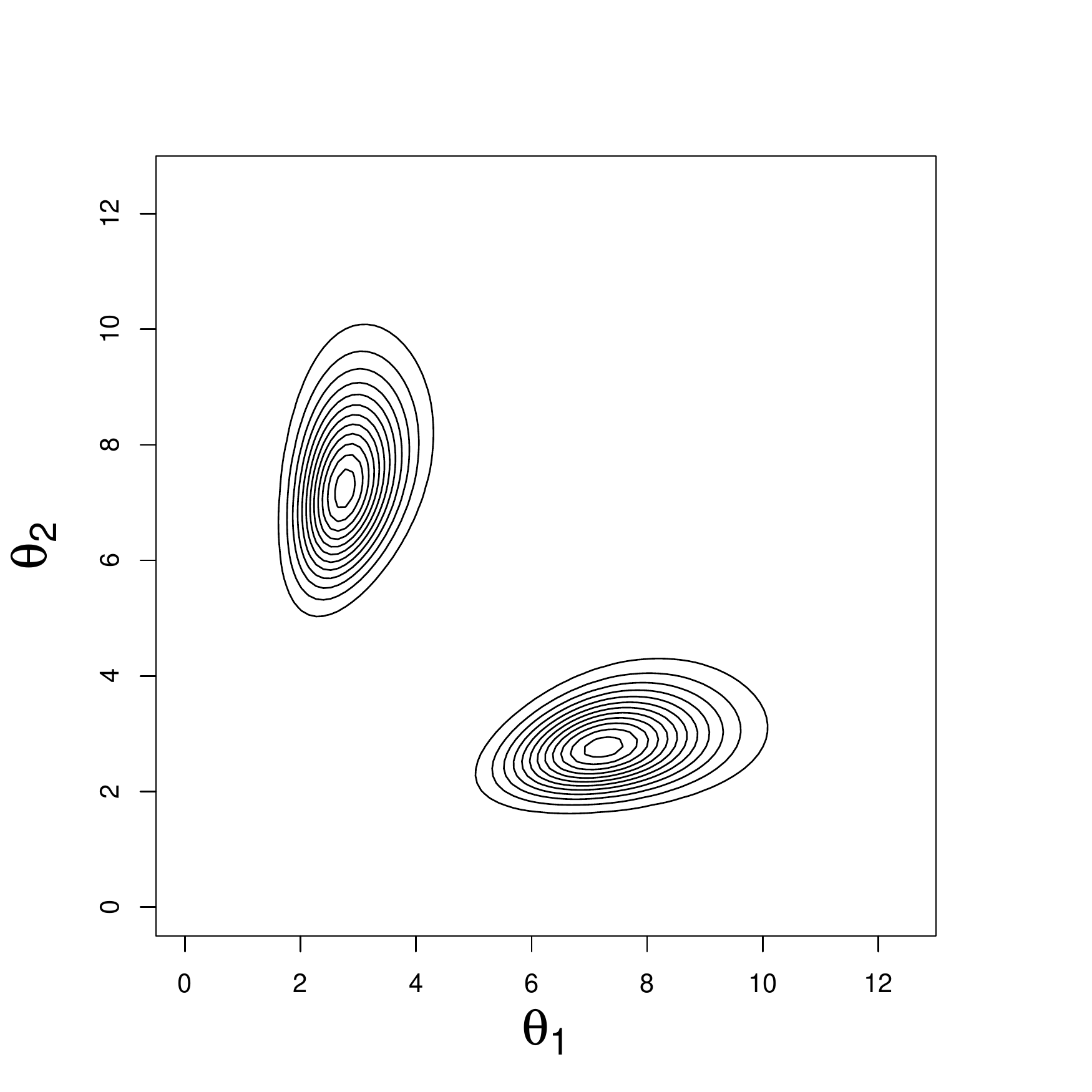}}
	\caption{Laguerre ensemble. Joint distribution of the eigenvalues of a 2-dimensional Wishart matrix for different choices of the parameters $(\alpha, \zeta)$.}
	\label{fig:Wishart_pdf_prior}
\end{figure}

\begin{figure}[ht]
	\centering
	\subfloat[{\tiny $(\alpha, \beta, \zeta) = (2,2,0.5)$}]{\includegraphics[width=0.35\linewidth]{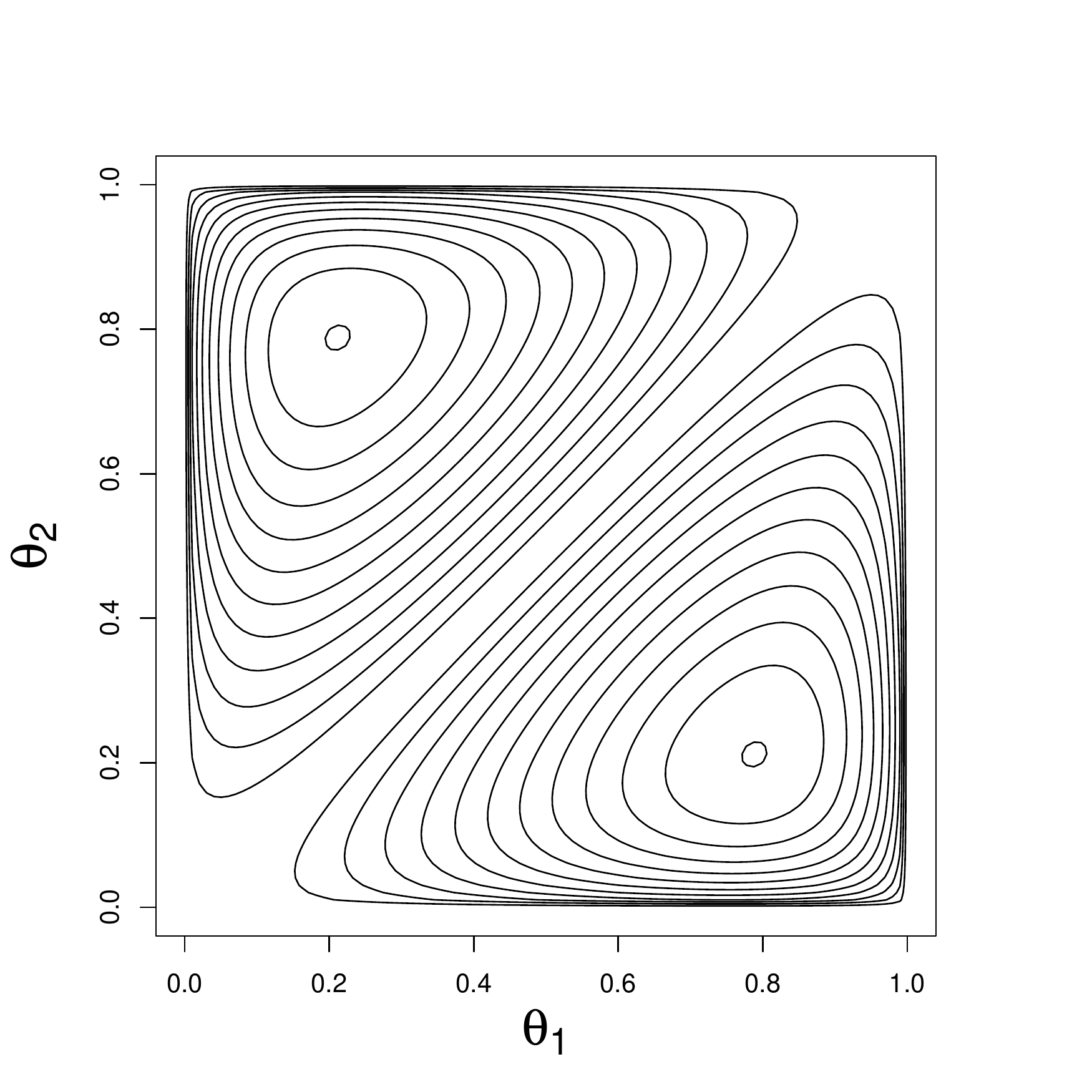}}
	\subfloat[{\tiny $(\alpha, \beta, \zeta) = (2,5,0.5)$}]{\includegraphics[width=0.35\linewidth]{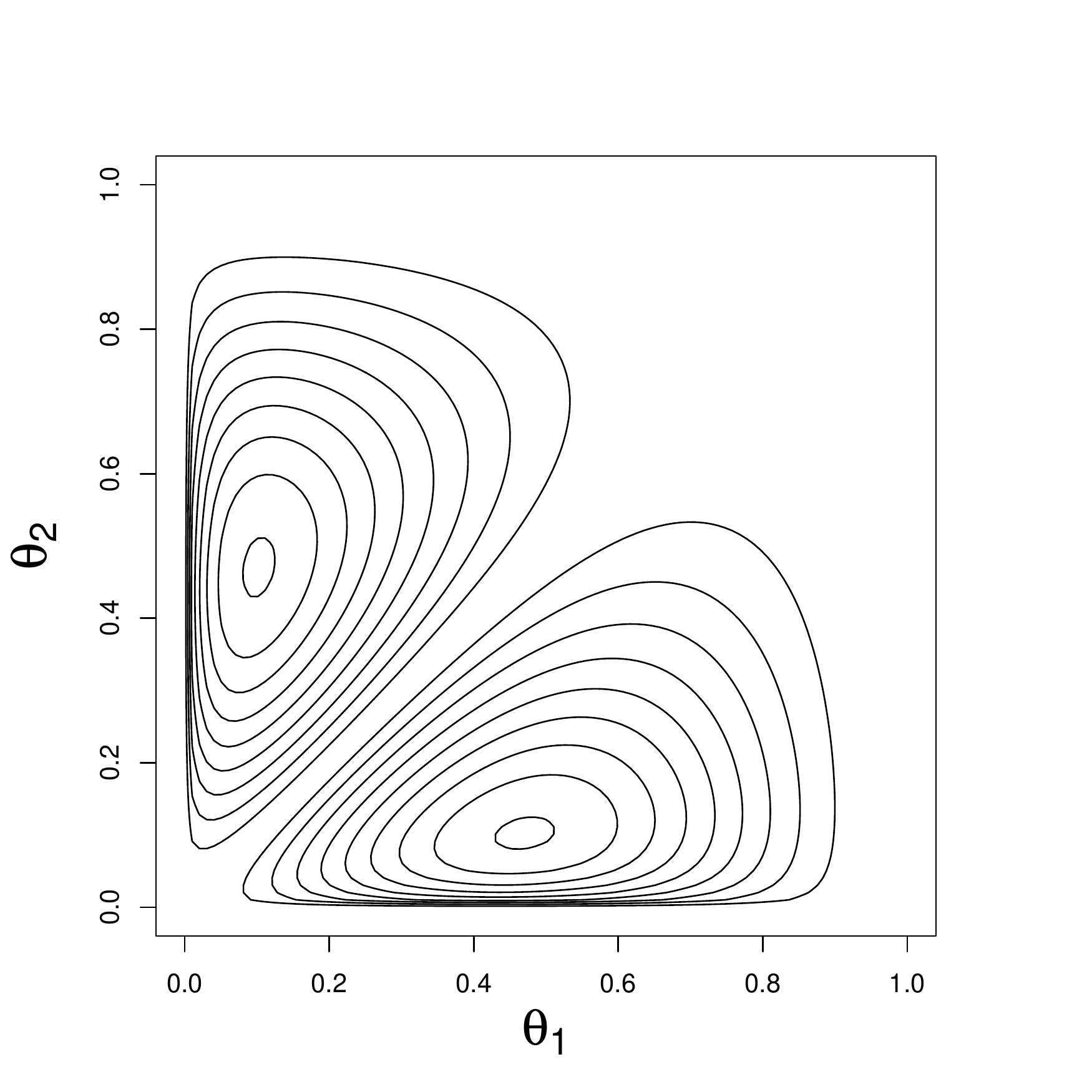}}
	\subfloat[{\tiny $(\alpha, \beta, \zeta) = (2,10,0.5)$}]{\includegraphics[width=0.35\linewidth]{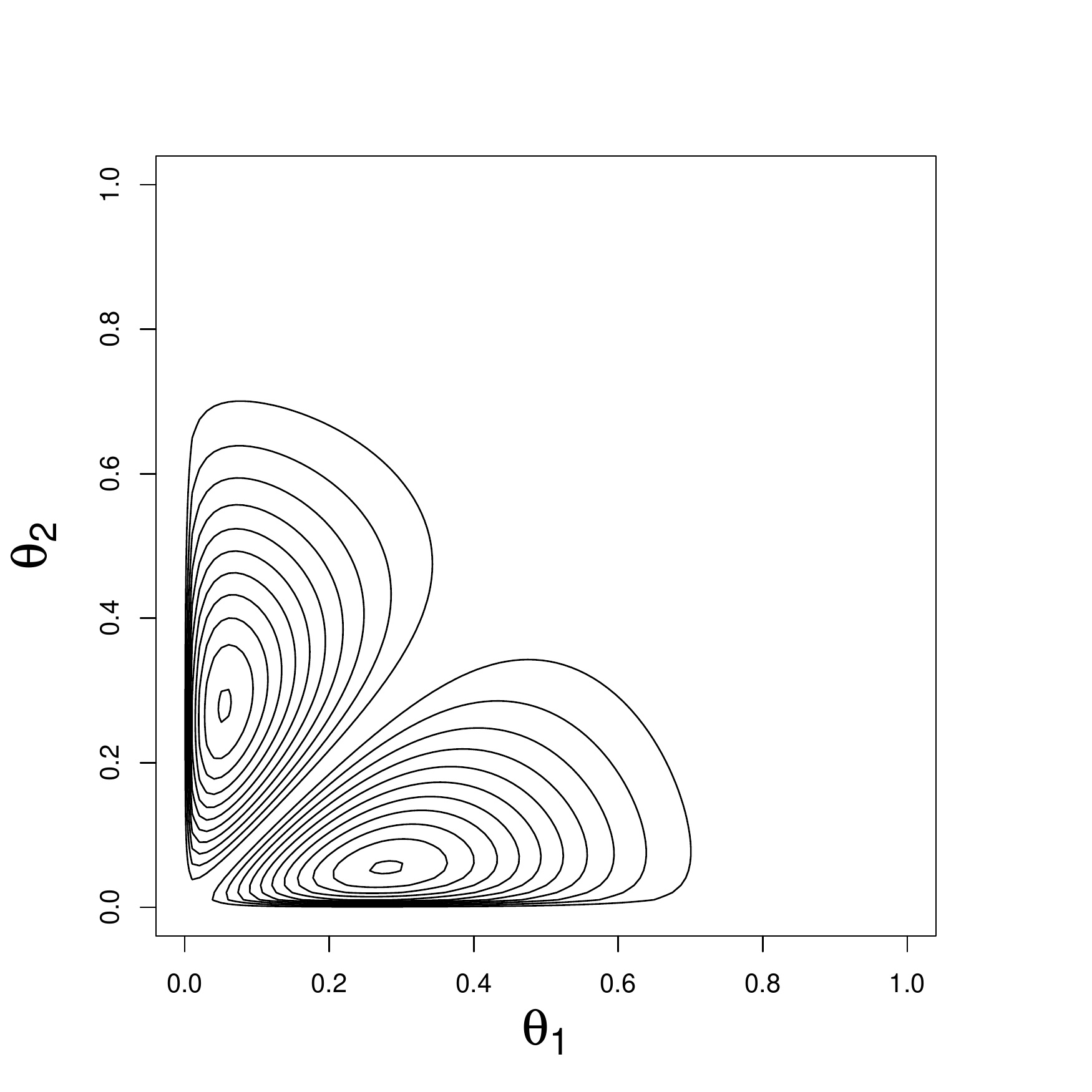}}
	
	\centering
	\subfloat[{\tiny $(\alpha, \beta, \zeta) = (5,2,1)$}]{\includegraphics[width=0.35\linewidth]{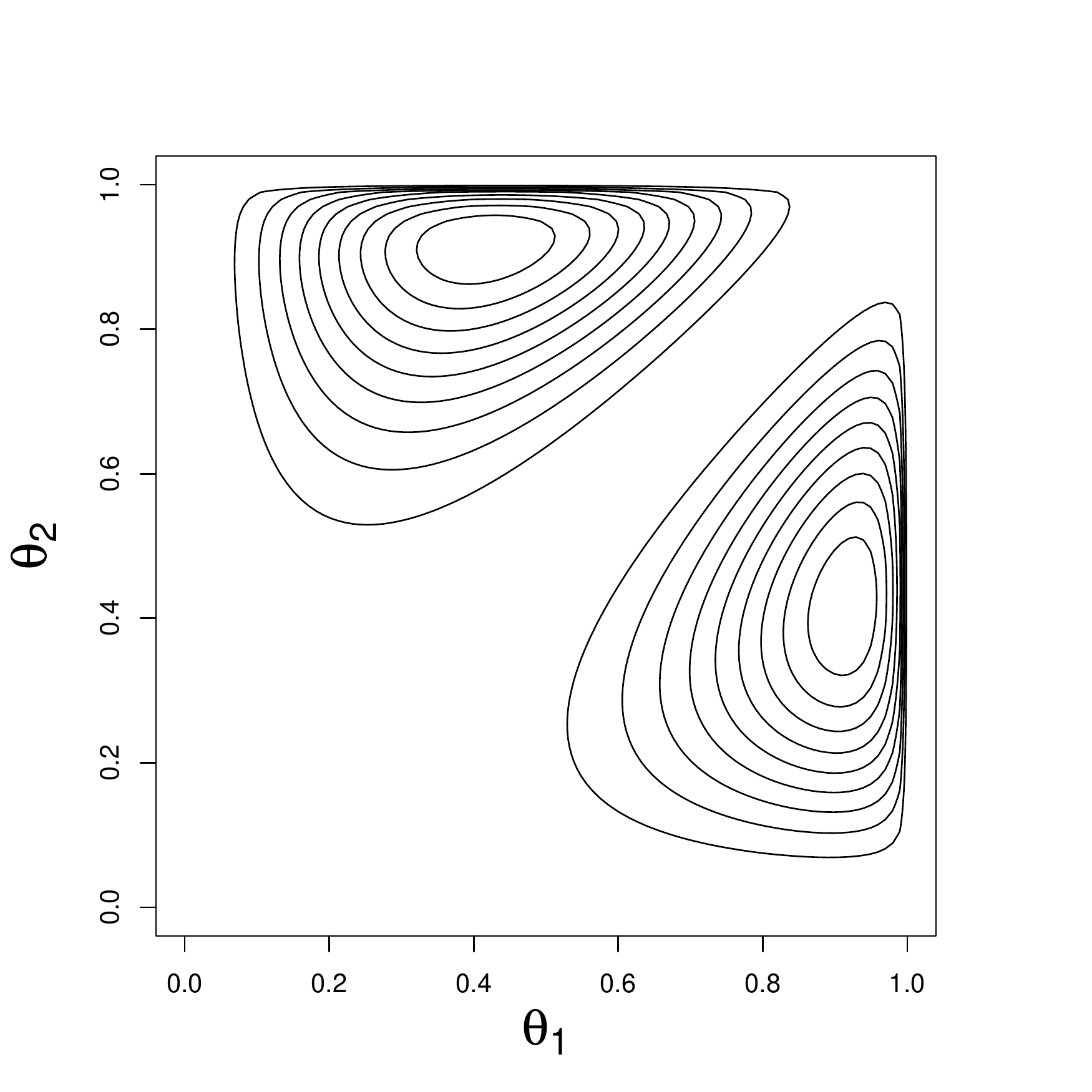}}
	\subfloat[{\tiny $(\alpha, \beta, \zeta) = (5,5,1)$}]{\includegraphics[width=0.35\linewidth]{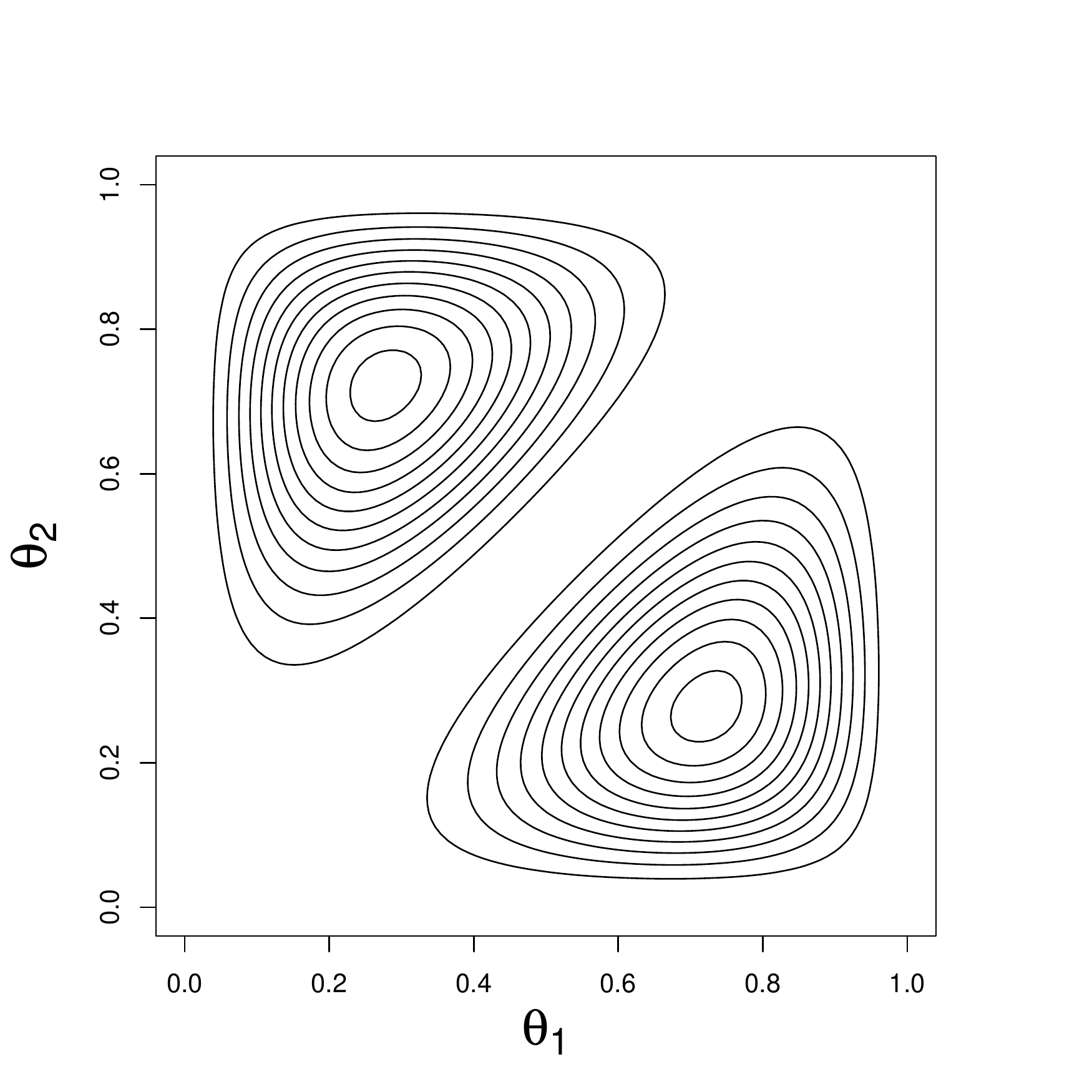}}
	\subfloat[{\tiny $(\alpha, \beta, \zeta) = (5,10,1)$}]{\includegraphics[width=0.35\linewidth]{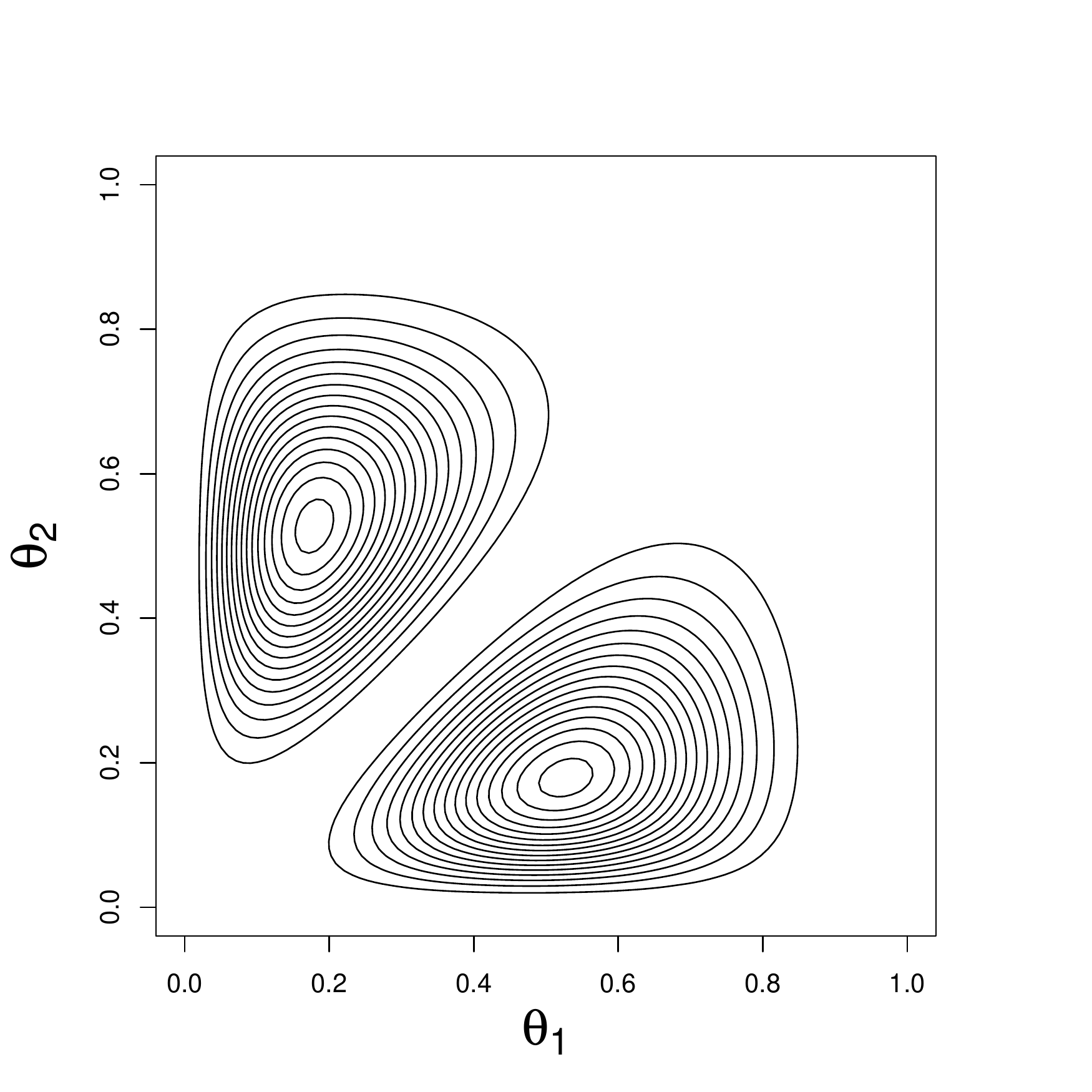}}
	
	\centering
	\subfloat[{\tiny $(\alpha, \beta, \zeta) = (10,2,5)$}]{\includegraphics[width=0.35\linewidth]{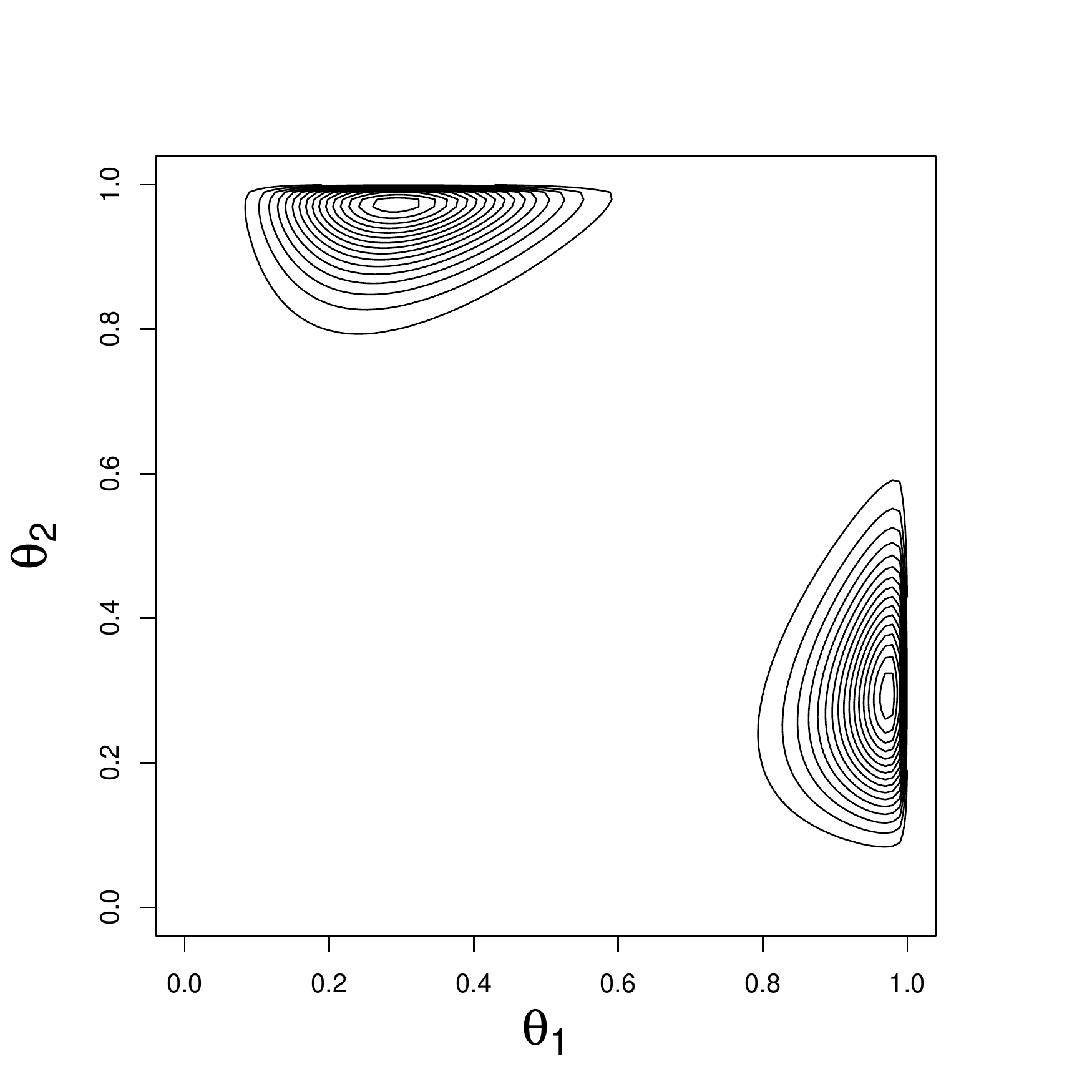}}
	\subfloat[{\tiny $(\alpha, \beta, \zeta) = (10,5,5)$}]{\includegraphics[width=0.35\linewidth]{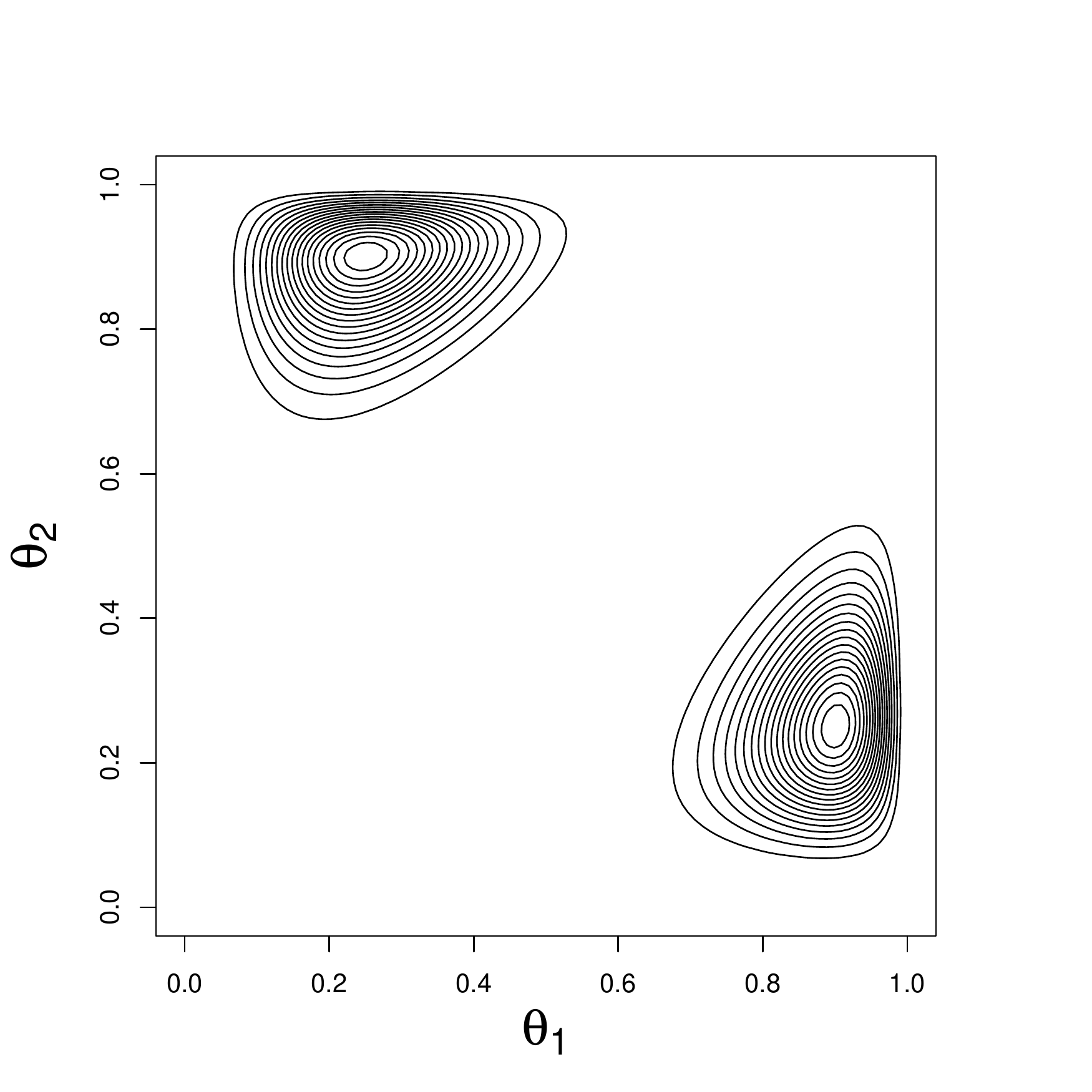}}
	\subfloat[{\tiny $(\alpha, \beta, \zeta) = (10,10,5)$}]{\includegraphics[width=0.35\linewidth]{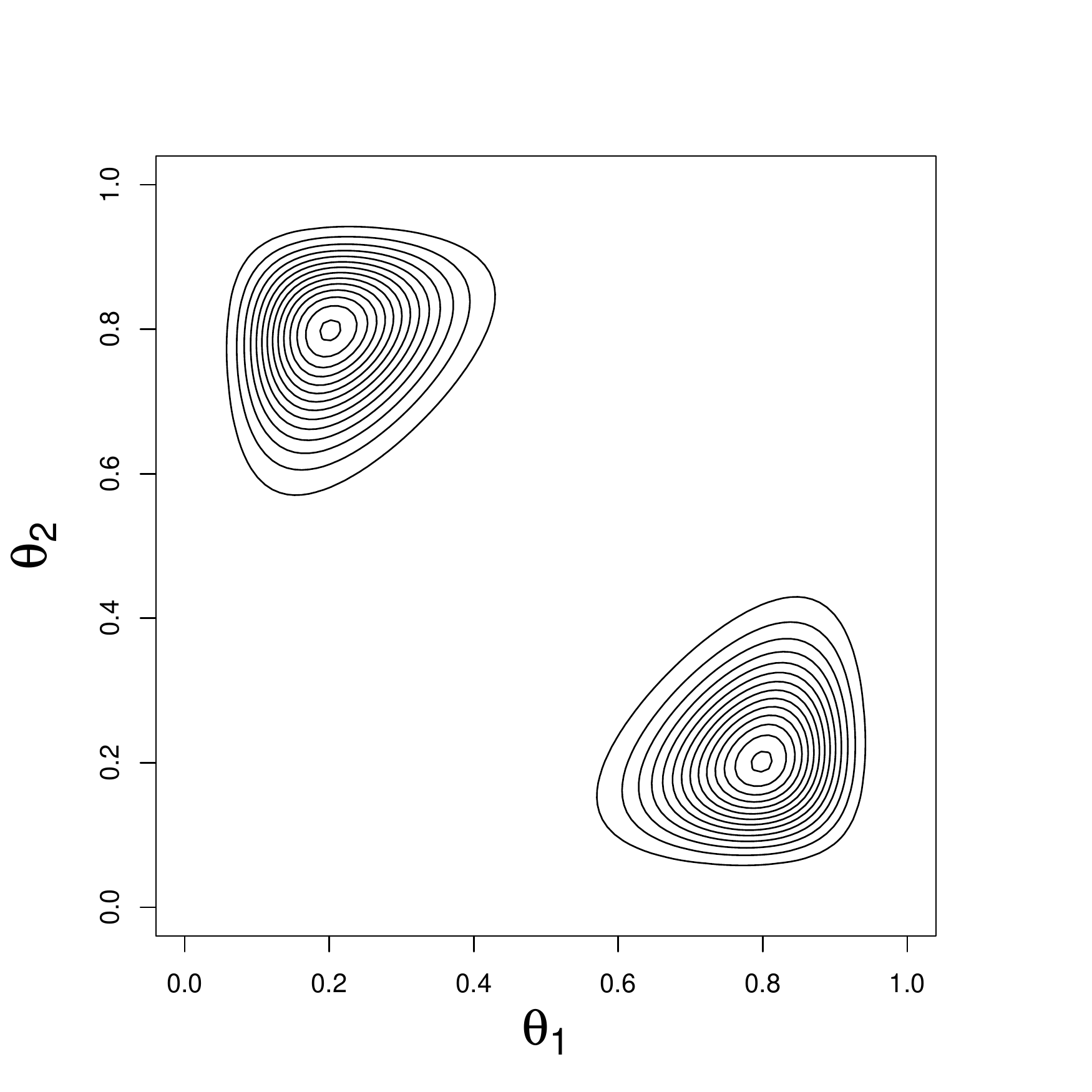}}
	\caption{Jacobi ensemble. Joint distribution of the eigenvalues of a 2-dimensional Beta matrix for different choices of the parameters $(\alpha, \beta, \zeta)$.}
	\label{fig:Beta_pdf_prior}
\end{figure}

We conclude this section by noting that the such eigenvalues distribution share a common form as the repulsive priors proposed previously: they are given by the product of independent kernels (in the three cases Gaussian, Gamma, Beta, respectively) and a repulsion term which is given by $\exp\left(\log \left(r_{ij}\right) \right)$. The repulsion we propose grows slower than the one utilised, for instance, in \cite{quinlan2021class} and allows for analytical tractability. On the other hand, their repulsion is very strong near the origin. In the Wishart case, when we do not have the identity matrix, we lose the normalising constant as well.

\section{Gibbs Measures and Chaos}\label{sec:gibbs-measures-and-chaos}

Allowing the number of components in our model to be arbitrarily large introduces some additional theoretical difficulties. Namely, we require the existence of an infinite-dimensional object from which our finite-dimensional prior distributions can be obtained. The most familiar method of obtaining such an infinite-dimensional object is the Kolmogorov's extension theorem. In this theorem, the existence and uniqueness of the infinite-dimensional distribution of interest are guaranteed under mild conditions on a set of finite-dimensional distributions. These finite-dimensional distributions are thus obtained from the infinite-dimensional one via marginalisation. One of these conditions, usually referred to as Kolmogorov consistency, is not satisfied by the finite-dimensional distributions presented in this work. Instead of Kolmogorov consistency, we consider the notion of DLR consistency, named for its theorists Dobrushin, Randford and Ruelle \citep{Dobrushin1968GibbsianRF, lanford1969observables}. The DLR approach is based on the definition of \textit{specifications}, i.e. families of probability kernels constructed using a given Hamiltonian (i.e., the \textit{interaction principle}) that are consistent via composition of kernels operation. Note that the notion of consistency here is different from the Kolmogorov one, which requires consistency via marginalisation. Under some conditions, it is possible to find a set of infinite-dimensional probability measures which is compatible with a given specification, i.e. each kernel in the specification yields a regular conditional distribution. The resulting set is a set of Gibbs measures. It can be shown that this set, when it exists, is a simplex, and as such every element in it is a linear combination of its extreme points (also called in statistical mechanics infinite-volume Gibbs states, in relation with the states of a physical system undergoing a phase transition). We refer to \cite{friedli_velenik_2017}, ch. 6, for a thorough exposition on this topic. Unlike in the Kolmogorov case, DLR consistency guarantees neither existence nor uniqueness. Nevertheless, it is the natural notion for specifying an infinite-dimensional object in terms of conditional, rather than marginal, distributions. It is at the heart of the theory of Gibbs measures. To obtain additional useful properties, e.g., existence and uniqueness, several tools are available. Here, we make use of a large deviation principle.

Once again we refer to work in statistical physics, where the need for an alternative definition (i.e. different from the Kolmogorov one) of a measure on an infinite-dimensional space derives from the necessity to describe complex systems of particles interacting with each other \citep{georgii2011gibbs}. Indeed, in this setting, Gibbs measures can be used to describe the equilibrium states of the particles with respect to the given interaction at a fixed temperature conditionally on the external system. Gibbs measures are defined in terms of conditional distributions rather than marginal distributions.

Consider a system of particles with configuration $\bm \theta$ and internal energy equal to $\mathcal{H}\left(\bm \theta \mid \zeta,h\right)$ in Eq.~\eqref{eq:EnergyIntro}, described by the Boltzmann Entropy law in Eq.~\eqref{eq:Boltzmann_Entropy}. Depending on the expressions chosen for the functions $\psi_1 $ and $\psi_2$ in the Hamiltonian $\mathcal{H}\left(\bm \theta \mid \zeta,h\right)$, a Gibbs measure corresponding to this specification of the conditional distributions needs not to exist, but even more interestingly--and in contrast to the Kolmogorov extension--its uniqueness is also not guaranteed. A popular example of such models is the Ising model for $d \geq 2$ \citep{ising1924beitrag}. Interestingly, for certain values of the parameters $h$ and  $\zeta$, the Ising model has a unique associated Gibbs measure, but its marginal distributions are never consistent (in the Kolmogorov sense) \citep{friedli_velenik_2017}. Thus, we see that the marginal distributions associated with a Gibbs measure do not generally satisfy the Kolmogorov consistency conditions, even when the associated Gibbs measure is unique. See Section \ref{sec:IsingExample} for a thorough discussion.

In general, there is no way to express the marginal distributions associated to an infinite-lattice version of the measure \eqref{eq:Boltzmann_distr} without making explicit reference to the Hamiltonian \citep{friedli_velenik_2017}. Moreover, the physical systems with distributions of the form \eqref{eq:Boltzmann_distr} are meant to model systems that undergo so-called phase transitions, abrupt changes in the physical properties of a system, such as the familiar liquid-to-gas transition of water boiling into steam. The association of phase transitions to non-uniqueness of Gibbs measure is due to \cite{Dobrushin1968GibbsianRF}, but it does not coincide with the physical notion in all cases \citep{RAS2015LDPgibbs}.

As Gibbs measures, we will use distributions derived from random matrices, so we will be able to establish the existence and uniqueness of the limiting Gibbs measures by using the tools of random matrix theory.

\paragraph{Chaos}

Beyond existence and uniqueness, we will show that the distributions \eqref{eq:Eigen_Gaussian}, \eqref{eq:Eigen_Wishart} and \eqref{eq:Eigen_Beta} satisfy another useful property, which in the physics literature is sometimes referred to as molecular chaos and can be seen as a sort of asymptotic independence. 
In the context of thermodynamic studies, molecular chaos is closely related to the second law of thermodynamic, which states that the total entropy of an isolated system never decreases. In particular, the statistical foundation of the second law of thermodynamics are based on the \textit{molecular chaos hypothesis}, assuming that the velocities of colliding particles are uncorrelated and independent of their individual positions \cite{boltzmann2003further}. 
More formally, let $p_M(\theta_1,\dots,\theta_M)$ be the joint law of some collection of random variables and let
\begin{equation}\label{marginal_dist}
p_{M,k} = \int\limits_{\mathbb{R}^{M-k}} p_M(\theta_1,\dots,\theta_M) d\theta_{k+1} \cdots d\theta_M
\end{equation}
be the marginal distribution of the first $k$ variables. The functions $p_{M,k}$ are referred to as correlation functions in some sources \citep{Monvel1995}. Let $f$ be a distribution on $\mathbb{R}$. Then, the sequence $\{p_M\}_{M=1}^\infty$ is called $f$-chaotic if $p_{M,k}$ converges weakly to the product $\prod_{i = 1}^k f(x_i)$. This notion is equivalent to another that is more familiar in random matrix theory. First, we need another definition.

Given the random variables $\theta_1, \dots, \theta_M$, we can define the empirical measure
\[ 
\mu_{M} := \frac{1}{M} \sum_{i=1}^M \delta_{\theta_i}
\]
where $\delta_x$ is the Dirac measure at $x$. It is well known in the literature \citep{JW2017MFLSP} that weak convergence of the empirical  measure to a measure $f$ is equivalent to the sequence $p_{M,k}$ being $f$-chaotic. When the random variables are the eigenvalues of a symmetric matrix, the empirical measure is referred to as the empirical spectral distribution. For the Gaussian \citep{wigner1958distribution}, Laguerre \citep{mp1967wishart} and Jacobi \citep{bai2015beta} ensembles, the empirical spectral distributions converge under an appropriate limit. Convergence of the empirical measures for the generalised distributions discussed previously, as well as the existence and uniqueness of the associated Gibbs measure, follows from a so-called large deviation principle \citep[LDP,][]{agz2009rmt}. We end this section by formalising the idea of an LDP and showing its connection to Gibbs measures.

\paragraph{Large Deviations}

The most classical method for showing the uniqueness of a Gibbs measure for a particular system are conditions due to Dobrushin \citep{dobrushin1970prescribing} and Dobrushin and Shlosman \citep{DobrushinShlosman1985}. These are often referred to as the Dobrushin Uniqueness Criterion and the Dobrushin-Shlosman Uniqueness Criterion, respectively. Here, we take a different approach that is more convenient in the context of random matrices, so we will not say more about these criteria, other than to note that they give sufficient, but not necessary conditions for uniqueness of a Gibbs measure.

Let $\mu_n$ be a sequence of probability measures on a topological space $T$. Let $I: T \to [0,\infty]$ be a lower-semicontinuous function and $r_n$ a sequence of positive real constants. Then the sequence $\{\mu_n\}$ is said to satisfy a LDP with rate function $I$ and scale $r_n$ if
\begin{enumerate}
\item for all open $G \subset T$, we have
\[ 
\liminf_{n} \frac{1}{r_n} \log \mu_n(G) \geq -\inf_{x \in G} I(x) 
\]
\item for all closed $F \subset T$, we have
\[ 
\limsup_{n} \frac{1}{r_n} \log \mu_n(F) \leq -\inf_{x \in F} I(x) 
\]
\end{enumerate}
When the sets $\{I \leq c\}$ are compact, the rate function is often called good or tight.

LDPs can be very helpful when working with Gibbs measures because every Gibbs measure corresponds to a zero of the rate function (see below), and we will be able to show that the rate function associated to the empirical measure has a unique zero. The rate function $I$ can be defined via a certain limit, and is guaranteed to exist by Theorem 8.3 in \cite{RAS2015LDPgibbs}. 

\begin{theorem}[DLR Variational Principle, \cite{RAS2015LDPgibbs}]
Fix a shift-invariant absolutely summable continuous interaction principle $\Phi$. Let $\mathcal{M}_\theta(\Omega)$ denote the space of invariant probability measures and let $\gamma \in \mathcal{M}_\theta(\Omega)$. Then $\gamma$ is a Gibbs measure for the specification determined by $\Phi$ if and only if $\gamma$ is a zero of the rate function.
\end{theorem}

To help build intuition about the usefulness of large deviations theory, we first present a classical result. We will use the notation $\mathcal{M}_1(\mathcal{S})$ to denote the collection of probability measures on the space $\mathcal{S}$.

\begin{definition}
Let $\nu, \lambda \in \mathcal{M}_1(\mathcal{S})$. Then, the entropy of $\nu$ relative to $\lambda$ is defined to be
\[ H(\nu|\lambda) = \left\{\begin{array}{lll} \int f \log f d\lambda & &  \nu \ll \lambda \textup{ and } f = \frac{d\nu}{d\lambda} \\ \infty & & \textup{otherwise.}  \end{array}\right. \]
\end{definition}

\begin{theorem}[Sanov's Theorem. \citep{sanov1961probability}]
Let $\mathcal{S}$ be a Polish space and let $\lambda \in \mathcal{M}_1(\mathcal{S})$ be a probability measure on $\mathcal{S}$. Let $\{X_n\}$ be a sequence of i.i.d. $\mathcal{S}$-valued $\lambda$-distributed random variables. Let $L_n$ be the associated empirical measure and define $\rho_n(A) = \mathbb{P}(L_n \in A)$ for Borel subsets $A \subset \mathcal{M}_1(\mathcal{S})$. Then the family $\rho_n$ satisfies an LDP with rate function $I(\nu) = H(\nu|\lambda)$.
\end{theorem}

In the case when $\mathcal{S} = \mathbb{R}$ and $\lambda$ has density $f$, we can interpret Sanov's theorem as saying that, under the assumption that our data is i.i.d. with distribution $f$, the probability that a histogram looks like the density of a different measure $\nu$ decays like $e^{-H(\nu|\lambda)n}$ \citep{RAS2015LDPgibbs}.

\paragraph{Connections}

Then weak convergence of the sequence of empirical measures to some deterministic measure $\mu$ is equivalent to the sequence $p_M$ being $\mu$-chaotic \citep{JW2017MFLSP}. Moreover, if the sequence $p_M$ satisfies a LDP, then it must converge weakly, so we also have that LDP implies chaos. The Dobrushin Uniqueness criterion is strictly stronger than the existence of a unique Gibbs measure.

\begin{figure}[h]
	\centering
\begin{tikzpicture}[scale=0.65]
	\node (LDP) at (-2.5,3.5) {LDP};
	\node (UGibbs) at (-2.5,7) {\makecell{Unique \\ Gibbs \\ Measure}};
	\node (UDobr) at (2.5,7) {\makecell{Dobrushin \\ Uniqueness}};
	\node (EDC1) at (-2.5,0) {\makecell{Empirical \\ Distribution \\ Converges}};
	\node (MC) at (3,0) {\makecell{Molecular \\ Chaos}};
		
	\draw[-latex, <->] (LDP) -- node [midway, above] {{\tiny DLR}}(UGibbs);
	\draw[-latex, ->] (UDobr) --node [midway, above] {{\tiny Dobrushin}} (UGibbs);
	\draw[-latex, ->] (LDP) -- node [midway, above] {{\tiny Guionnet et. al.}} (EDC1);
	\draw[-latex, <->] (EDC1) --node [midway, above] {{\tiny Jabin \& Wang}} (MC);
\end{tikzpicture}
	\caption{Connections between properties of Gibbs Measures.}\label{fig:Th_Links}
\end{figure}
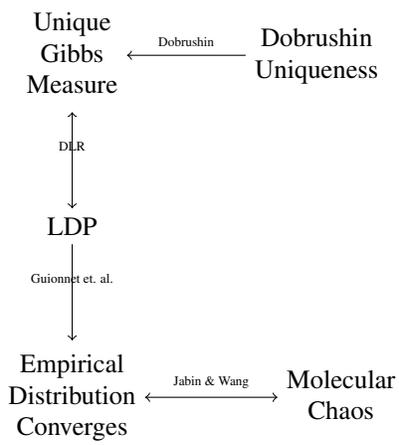

The connections between these various concepts are visualised in the diagram in Figure \ref{fig:Th_Links}. In brief
\begin{itemize}
\item The DLR Variational Principle tells us that uniqueness of the Gibbs measure is equivalent to having a LDP whose rate function has a unique zero
\item A LDP also implies the convergence of the empirical measure, but the converse is not true in general
\item Convergence of the empirical measure is equivalent to molecular chaos
\end{itemize}

\section{Large Deviations}\label{sec:LDP}

As discussed previously, measures of the type \eqref{eq:Boltzmann_distr} and \eqref{eq:Gibbs_distr} are often used in statistical physics to model interacting systems of particles. In this context, one is interested in the conditional distribution conditioned on events of extremely low probability. Bounds on the probability of rare events (e.g., the weak law of large numbers) are an elementary topic in probability. In the context of statistical mechanics, the usual bounds are generally insufficient, so we turn to LDPs. As we have just seen, establishing a LDP will have far-reaching consequences. LDPs have already been shown for the Gaussian and the Wishart type distributions introduced earlier \citep{guionnet2004large, HP1998WishartLDP}. Here, we present an analogous result for the  matrix variate Beta distribution discussed above. Our approach is inspired by \cite{HP2000unitaryLDP}.

In order to establish a LDP for the Beta matrix, we need some control over the growth of the parameters. Namely, we require that $M$ grows with $N_1$ and $N_2$ and that the following relationships are satisfied:
\begin{equation*}
\begin{gathered}
    \lim_{N_1 \to \infty} \frac{M}{N_1} = a, \quad \lim_{N_2 \to \infty} \frac{M}{N_2} = b \\
    \lim_{N_1 \to \infty} \lim_{N_2 \to \infty} \frac{M}{N_1+N_2} = \frac{ab}{a+b} =:c \in (0,1)
\end{gathered}
\end{equation*}
These assumptions are required in order for the empirical spectral measure to converge \citep{bai2015beta}. We will also require that the distance between the parameters $N_1$ and $N_2$ does not grow too quickly. For notational simplicity, we set $\gamma(x) = x - (M+1)/2$ and let $A = a(1-a/2)$. Then we require that the parameters $N_1$ and $N_2$ satisfy 
\[ 
\lim_{N_1 \to \infty} \frac{\gamma\left(N_1\right)}{N_1} = \lim_{N_1 \to \infty} \frac{\gamma\left(N_2\right)}{N_1} = 1 - \frac{a}{2} 
\]
For example, requiring that $\mid N_1 - N_2 \mid \leq k$ ensures this condition is met. More generally, sublinear growth of the distance between the parameters is sufficient. Using the notation above, we write the joint distribution as
\begin{equation}\label{betaLDP} 
p_M\left(\bm \theta \mid \gamma\right) = \mathcal{B}^{-1}_{N_1,N_2, \zeta, M} \prod_{l = 1}^M \left(\theta_l\right)^{\gamma\left(N_1\right)}\left(1 - \theta_l\right)^{\gamma\left(N_2\right)} \prod_{i < j} |\theta_i - \theta_j|^{2 \zeta}
\end{equation}


We state our main theorem below. We use the notation $\mathcal{M}([0,1])$ to denote the collection of probability measures on the interval $[0,1]$. It is important to note that $[0,1]$ is compact and $\mathcal{M}([0,1])$ equipped with the weak-$*$ topology is compact and metrisable. 

\begin{theorem}
Suppose $\theta_1, \dots, \theta_M$ are distributed according to \eqref{betaLDP}, and let $\mu_{\bm U_M}$ be the associated empirical measure. Suppose that the conditions on $N_1,N_2,M$ described above hold, and let $N=N_1$. Then the limit
\begin{equation*}
        \lim_{N \to \infty} \frac{1}{N^2} \log \mathcal{B}_{N_1,N_2, \zeta, M} := B < \infty
\end{equation*}
exists and $\mu_{\bm U_n}$ satisfies the large deviation principle in the scale $N^{-2}$ with rate function
\begin{equation*}
    \begin{gathered}
    I(\mu) = -2a^2\zeta\int\limits_0^1\int\limits_0^1 \log|x-y| d\mu(x)d\mu(y) \\
    - A \int\limits_0^1 \log x + \log(1-x) d\mu(x) + B
        \end{gathered}
\end{equation*}
for $\mu \in \mathcal{M}([0,1])$. Moreover, there exists a unique $\mu_0 \in \mathcal{M}([0,1])$ such that $I(\mu_0) = 0$.
\end{theorem}

In this case, if $\mathcal{A}$ is a base for the topology, then the large deviation principle is equivalent to the conditions \citep{dembo2009large}
\begin{align}
-I(x)  & = \inf \left\{ \limsup_{n \to \infty} \frac{1}{r_n} \log \mu_n(G) : G \in \mathcal{A}, x \in G \right\} \\ &  = \inf \left\{ \liminf_{n \to \infty} \frac{1}{r_n} \log \mu_n(G) : G \in \mathcal{A}, x \in G \right\}
\end{align}
Intuitively, we can think of these as conditions on the tails of the distribution. Now, to obtain our theorem, we must show that
\begin{equation}\label{LDP1}
-I(\mu)  \geq \inf \left\{ \limsup_{n \to \infty} \frac{1}{r_n} \log \mu_n(G) : G \right\}
\end{equation}
and
\begin{equation}\label{LDP2}
-I(\mu)  \leq \inf \left\{ \liminf_{n \to \infty} \frac{1}{r_n} \log \mu_n(G) : G \right\}
\end{equation}
where $G$ runs over neighbourhoods of $\mu$.

In general, the conditions above are equivalent only to a weak LDP, in which case another property of the sequence of measures, exponential tightness, is required to establish the full LDP \citep{RAS2015LDPgibbs}. Due to compactness of the underlying space, we do not need this additional property for the full LDP.

The proof of our theorem will be accomplished via several lemmas, for which proofs are given in Appendix~\ref{appendix1}. In those lemmas, we will make frequent use of the following functions, defined for $(x,y) \in (0,1)^2$, the interior of the unit square.
\begin{align}
    F(x,y) := &-a^2\zeta\log|x-y| \nonumber \\
    & - \frac{A}{2} \left[\log x + \log y + \log(1-x) + \log(1-y) \right]
\end{align}
and
\begin{align}
    F_{N}(x,y) = &-2\frac{M^2}{N^2}\zeta\log|x-y| \nonumber \\
    &- \frac{M\gamma(N_1)}{2N^2} (\log x + \log y) \\
    &- \frac{M\gamma(N_2)}{2N^2}(\log(1-x)+\log(1-y)) \nonumber
\end{align}
For $R > 0$, define
\begin{align}
&F_R(x,y) := \min(F(x,y),R) \\
&F_{R,N}(x,y) := \min(F_{N}(x,y),R)
\end{align}
Note that $F_{R,N} \to F_R$ uniformly as $N \to \infty$. Moreover, for each $R$, $F_R(x,y)$ is bounded and continuous on $[0,1]^2$. Then, since $F_R(x,y) \to F(x,y)$ monotonically as $R \to \infty$, the functional $J: \mathcal{M}([0,1]) \to \mathbb{R}$ defined by
\[ 
J(\mu) := \int\limits_0^1\int\limits_0^1 F(x,y) d\mu(x) d\mu(y) 
\]
is lower-semicontinuous. Similarly, $F_{N}(x,y)$ also defines a lower-semicontinuous functional $J_{N}$. Thus, there exist unique \citep{totik1994weighted} $\mu_0, \mu_{N} \in \mathcal{M}([0,1])$ having compact support such that
\begin{equation}
\int\limits_0^1\int\limits_0^1 F(x,y) d\mu_0(x) d\mu_0(y) = \inf_{\mu \in \mathcal{M}([0,1])} J(\mu)
\end{equation}
and
\begin{equation}
\int\limits_0^1\int\limits_0^1 F_{N}(x,y) d\mu_{N}(x) d\mu_{N}(y) = \inf_{\mu \in \mathcal{M}([0,1])} J_{N}(\mu)
\end{equation}
Finally, note that $I(\mu) = J(\mu) + B$.

\begin{restatable}{lemma}{lemmaone}\label{lemma1}
The sequence $\{\mu_{N}\}$ is tight and 
\begin{equation*}
    \begin{gathered}
    \int\limits_0^1\int\limits_0^1 F(x,y) d\mu_0(x) d\mu_0(y) \leq \\
    \liminf_{N \to \infty} \int\limits_0^1\int\limits_0^1 F_{N}(x,y) d\mu_{N}(x) d\mu_{N}(y)
\end{gathered}
\end{equation*}
\end{restatable}
\textit{Proof}: see Appendix~\ref{appendix1:lemma1proof} \qedsymbol{}

\begin{restatable}{lemma}{lemmatwo}\label{lemma2}
For the normalising constant $\mathcal{B}_{N_1,N_2, \zeta, M}$, we have
\begin{equation*}
    \limsup_{N \to \infty} \frac{1}{N^2} \log \mathcal{B}_{N_1,N_2, \zeta, M} \leq - \int\limits_0^1\int\limits_0^1 F(x,y) d\mu_0(x) d\mu_0(y)
\end{equation*}
\end{restatable}
\textit{Proof}: see Appendix~\ref{appendix1:lemma2proof} \qedsymbol{}

\begin{restatable}{lemma}{lemmathree}\label{lemma3}
For every $\mu \in \mathcal{M}([0,1])$ we have
\begin{equation*}
    \begin{gathered}
        \inf_{G} \left( \limsup_{N \to \infty} \frac{1}{N^2} \log (\mu_{\bm U_M}(G)) \right) \leq \\
        - \int\limits_0^1\int\limits_0^1 F(x,y) d\mu(x)d\mu(y) - \liminf_{N \to \infty} \frac{1}{N^2} \log \mathcal{B}_{N_1,N_2, \zeta, M}
    \end{gathered}
\end{equation*}
where $G$ runs over neighbourhoods of $\mu$ in the weak topology.
\end{restatable}
\textit{Proof}: see Appendix~\ref{appendix1:lemma3proof} \qedsymbol{}

\begin{restatable}{lemma}{lemmafour}\label{lemma4}
\begin{equation*}
    \liminf_{N \to \infty} \frac{1}{N^2} \log \mathcal{B}_{N_1,N_2, \zeta, M} \geq - \iint F(x,y) d\mu_0(x)d\mu_0(y)
\end{equation*}
and for every $\mu \in \mathcal{M}([0,1])$
\begin{equation*}
    \begin{gathered}
    \inf_G \left( \liminf_{N \to \infty} \frac{1}{N^2} \log \mu_{\bm U_M}(G)\right) \geq \\
    - \iint F(x,y)d\mu(x)d\mu(y) - \limsup_{N \to \infty} \frac{1}{N^2} \log \mathcal{B}_{N_1,N_2, \zeta, M}
        \end{gathered}
\end{equation*}
where $G$ runs over neighbourhoods of $\mu$ in the weak topology.
\end{restatable}
\textit{Proof}: see Appendix~\ref{appendix1:lemma4proof} \qedsymbol{}

Note that this result, along with a previous lemma, implies that the limit 
\[ 
\lim_{N \to \infty} \frac{1}{N^2} \log \mathcal{B}_{N_1,N_2, \zeta, M} := B < \infty 
\]
exists.  

Noting again that $I(\mu) = J(\mu) +B$, then the preceding Lemmas \ref{lemma1}-\ref{lemma4} imply that \eqref{LDP1} and \eqref{LDP2} hold. Moreover, we see that $I(\mu_0) = 0$ and $I$ is strictly convex, so that it is the only measure with this property.

In summary, we have shown that the empirical spectral measure associated to the distribution \eqref{eq:Eigen_Beta} satisfies a large deviation principle with a unique minimiser. Then, by \cite{RAS2015LDPgibbs}, we see that there is a unique Gibbs measure whose conditional distributions are given by \eqref{eq:Eigen_Beta}.

\section{Hierarchical Mixture Model}\label{sec:model-and-algorithm}

In this Section we show how to use the Coulomb priors in the context of mixture models. In Section \ref{sec:repulsive_priors}, we introduce a repulsive effect on the latent parameters $\bm \theta$ by specifying $P_0(\bm \theta)$ from the joint distribution of the eigenvalues of random matrices. As shown earlier, such distributions induce sparsity in the mixture by repelling the cluster centres from each other, thus allowing flexible configurations of the parameters. Due to the fact that the support of each eigenvalue is not the whole real line ($\mathbb{R}^+$ in the Wishart case and $(0,1)$ in the Beta one), each $\bm \theta$ is modelled on an appropriate scale depending on the application at hand. For instance, in a setting where the latent parameters represent the locations of real-valued observations, these can be transformed via suitable link functions, e.g. $\log$ or logit. Notice that these laws preserve the repulsion term in their distribution, as it is easily seen by applying a change of variable transformation.

To specify a prior distribution on the weights $\bm w$ and the number of components $M$ of the mixture, we follow the approach of \cite{argiento2022infinity}. We assume the weights are derived by normalising a finite point process, which requires the unnormalised weights $\bm S = \left(S_1, \dots, S_M\right)$ to be infinitely divisible. \cite{argiento2022infinity} show that a finite mixture model is simply a realisation of a stochastic process whose dimension is random and has an infinite dimensional support. This leads to flexible distributions for the weights of the mixture, which also allows for efficient posterior computations. They refer to their construction as a \textit{normalised independent finite point process}. Here, we assume that $S_1, \dots, S_M$ are i.i.d. from a Gamma distribution. This construction is the finite-dimensional version of the class of normalised random measures with independent increments \citep{regazzini2003distributional, james2009posterior}. The resulting model is the following:
\begin{align}\label{eq:Application_mixture}
	&\bm y_i \mid z_i, \bm \theta, M \sim f\left(\bm y_i \mid \bm \theta_{z_i}\right) \quad i = 1, \dots, n \nonumber \\
	&\bm \theta_1, \dots, \bm \theta_M \mid M \sim P_0(\bm \theta) \nonumber \\
	&\mathbb{P}\left(z_i = h \mid \bm S, M\right) \propto S_h \quad h = 1, \dots, M \\
	&S_1, \dots, S_M \mid M \iid \text{Gamma}\left(\gamma_S, 1\right) \nonumber \\
	&M \sim \text{Poi}_{1}\left(\Lambda\right) \nonumber 
\end{align}
where $\text{Gamma}(a, b)$ represents the Gamma distribution with mean $a/b$ and variance $a/b^2$, while $\text{Poi}_{1}(\Lambda)$ is the Poisson distribution shifted by one unit, i.e. the random variable $M - 1$ is Poisson with mean $\Lambda$. Note that upon marginalisation of the allocation variables $z_i$, for $i = 1, \dots, n$, we obtain the normalised weights $w_m = S_m/\sum_j S_j$. Here we highlight the distinction between the number, $M$, of components in a mixture, i.e. of possible clusters/sub-populations, and the number, $K_n$, of clusters, i.e. the number of allocated
components (components to which at least one observation has been assigned). The latter quantity $K_n\leq M$ can only be estimated a-posteriori.

In the original setting proposed by \cite{argiento2022infinity}, the updates of the location parameters $\bm \theta$ for the allocated and non-allocated components are performed separately, thanks to the independence assumption between these two components of the process. In our scenario, the parameters $\bm \theta$ are not i.i.d., and therefore their update require a different technique. Just as in \cite{argiento2022infinity}, we use the information provided by the vector of allocation variables $\bm z$ to identify which components have been allocated and which ones have not. Thus, we can split the vector of parameters as $\bm \theta = \left(\bm \theta^{(a)}, \bm \theta^{(na)}\right)$ to highlight their allocation status. Update of the allocated components of the mixture is performed via a Metropolis-Hastings step. On the other hand, the update of the non-allocated components is tackled with a Metropolis-Hastings birth-and-death move, as proposed by \cite{geyer1994simulation} and implemented in the context of repulsive mixtures based on DPPs in \cite{beraha2022mcmc}. One of the advantages of the proposed class of prior distributions for the parameters $\bm \theta$ is their tractability. Indeed, all distributions described in Section \ref{sec:repulsive_priors} are known in closed form and have tractable normalising constants. This allows to perform inference a-posteriori for the hyperparameters driving these distributions, a feature that can be very appealing in clustering analysis.
In Supplementary Material, we detail the MCMC algorithm.

\subsection{Examples}\label{sec:examples}

In this Section, we first present the popular Ising model used in the thermodynamic literature to describe the behaviour of polarised particles on the 2-dimensional lattice $\mathbb{Z}^2$, highlighting an interesting behaviour related to the uniqueness of the associated Gibbs measures. Then, 
we demonstrate the proposed approach on simulated and benchmark data. We explore the partitions induced a-posteriori by the Coulomb priors introduced in Section \ref{sec:repulsive_priors}, and compare them with alternative repulsive mixture models proposed in the literature, as well as the mixture model with a random number of components proposed by \cite{argiento2022infinity}, which can be recovered from our model when no repulsion is assumed. The latter can be easily implemented using the R package \texttt{AntMAN} \citep{antmanpackage}. We assess the clustering results by comparing the estimates of the partitions obtained by minimising the Binder loss function \citep{binder1978bayesian}, while density estimation is assessed by computing the predictive densities.

\paragraph{The 2-dimensional Ising model}\label{sec:IsingExample}

Consider $M$ particles on a 2-dimensional lattice. Let $\bm \theta  = \left( \theta_1, \dots, \theta_M \right)$ be the vector of spins of the $M$ particles under study, so that $\theta_i \in \{-1, +1\}$, where $-1$ and $+1$ indicate negative and positive spins, respectively. The Hamiltonian of the Ising model is given by:
\begin{equation}\label{eq:Ising_Hamiltonian}
	\mathcal{H}(\bm \theta \mid \zeta, h) = h \sum_{i = 1}^M\theta_i + \zeta \sum_{i \sim j} \theta_i \theta_j
\end{equation}
where $\zeta > 0$ is the inverse temperature and $h \in \mathbb{R}$ quantifies the influence of a magnetic field acting on the particles and in particular on the values of the spins $\bm \theta$. Notice how the expression of the Hamiltonian of the Ising model is of the same form as the Gibbs potential in Eq.~\eqref{eq:EnergyIntro}, where $\psi_1(\theta) = \theta$ and $\psi_2(\theta_i, \theta_j) = \theta_i \theta_j$. In the standard Ising model, the particles are assumed to interact only with their immediate neighbours, as indicated by the notation $i\sim j$ in Eq.~\eqref{eq:Ising_Hamiltonian}. From Eq.~\eqref{eq:Ising_Hamiltonian} we obtain the probability of observing a specific configuration of the particle spins $\bm \theta$:
\begin{equation}\label{eq:Ising_Boltzmann}
	p_M\left(\bm \theta \mid \zeta, h\right) = \frac{e^{-\mathcal{H}\left(\bm \theta \mid \zeta, h\right)}}{Z^{Ising}_n\left(\zeta, h\right)}
\end{equation}
with $Z^{Ising}_n\left(\zeta, h\right) = \sum_{\bm \theta \in \{-1, +1\}^n}e^{-\mathcal{H}\left(\bm \theta \mid \zeta, h\right)}$ representing the normalising constant (i.e., the partition function). Computations under the Ising model are challenging due to the intractability of $Z^{Ising}_n\left(\zeta, h\right)$, with exception of the case $d = 1$ which has a closed form solution.

The behaviour of the Ising model at the thermodynamic limit, i.e. for $M \rightarrow +\infty$, is of particular interest as it depends on the dimension of the lattice $d$ and on the parameters $\zeta$ and $h$. Theorem 3.25 by \cite{friedli_velenik_2017} gives an overview of the limiting distributions that arise under different Ising model specifications. When $d \geq 2$, a first order phase transition may occur, allowing the set of limiting Gibbs measures to contain more than one element. In particular, when the external magnetic field is absent ($h = 0$) there exist a \textit{critical value} $\zeta_c = \zeta_c(d)$ of the inverse temperature above which two possible states are obtained for $M \rightarrow +\infty$. On the other hand, when $h \neq 0$, the behaviour of the particles is dominated by the external magnetic field.

We illustrate this result numerically for the 2-dimensional Ising model. We specify a univariate Gaussian mixture model with two components, indexed by an underlying Ising model, as follows:
\begin{align}\label{eq:Ising_mixture}
    &y_i \sim \text{N}\left(y_i \mid \theta_i, 1 \right), \quad i = 1, \dots, M \nonumber \\
    &\theta_1, \dots, \theta_M \mid \zeta, h \sim \text{Ising}^2_M\left(\bm \theta \mid \zeta, h\right) 
\end{align}
where $\text{Ising}^2_M\left(\bm \theta \mid \zeta, h\right)$ indicates that the $M$ particle spins $\bm \theta$ are distributed according to the 2-dimensional Ising model specified above. We simulate $M = n \times n$ observations from model \eqref{eq:Ising_mixture}, on a square lattice of size $n \in \{5,10,15,20\}$, with equal proportion $\pi_{\theta}$ of positive and negative spins. We fit the above model for varying values of the parameters $\zeta \in \{0, 0.1, 0.25, \zeta_c, 0.5, 0.75, 1, 10, 100\}$, with $\zeta_c = \log\left(1 + \sqrt 2\right)/2 \approx 0.44$, and $h \in \left[-5,5\right]$. We run the MCMC algorithm for length 1500 iterations, of which the first 1000 are discarded as burn-in. We start the MCMC chains from different spin configurations $\left(\theta_1,\dots,\theta_M\right)$ sampled from independent  Bernoulli distributions  over $\{-1,1\}$ with varying success probability (i.e., expected proportion of spins equal to 1), $\pi_{\theta} \in \{0, 0.25, 0.5, 0.75, 1\}$. 

As we show in Figure \ref{fig:Ising_pitheta}, when $h = 0$, the initial configuration is of particular importance as it determines the limiting distribution, when non-uniqueness is encountered. Specifically, for values of the inverse temperature $\zeta$ greater than the critical $\zeta_c$, often referred to as the \textit{low temperatures} case, the distribution of the spins is the same as the one used to sample the initial configuration.
\begin{figure}[ht]
	\centering
	\subfloat[$n = 5$]{\includegraphics[width=0.5\linewidth]{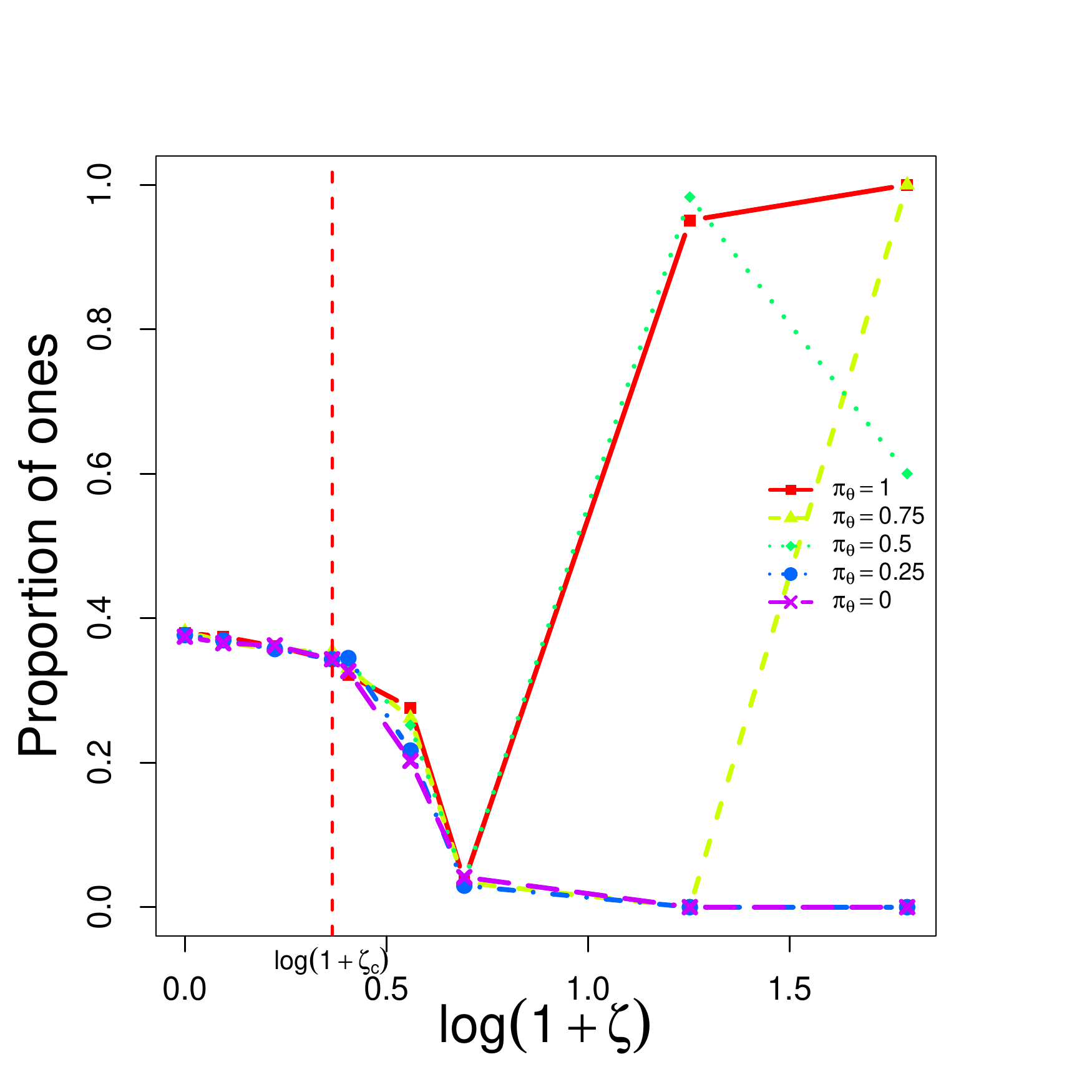}}
	\subfloat[$n = 10$]{\includegraphics[width=0.5\linewidth]{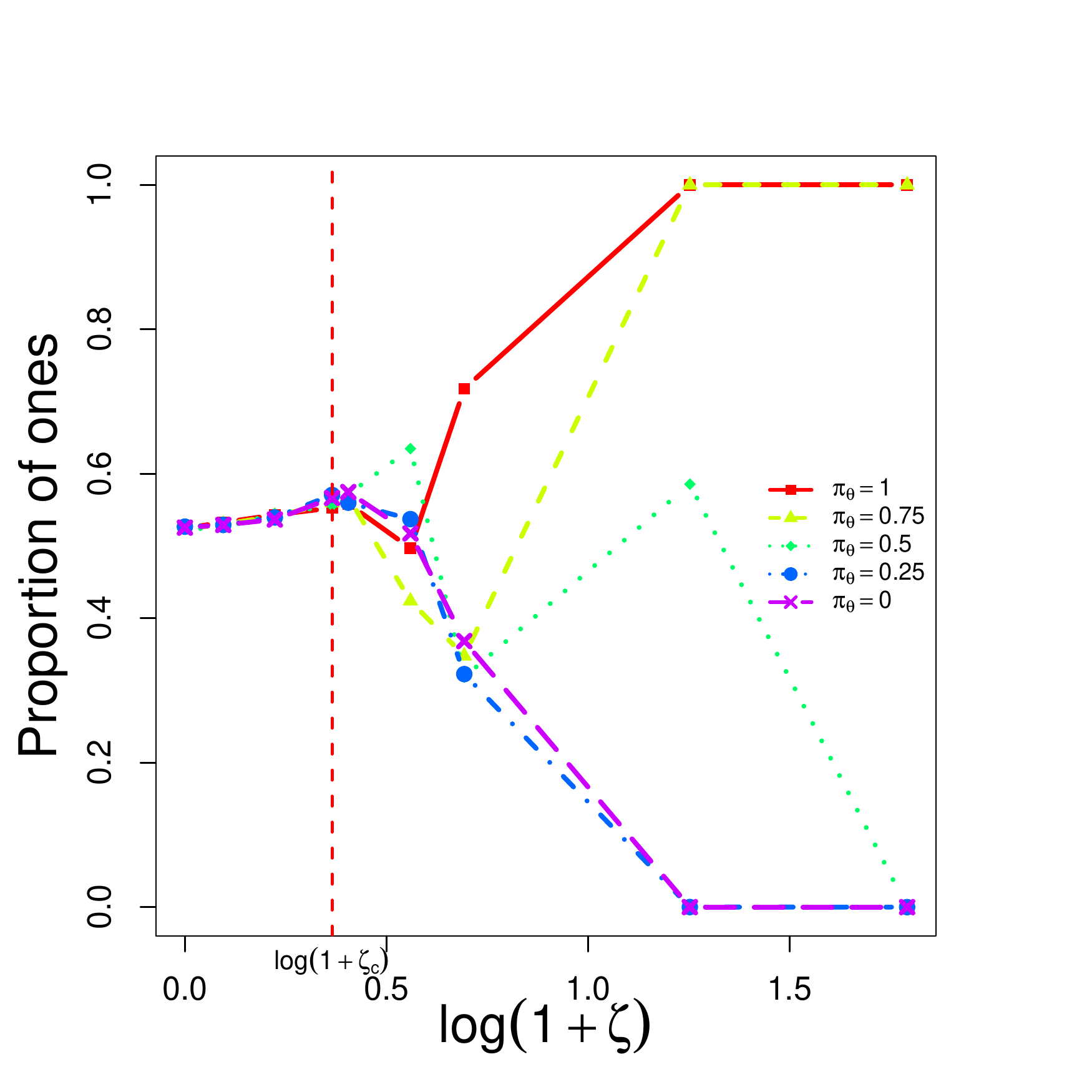}} \\
	\subfloat[$n = 15$]{\includegraphics[width=0.5\linewidth]{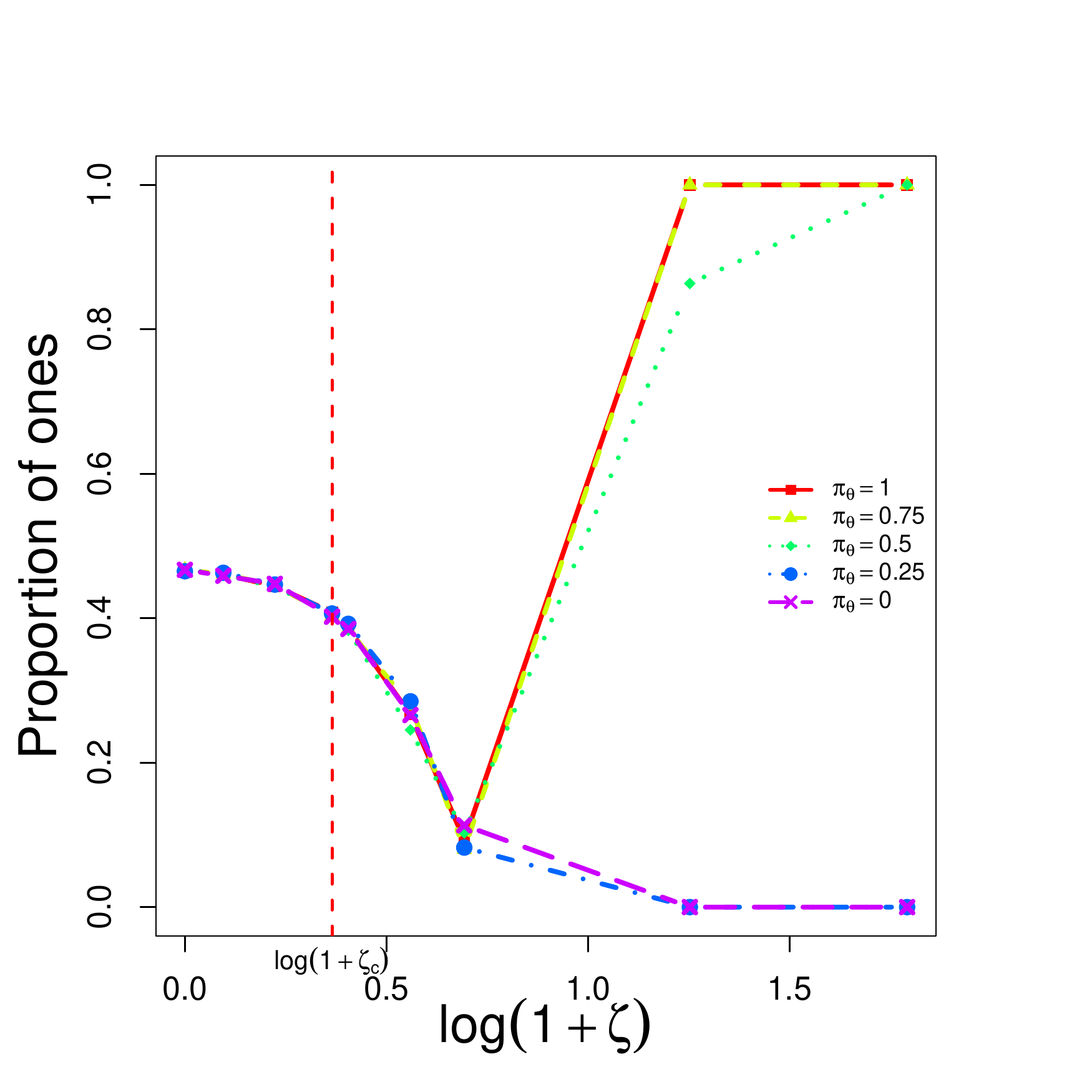}}
        \subfloat[$n = 20$]{\includegraphics[width=0.5\linewidth]{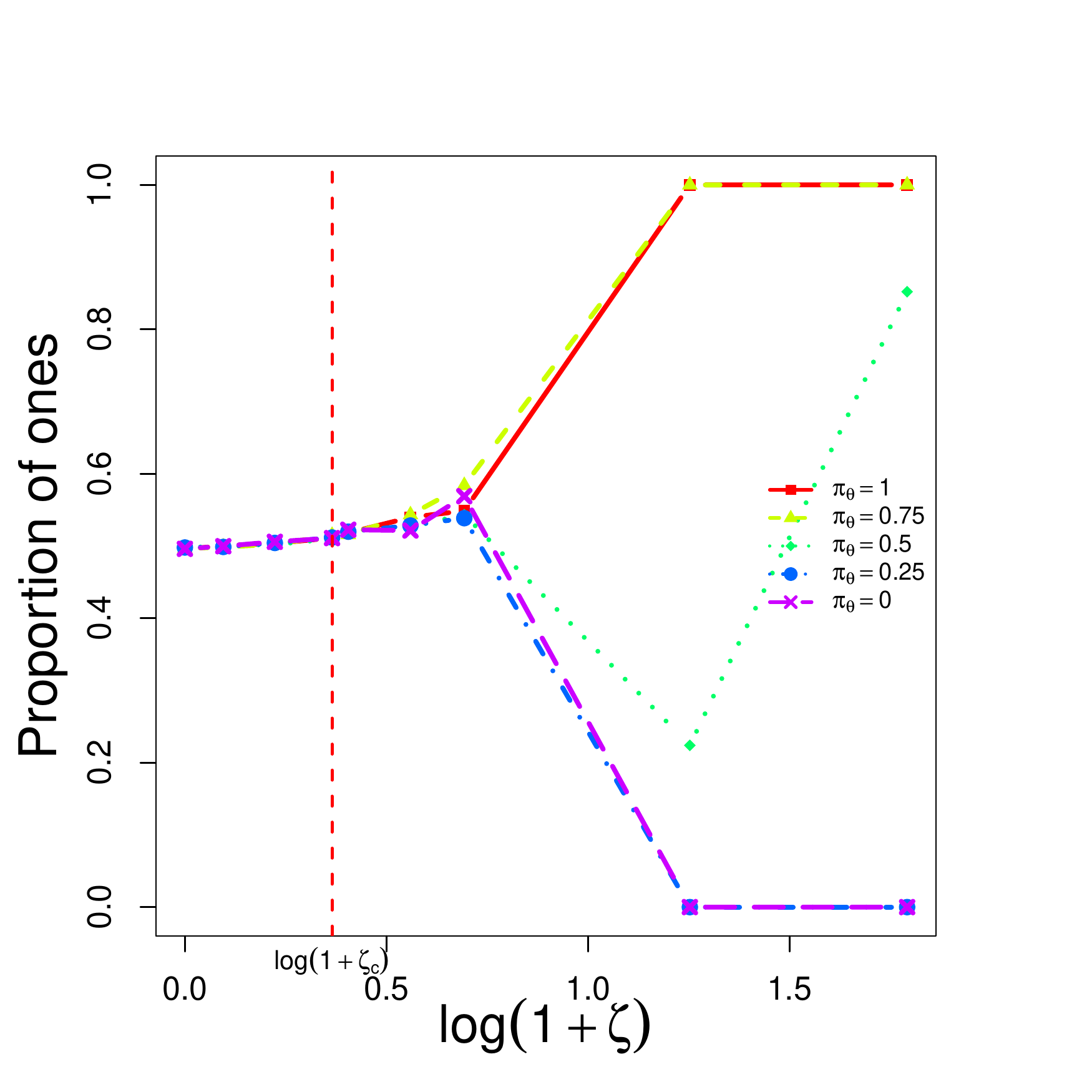}}
	\caption{Ising model ($d = 2$). Posterior mean of the proportion of positive spins in the sample, for different values of $\zeta$ and $h = 0$. Each line corresponds to a different initialisation setting, i.e. $\pi_{\theta}$.}
	\label{fig:Ising_pitheta}
\end{figure}

Another important quantity monitored in the context of the Ising mode is the \textit{magnetisation} of the particle system, defined as the average spin value among the $M$ particles, i.e. $\mathcal{M}_M = \frac{1}{M}\sum_{i = 1}^M \theta_i$. As it can be observe from Figure \ref{fig:Ising_Mn}, the posterior mean of the magnetisation strongly depends on the value of the external field $h$. This is in agreement with the results of Theorem 3.25 by \cite{friedli_velenik_2017}.

\begin{figure}[ht]
	\centering
	\subfloat[$\pi_{\theta} = 0.25$]{\includegraphics[width=0.35\linewidth]{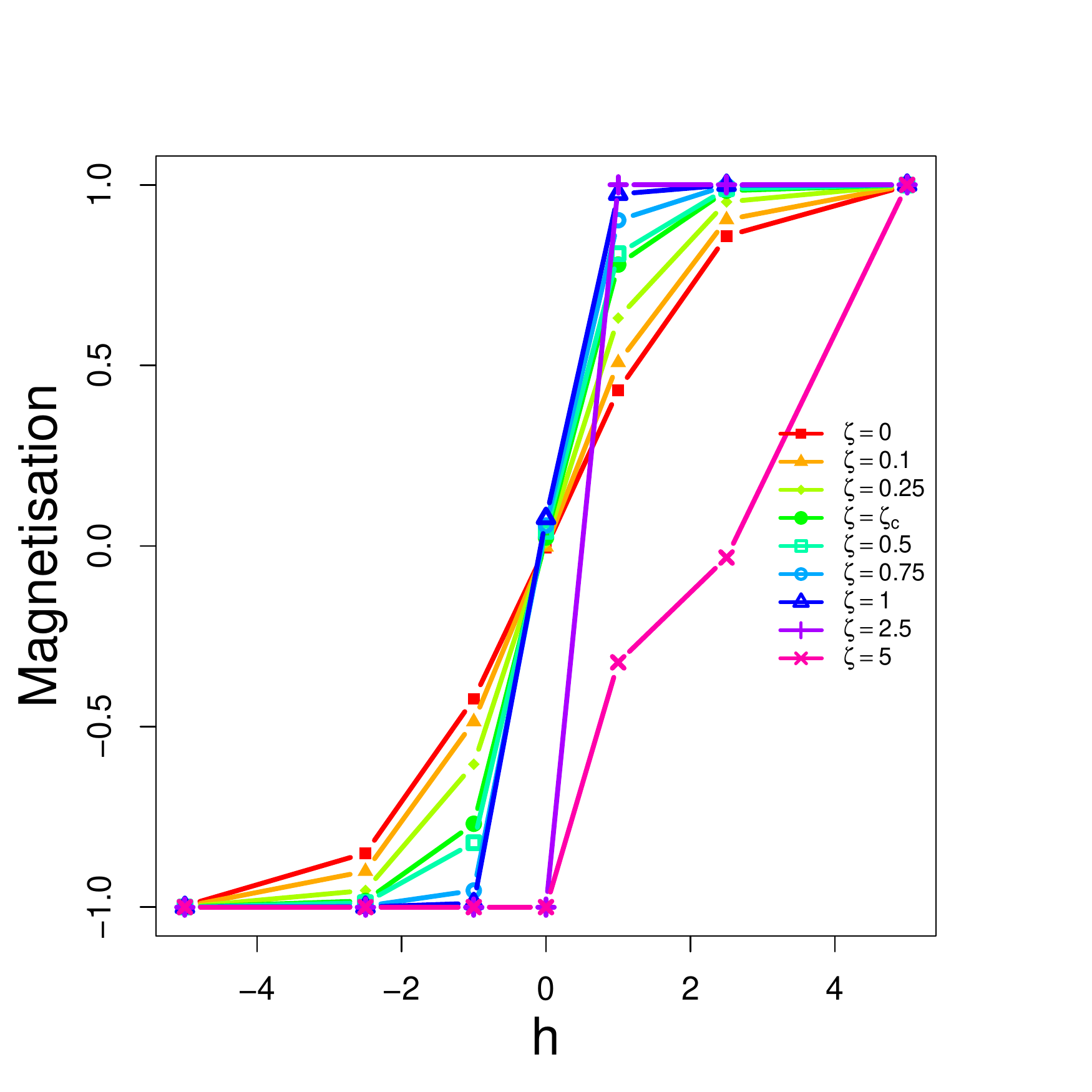}}
	\subfloat[$\pi_{\theta} = 0.5$]{\includegraphics[width=0.35\linewidth]{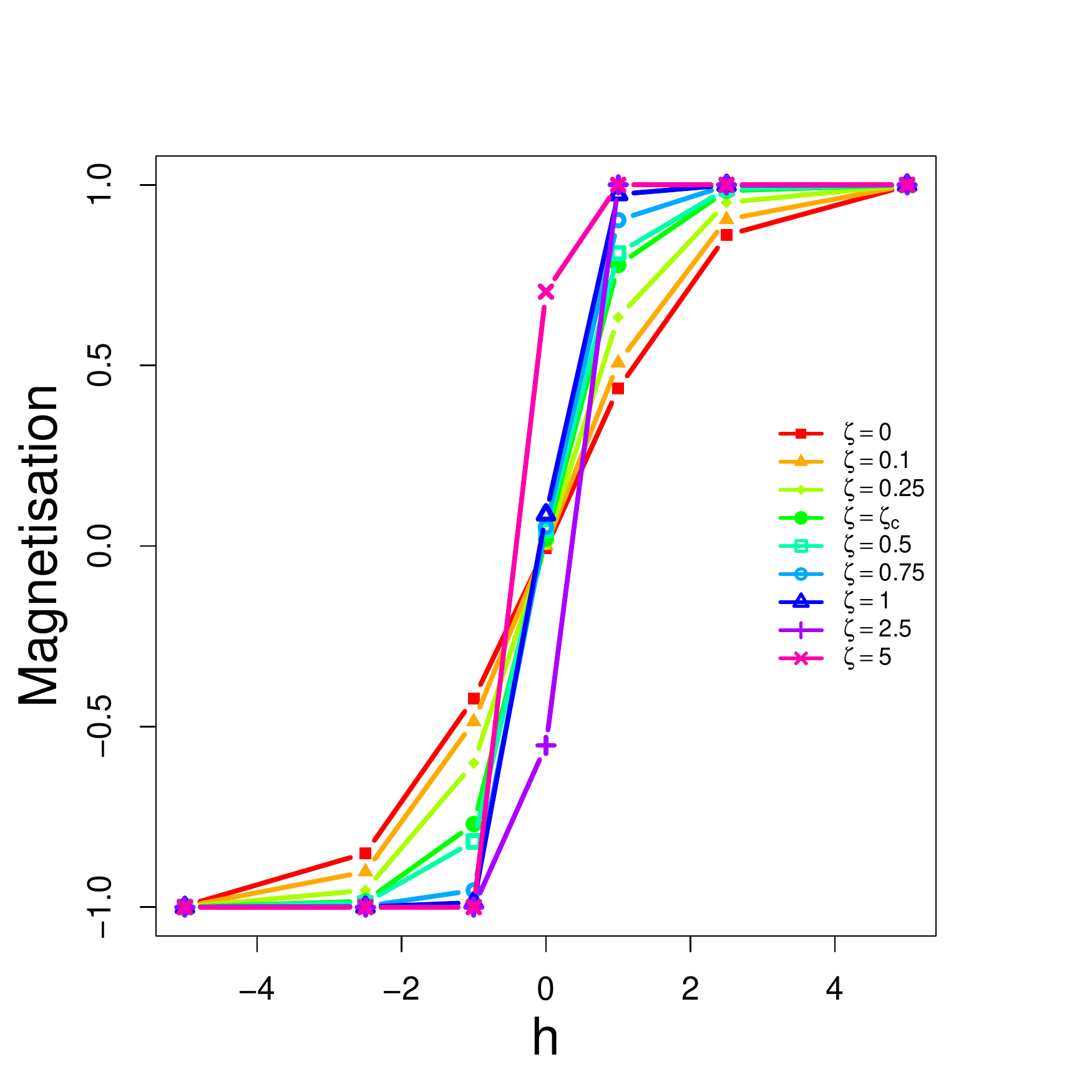}}
        \subfloat[$\pi_{\theta} = 0.75$]{\includegraphics[width=0.35\linewidth]{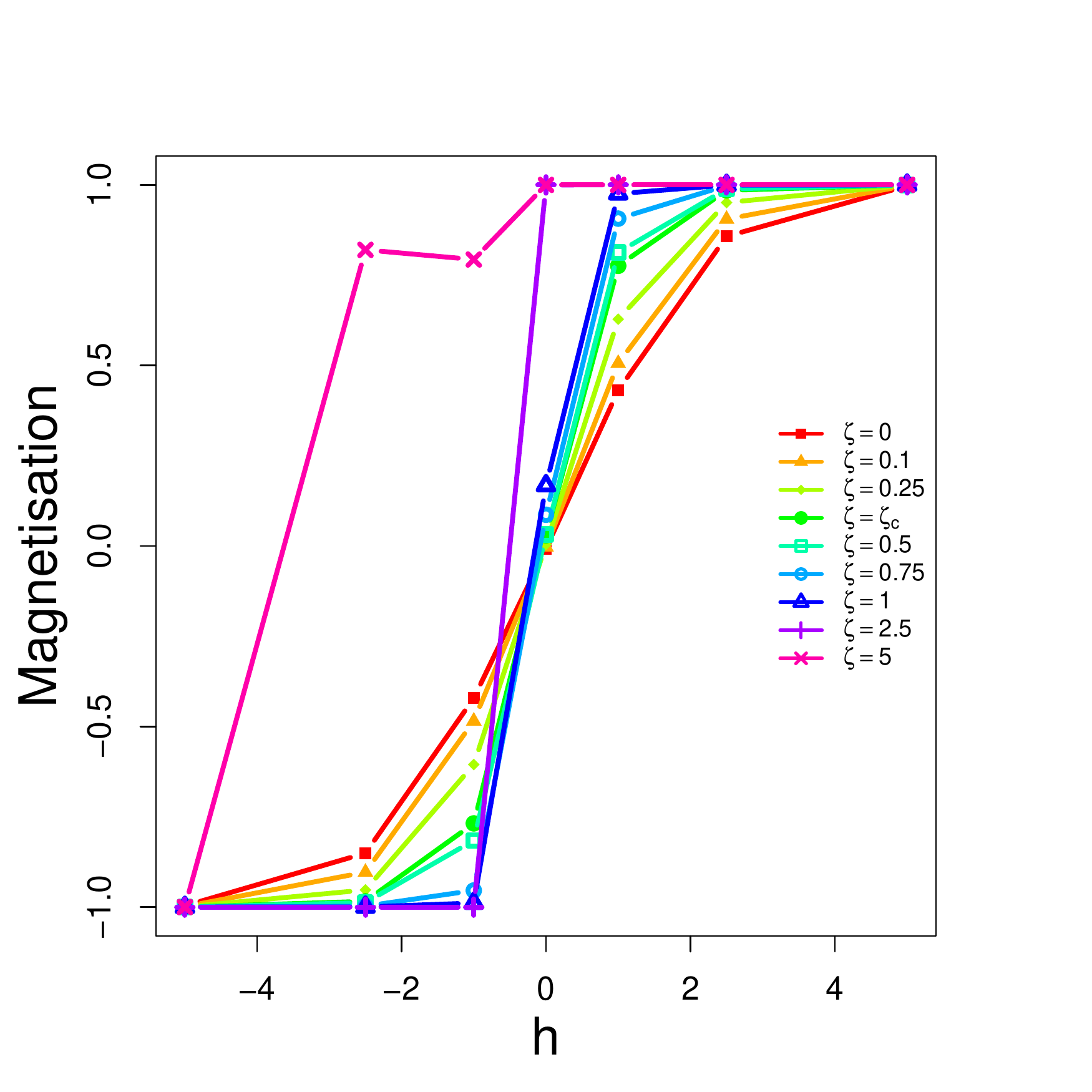}}
	\caption{Ising model ($d = 2$, $n = 20$). Posterior mean of the magnetisation of the system $\mathcal{M}_M$, for different values of $\zeta$ and $h$. Each panel corresponds to a different initialisation setting, i.e. $\pi_{\theta}$.}
	\label{fig:Ising_Mn}
\end{figure}

These results are a consequence of the non-uniqueness of the Gibbs measure. We conclude by noting that the approach of \cite{petralia2012repulsive,quinlan2021class} does not have a corresponding infinite-dimensional measures, as their prior process is not Kolmogorov consistent, and they have not shown there is an infinite  dimensional object with their distributions as conditionals.

\paragraph{Mixture of binomial distributions}

We illustrate the performance of the proposed model on data simulated from a mixture of binomial distributions.
Specifically, $n = 150$ observations are sampled from the following mixture with five components:
\begin{align*}
    &f\left(y_i \mid \bm \theta\right) = \sum_{m = 1}^5 w_m \text{Bin}\left( y_i \mid 15, \pi_m \right) \quad i = 1, \dots, n \\
    & \bm w \sim \text{Dir}_M\left(\bm w \mid 1, \dots, 1\right) \\
    & \bm \pi = \text{logit}^{-1}\left(-5,-2.5,0,2.5,5\right)
\end{align*}
where $\text{Bin}\left( y \mid R, \pi \right)$ indicates the Binomial distribution with $R$ trials and success probability $\pi$, while $\text{Dir}_M\left(\bm w \mid a, \dots, a\right)$ represents the $M$-dimensional Dirichlet distribution with shape parameter vector $\left(a, \dots, a\right)$ and $a>0$.

To investigate the results obtained by fitting misspecified mixture models to the simulated data, we compare several repulsive mixture models in a finite mixture setting, i.e. when $M$ is fixed. In particular, we implement the models proposed by \cite{quinlan2021class}, \cite{petralia2012repulsive} and \cite{xu2016bayesian}, as well as the one proposed in this paper, with suitable prior distributions specified for the parameter vector of interest $\bm \pi$. We opt for a Beta kernel density for the methods proposed by \cite{quinlan2021class} and \cite{petralia2012repulsive}, while use Gaussian kernels for the specification of the DPP in \cite{xu2016bayesian}, and impose an inverse-logit link on  the parameters of interest $\pi_m$, for $m = 1, \dots, M$. To fix the hyperparameters of the competitor models, we follow the authors' guidelines presented in the corresponding works. Finally, the Jacobi-Coulomb prior distribution is used when fitting our model. 

First, we assess posterior inference on the partition estimated by the different models. We show in Figure \ref{fig:BinMix_Kn} the posterior distribution of the number of clusters $K_n$ under different model specifications. We observe right skewness in most scenarios, and more markedly for the competitor models. In particular, the proposed model induces partitions characterised by a lower number of clusters when compared to the model of \cite{petralia2012repulsive} and comparable results to the DDP model of \cite{xu2016bayesian} and the one of \cite{quinlan2021class}. Such difference is less evident when a large number of mixture components $M$ is specified.
This behaviour is confirmed by computing the partition minimising the Binder loss function, for which we report the corresponding number of clusters in Table \ref{tab:BinMix_Binder}. The proposed method yields partitions more robust to increasing values of $M$ and composed of less clusters.

\begin{figure}[ht]
	\centering
 	\includegraphics[width=1\linewidth]{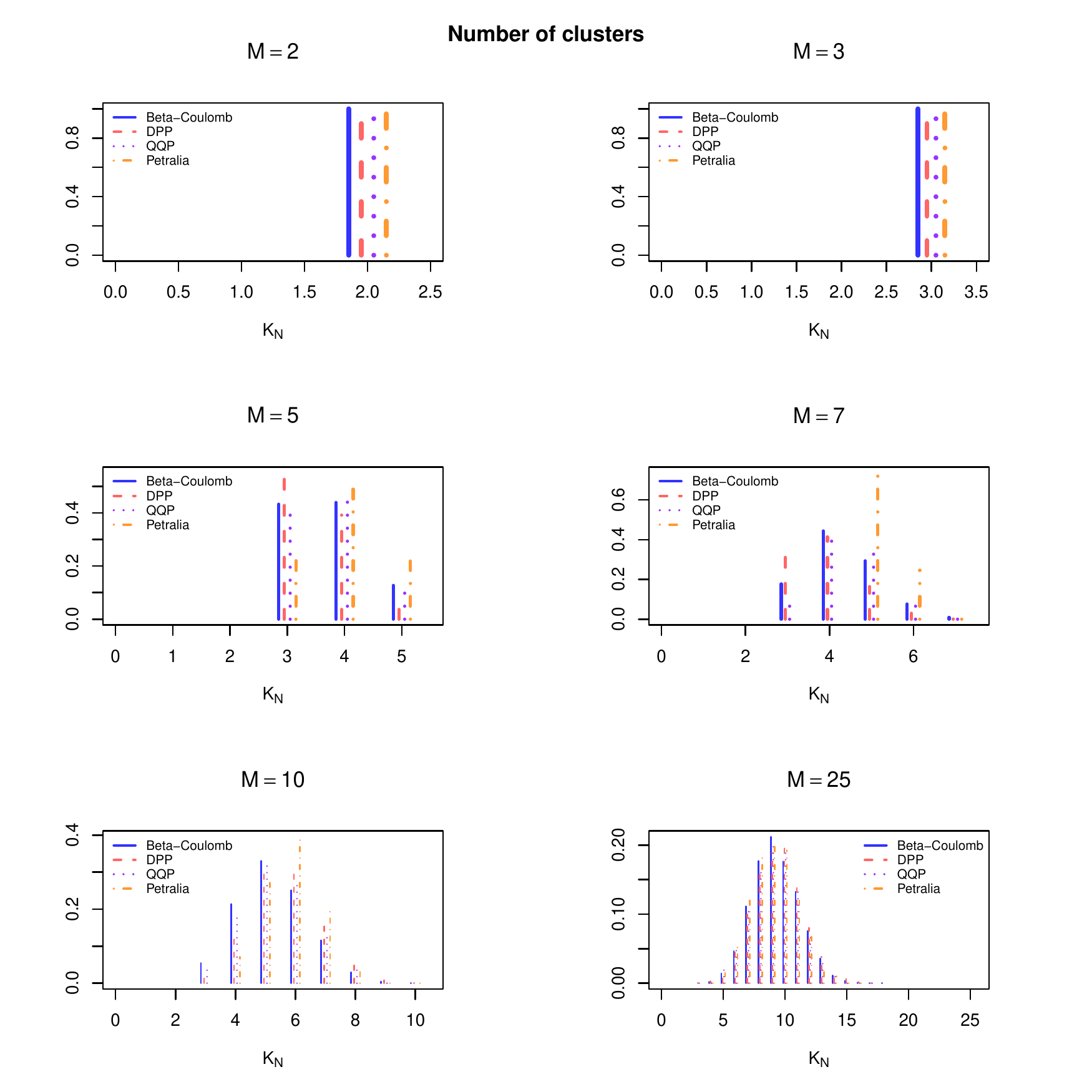}
	\caption{Binomial mixture data. Posterior distribution of the number of clusters $K_n$ for different model specifications. Each panel shows the results for a different value of $M$.}
	\label{fig:BinMix_Kn}
\end{figure}

\begin{table}
\caption{Binomial mixture data. Number of clusters for the partition estimates  obtained by minimising the Binder loss function, for increasing value of $M$ and different repulsive models.}
\label{tab:BinMix_Binder}
\begin{center}
\begin{tabular}{c|cccccc}
  $M$ & 2 & 3 & 5 & 7 & 10 & 25 \\\hline
 Beta-Coulomb & 2 & 3 & 3 & 3 & 3 & 5 \\ 
 DPP & 2 & 3 & 3 & 4 & 3 & 5 \\  
 QQP & 2 & 3 & 3 & 3 & 4 & 9 \\
 Petralia & 2 & 3 & 4 & 5 & 4 & 6
\end{tabular}
\end{center}
\end{table}

\paragraph{Air quality data}

We illustrate the proposed model on the bivariate \textit{Air Quality} dataset \citep{chambers1983graphical}, containing $n = 111$ air quality measurements, indicating the levels of New York's Ozone and Solar radiation in 1973. We fit to the Air Quality data a bivariate mixture of Gaussian distributions, with random number of components $M$ and Gaussian Coulomb prior for the location parameters. Note that each dimension of the mean vector is a-priori independent. The model is specified as follows: 
\begin{align}\label{eq:AirQuality_mixture}
	&\bm y_i \mid z_i, \bm \theta, M \sim \text{N}_2\left(\bm y_i \mid \bm \theta_{z_i}, \bm \Sigma_{z_i}\right) \quad i = 1, \dots, n \nonumber \\
	&\bm \theta_1, \dots, \bm \theta_M \mid M, \zeta \sim p_M\left(\bm \theta\right) \nonumber \\
        &\zeta \sim \text{Gamma}\left(1,1 \right) \nonumber \\
        &\bm \Sigma_1, \dots, \bm \Sigma_M \mid M \iid \text{inv-Wishart}\left(\nu, \bm \Psi \right) \nonumber \\
	&\mathbb{P}\left(z_i = h \mid \bm S, M\right) \propto S_h, \quad h = 1, \dots, M \nonumber \\
	&S_1, \dots, S_M \mid M \iid \text{Gamma}\left(\gamma_S, 1\right) \nonumber \\
	&M \sim \text{Poi}_{1}\left(\Lambda\right) \nonumber 
\end{align}
Note that the covariance matrices $\bm \Sigma_m$, for $m = 1, \dots, M$, are component-specific but modelled independently from the location parameters $\bm \theta$, as an i.i.d. sample from an inverse Wishart distribution with with $\nu = 6$ degree of freedoms and $\bm \Psi$ set equal to the 2-dimensional identity matrix. We also impose a prior distribution on the repulsion parameter $\zeta$. We compare our results with those obtained using the repulsive finite mixture model by \cite{quinlan2021class} with $M = 5$ components, and with the model implemented in the R package \texttt{AntMAN}. The latter corresponds to the the proposed model when the repulsive term is omitted, and we fix the hyperparameters $\gamma_S = \Lambda = 1$ in both models. After an initial burn-in of 100 iterations used to initialise the adaptive steps of the MCMC, all the simulations are run for a total of 7500 iterations. These include 2500 iterations discarded as burn-in and 5000 thinned to obtain a final sample of 2500 iterations, used for posterior inference. 

Figures \ref{fig:AirQuality_data_Mzeta}(a-b) show the posterior distribution of the number of clusters $K_n$ and components $M$ for the three models. The proposed approach is able to identify coarser partitions a-posteriori than the model of \cite{quinlan2021class} and the one without repulsion implemented in the \texttt{AntMAN} package. The posterior distribution of the repulsive parameter $\zeta$ is shown in Figure \ref{fig:AirQuality_data_Mzeta}(c). The partitions estimated by minimising the Binder loss function \citep{binder1978bayesian} are displayed, together with the contour plots of the predictive densities, in Figure \ref{fig:AirQuality_data_contour}. These figures show how the proposed method is able to gather in the same cluster points spread across a wider region, avoiding redundancy.
\begin{figure}[ht]
	\centering
 	\subfloat[{\tiny Number of clusters $K_n$}]{\includegraphics[width=0.35\linewidth]{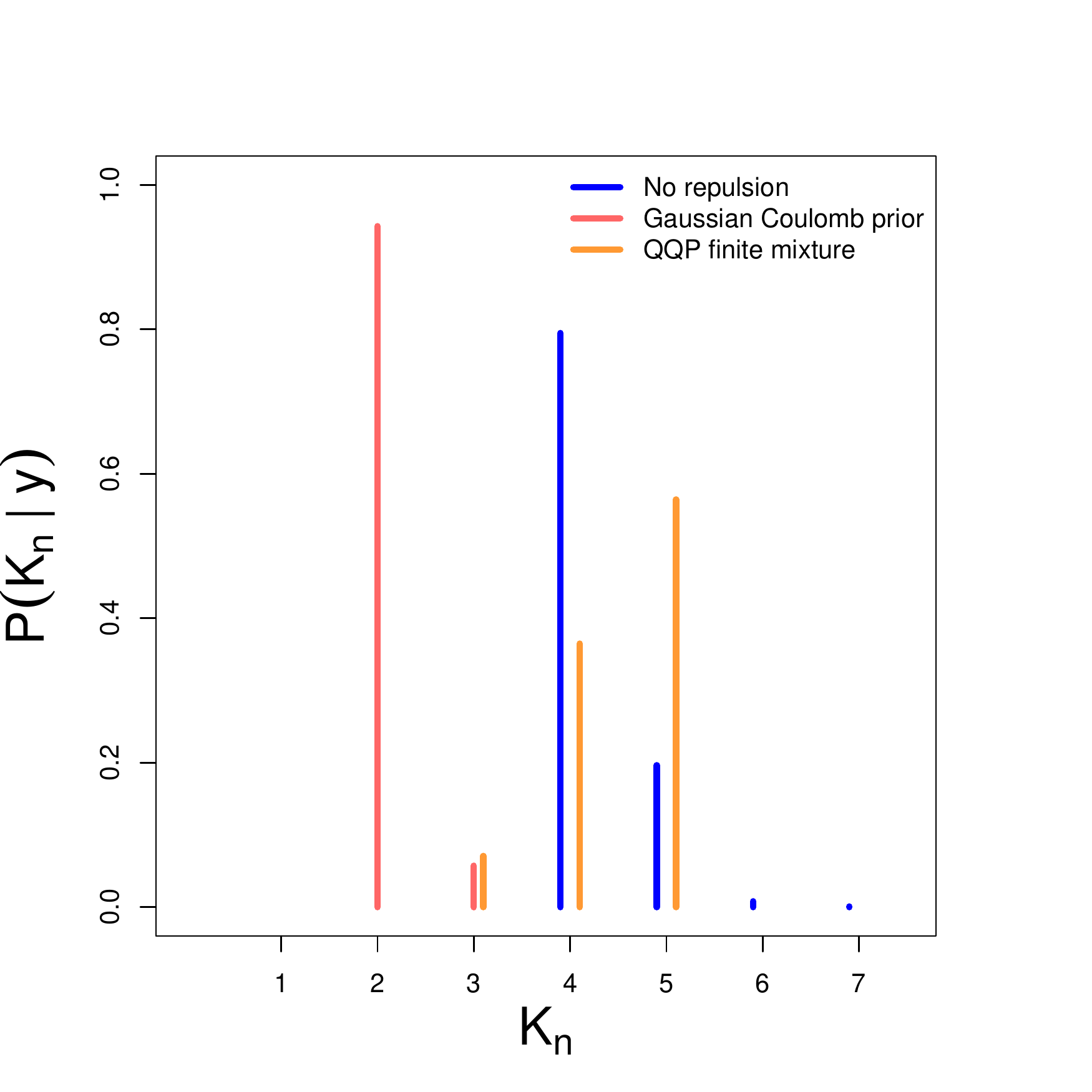}}
        \subfloat[{\tiny Number of components $M$}]{\includegraphics[width=0.35\linewidth]{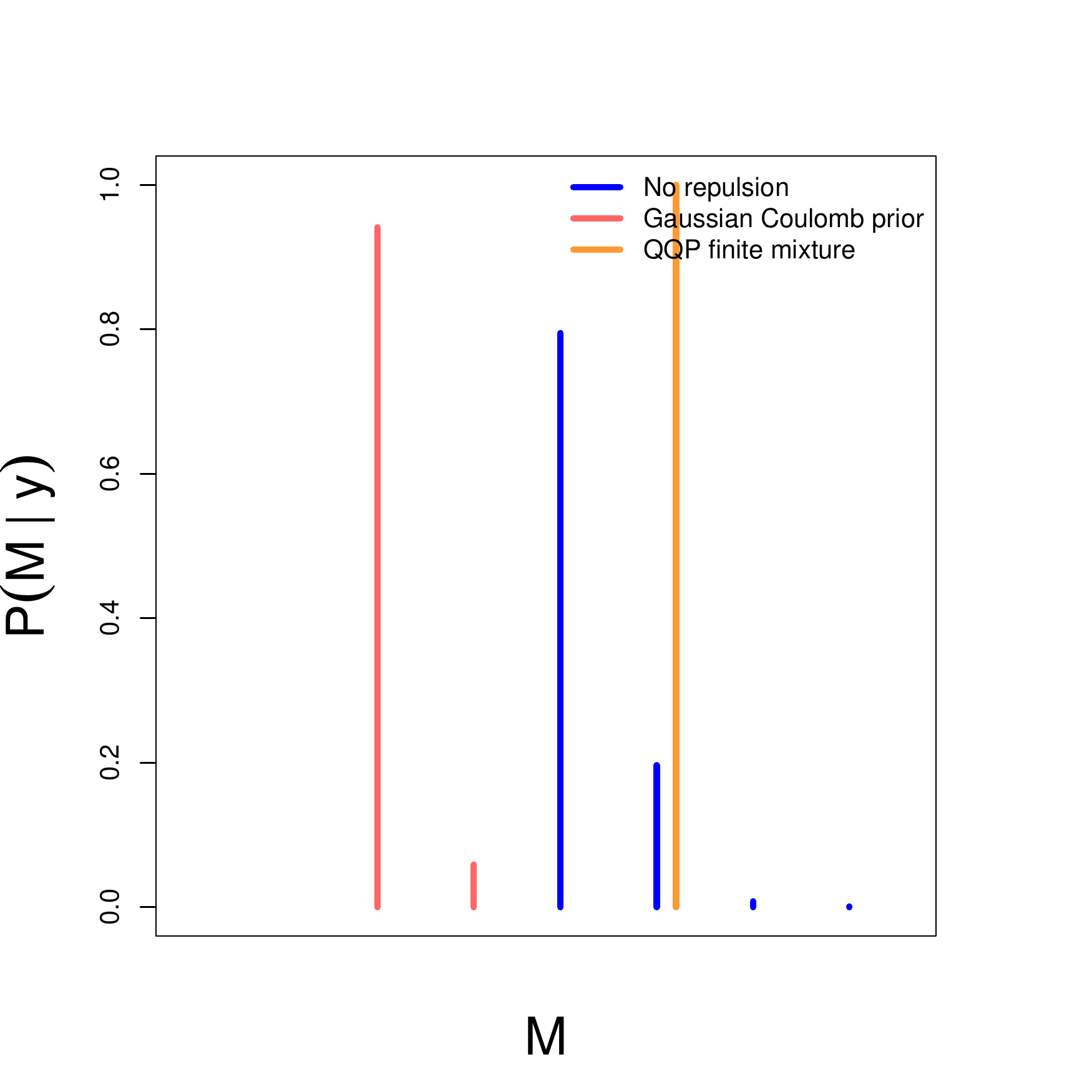}}
	\subfloat[{\tiny Repulsiveness $\zeta$}]{\includegraphics[width=0.35\linewidth]{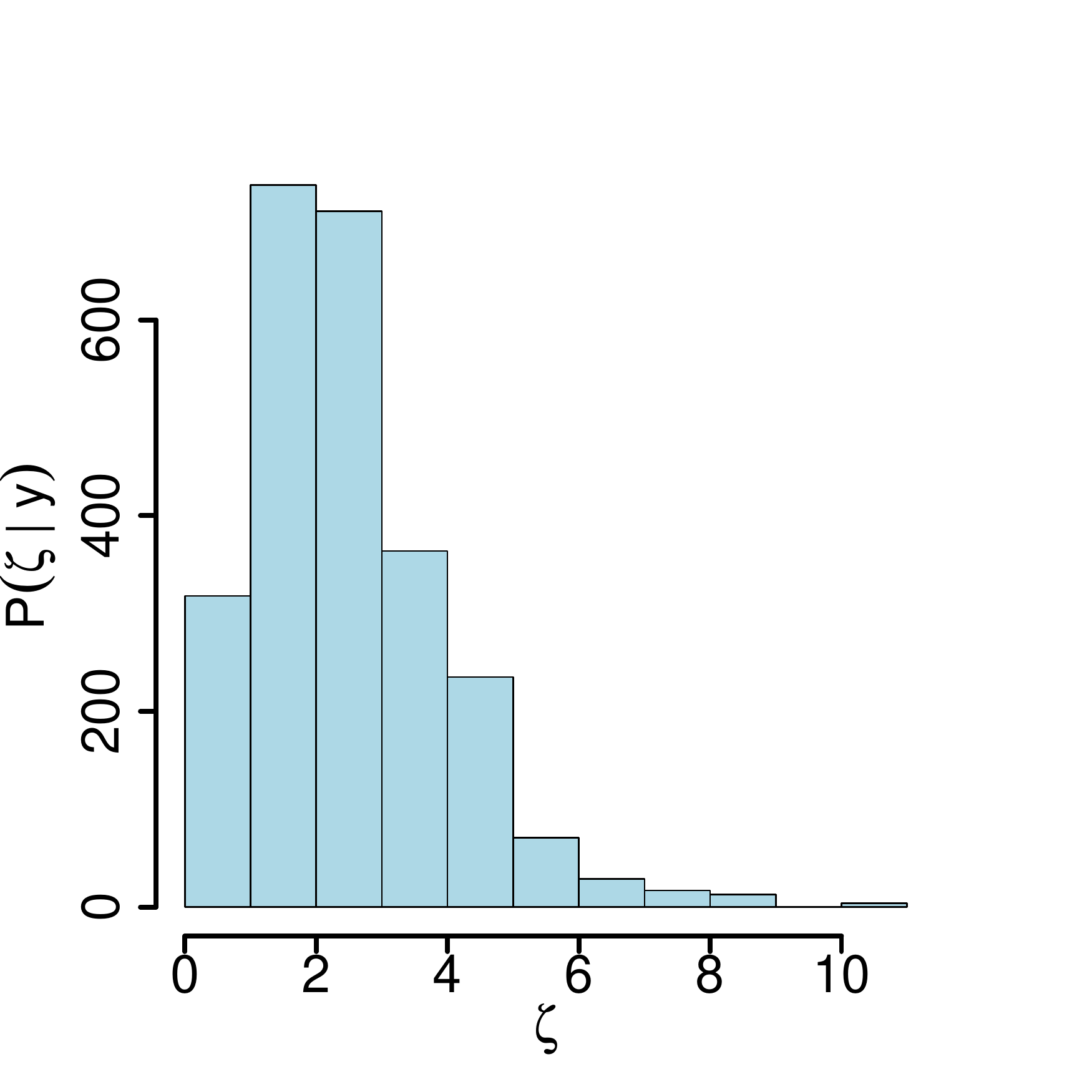}}
	\caption{Air Quality data. Posterior distribution of (a) the number of clusters $K_n$ and of (b) the number of components $M$ obtained under different modelling specifications. Panel (c) shows the posterior distribution of the repulsion parameter $\zeta$ of the Gaussian Coulomb prior.}
	\label{fig:AirQuality_data_Mzeta}
\end{figure}

\begin{figure}[ht]
	\centering
	\subfloat[{\tiny Gaussian Coulomb prior}]{\includegraphics[width=0.35\linewidth]{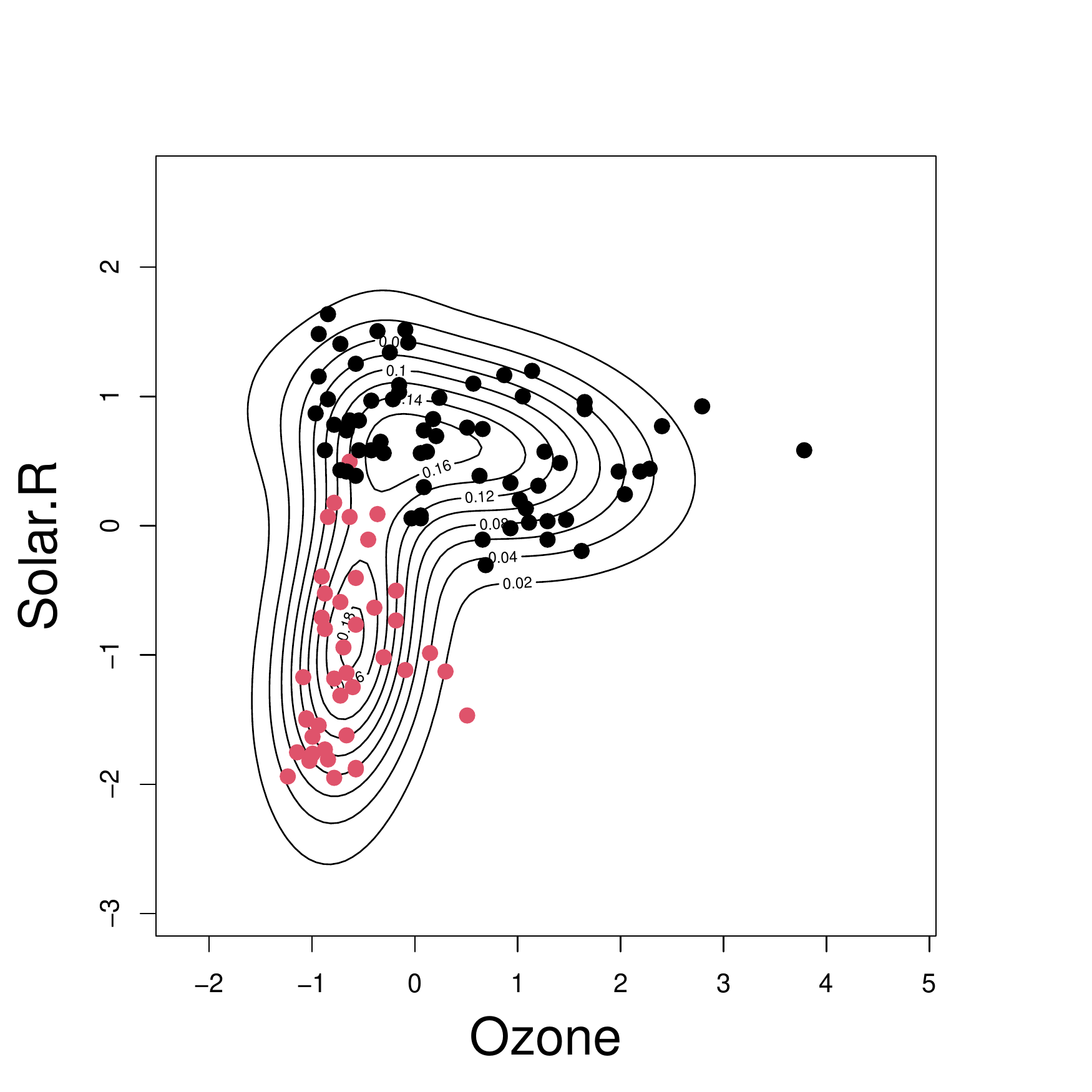}}
	\subfloat[{\tiny QQP}]{\includegraphics[width=0.35\linewidth]{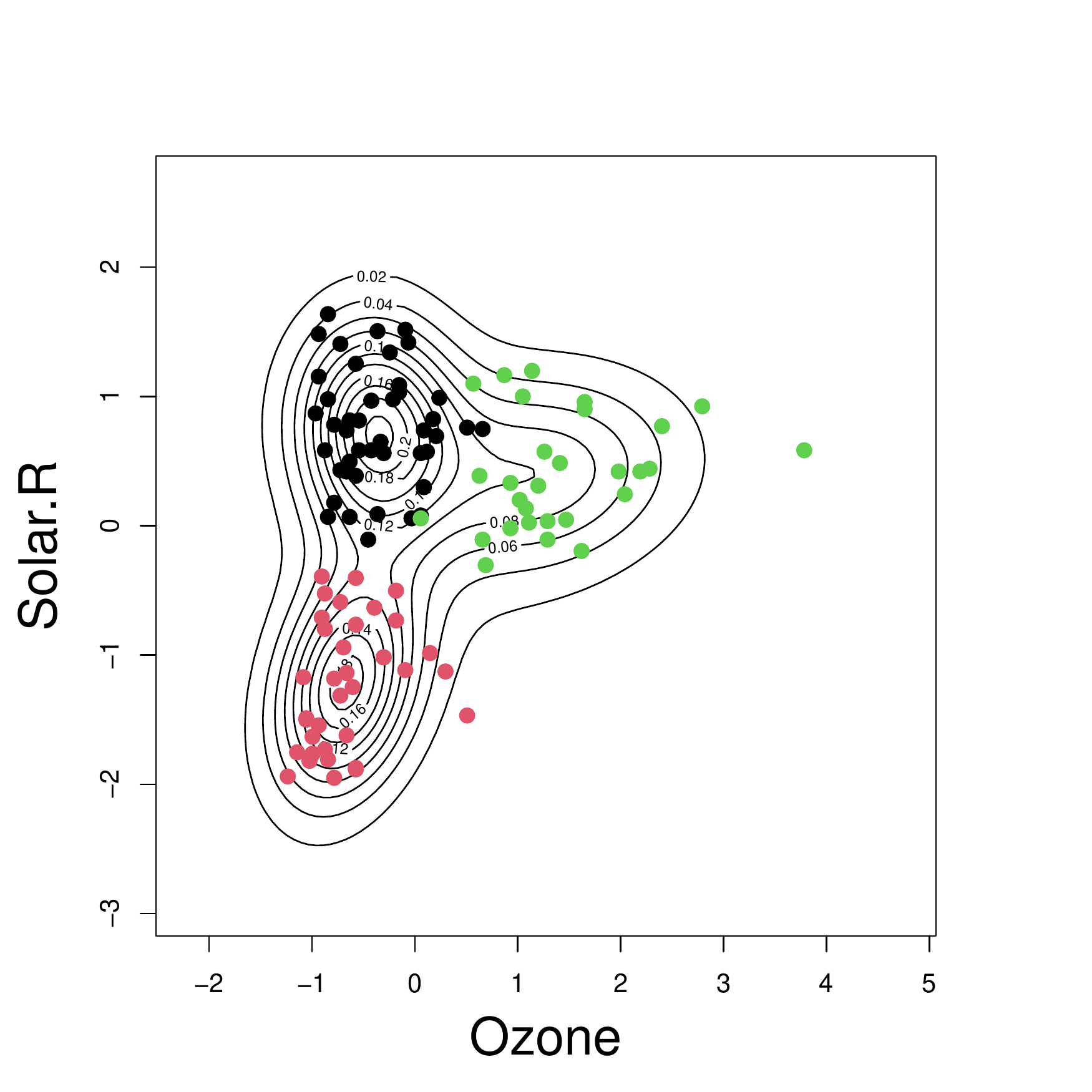}}
	\subfloat[{\tiny AntMAN}]{\includegraphics[width=0.35\linewidth]{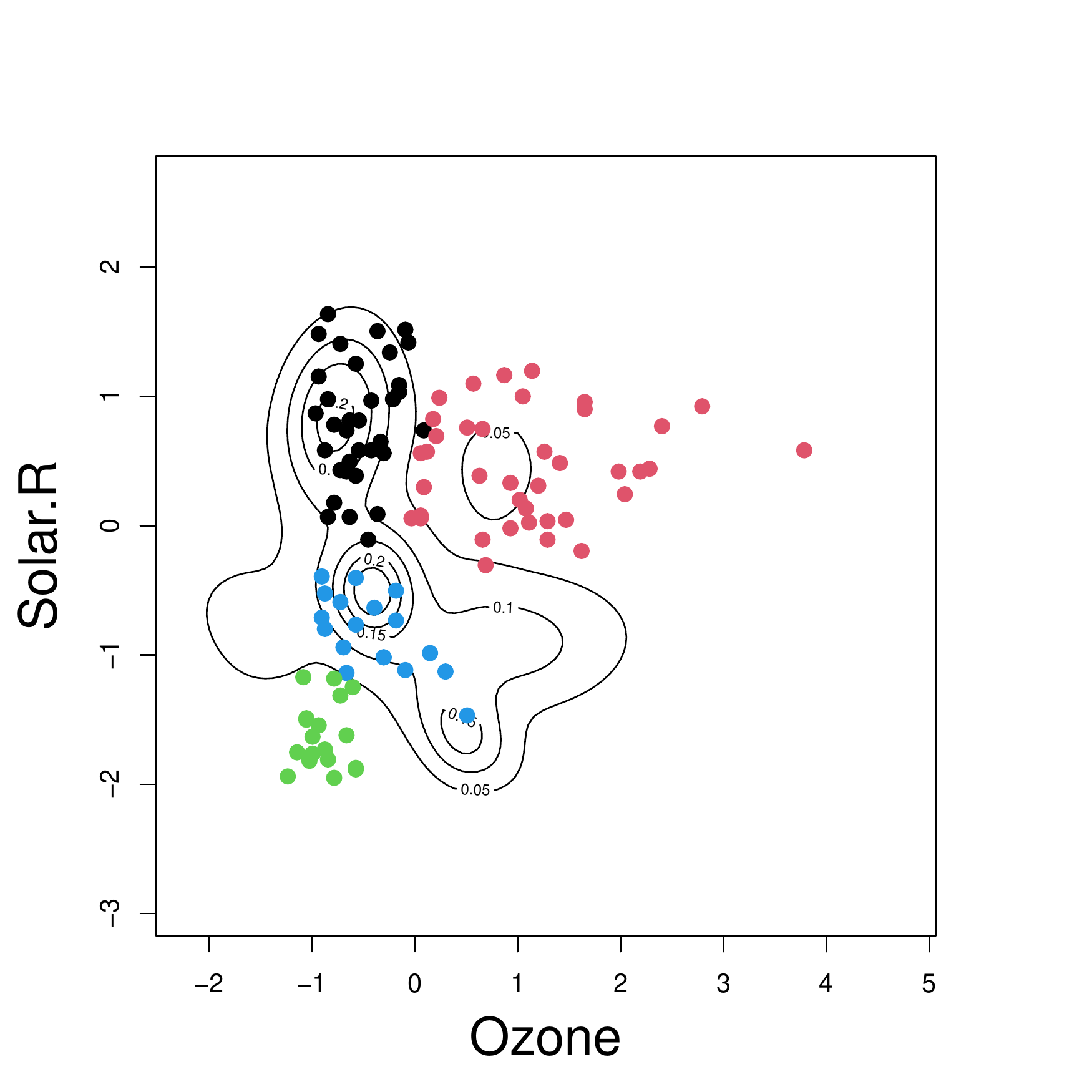}}
        \caption{Air Quality data. Contour plots of the predictive distributions obtained from a mixture model with   (a) the  Gaussian Coulomb prior; (b) the repulsive finite mixture proposed by \cite{quinlan2021class}; (c) non-repulsive prior as implemented in the R package \texttt{AntMAN}. Different colours refer to cluster assignment as estimated by minimising the Binder loss function.}
	\label{fig:AirQuality_data_contour}
\end{figure}

\section{Discussion}\label{sec:concl}

In this work, we develop a wide class of Bayesian repulsive mixture models that encourages well separated
clusters, with the goal of reducing potentially redundant components produced by independent priors for the location parameters. The prior distributions we propose are based on well-known eigenvalue distributions of specific random matrices, whose theoretical properties allow for efficient computation and robust posterior inference as compared to other repulsive priors. We refer to such prior distributions as Coulomb priors for their relationship with the joint Gibbs canonical distributions used to model Coulomb gases.

A key property of Coulomb priors is a soft penalisation of components close together, which leads to sparsity in the number of estimated components. A key advantage of our approach is the availability of the normalising constant in closed form, as well as the existence of a unique infinite dimensional distribution for $M\rightarrow \infty$. 
These properties allow specifying also a prior on the number of components in the mixture (differently from previous approaches), the ability of devising efficient computational schemes and the possibility of performing posterior inference on the repulsion parameters, leading to substantially improved clustering performance in general applications.
We show that compared to the independent prior on the component centres and competitor repulsive approaches, the Coulomb priors induce additional shrinkage effect on the tail probability of the posterior number of components, reducing model complexity. 

There are many extensions to our work. Promising venues of future research are the extension of dependent priors in the context of nested clustering, and the specification of a prior on the weights of the mixture which favours large components (as in \cite{fuquene2019choosing}), while maintaining a Coulomb prior on the locations.

\begin{appendices}

\section{Appendix A: Proofs}\label{appendix1}

%

\subsection{A.1: Proof of Lemma \ref{lemma1}}\label{appendix1:lemma1proof}
\lemmaone*

\begin{proof}
Note that since the unit square is compact, we can immediately conclude that the sequence is tight. Then, there exists a subsequence $\{ \mu_{N_i}\}$ that converges weakly to some $\hat{\mu} \in \mathcal{M}([0,1])$, and
\begin{equation*}
\begin{gathered}
    \lim_{i \to \infty} \int\limits_0^1\int\limits_0^1 F_{N_i}(x,y) d\mu_{N_i}(x) d\mu_{N_i} (y) = \\
    \liminf_{N \to \infty} \int\limits_0^1\int\limits_0^1 F_{N}(x,y) d\mu_{N}(x) d\mu_{N}(y)
\end{gathered}
\end{equation*}
Now,
\begin{align*}
\int\limits_0^1\int\limits_0^1 F(x,y) d\mu_0(x) d\mu_0(y) & \leq \int\limits_0^1\int\limits_0^1 F(x,y) d\hat{\mu}(x) d\hat{\mu}(y) \\ & = \sup_{R > 0} \int\limits_0^1\int\limits_0^1 F_R(x,y) d\hat{\mu}(x) d\hat{\mu}(y) \\ & = \sup_{R > 0} \lim_{N \to \infty} \int\limits_0^1\int\limits_0^1 F_{R,N}(x,y) d\mu_{N_i}(x)d\mu_{N_i}(y) \\ & \leq \lim_{N \to \infty} \int\limits_0^1\int\limits_0^1 F_{R,N}(x,y) d\mu_{N_i}(x)d\mu_{N_i}(y)
\end{align*}
which yields the desired inequality.
\end{proof}

\subsection{A.2: Proof of Lemma \ref{lemma2}}\label{appendix1:lemma2proof}
\lemmatwo*

\begin{proof}
We begin by rewriting the density \eqref{betaLDP} as
\begin{equation}
    \begin{gathered}
        \exp\left\{ -2 \frac{N^2}{M^2}  \sum_{i < j} F_{N} (\theta_i,\theta_j) \right\} \times \\
        \exp\left\{ \sum_{i = 1}^M\frac{\gamma(N_1)}{M}\log (\theta_i) + \frac{\gamma(N_2)}{M}\log(1 - \theta_i)  \right\}
    \end{gathered}
\end{equation}
Now, note that
\[ \sum_{i < j} F_{N} (\theta_i,\theta_j) = \iint\limits_{x \not= y} M^2 F_{N} (x,y) d\mu_{\bm U_M}(x) d\mu_{\bm U_M}(y) \]
Then,
\begin{align*}
\mathcal{B}_{N_1,N_2, \zeta, M} & \leq \left[ \int\limits_0^1 \exp\left( \frac{\gamma(N_1)}{M} \log(x) + \frac{\gamma(N_2)}{M} \log(1-x) \right) dx \right]^M \times \\ 
& \exp \left( - N^2 \iint\limits_{x \not= y} F_{N} (x,y) d\mu_{\bm U_M}(x) d\mu_{\bm U_M}(y) \right)
\end{align*}
Since
\[ \sup_{N \geq 1} \int\limits_0^1 \exp\left( \frac{\gamma(N_1)}{M} \log(x) + \frac{\gamma(N_2)}{M} \log(1-x) \right) dx < \infty \]
we have
\begin{align*} 
\limsup_{N \to \infty} \frac{1}{N^2} \mathcal{B}_{N_1,N_2, \zeta, M} & \leq - \liminf_{N \to \infty} \iint\limits_{x \not= y} F_{N} (x,y) d\mu_{\bm U_M}(x) d\mu_{\bm U_M}(y) \\ & \leq - \int\limits_0^1\int\limits_0^1 F(x,y) d\mu_0(x) d\mu_0(y)
\end{align*}
\end{proof}

\subsection{A.3: Proof of Lemma \ref{lemma3}}\label{appendix1:lemma3proof}
\lemmathree*

\begin{proof}
Let $\mu \in \mathcal{M}([0,1])$ and let $G$ be a neighbourhood of $\mu$. For $t \in [0,1]^M$, define $D_{t,M}$ to be the $M \times M$ diagonal matrix whose entries are given by $t$. Define $\tilde{G} := \{ t \in [0,1]^M : \mu_{D_{t,M}} \in G \}$. Then, letting $\nu$ denote the measure corresponding to the density $p(\bm \theta | \gamma)$, we have
\[ \mu_{\bm U_M}(G) = \nu(\tilde{G}) \]
Now, note that $\mu_{\bm U_M} \otimes \mu_{\bm U_M}\left( \{x=y\} \right) = \frac{1}{M}$, from which we see
\begin{equation*}
    \begin{gathered}
        \iint f_{R,N}(x,y) d\mu_{\bm U_M}(x) d\mu_{\bm U_M}(y) = \\
        \iint\limits_{x \not = y} f_{R,N}(x,y) d\mu_{\bm U_M}(x) d\mu_{\bm U_M}(y) + \frac{R}{M}
    \end{gathered}
\end{equation*}
Then, rewriting the density as before, we obtain
\begin{align*} 
&\mu_{\bm U}(G) = \nu(\tilde{G})  \\ & \leq \mathcal{B}_{N_1,N_2, \zeta, M}^{-1} \left[ \int\limits_0^1 \exp\left( \frac{\gamma(N_1)}{M} \log(x) + \frac{\gamma(N_2)}{M} \log(1-x) \right) dx \right]^M \\  & \times \exp \left( - N^2 \inf_{\sigma \in G} \iint\limits_{x \not= y} F_{R,N} (x,y) d\sigma(x) d\sigma(y)  + NR \right)
\end{align*}
for any $R > 0$. Moreover, we know that
\begin{equation*}
    \begin{gathered}
        \lim_{N \to \infty} \left( \inf_{\sigma \in G} \iint F_{R,N}(x,y) d\sigma(x)d\sigma(y) \right) = \\
        \inf_{\sigma \in G} \iint F_R(x,y) d\sigma(x) d\sigma(y)
    \end{gathered}
\end{equation*}
Hence
\begin{equation*}
    \begin{gathered}
        \limsup_{N \to \infty} \frac{1}{N^2} \log \mu_{\bm U_M}(G) \leq \\
        - \inf_{\sigma \in G} \iint F_R(x,y) d\sigma(x) d\sigma(y) - \liminf_{N \to \infty} \frac{1}{N^2} \mathcal{B}_{N_1,N_2, \zeta, M}
    \end{gathered}
\end{equation*}
Since $F_R(x,y)$ is bounded and continuous, it defines a continuous functional so that
\begin{equation*}
    \begin{gathered}
        \inf_{G} \left( \limsup_{N \to \infty} \frac{1}{N^2} \log (\mu_{\bm U_M}(G)) \right) \leq - \\
        \iint F_R(x,y) d\mu(x) d\mu(y) - \liminf_{N \to \infty} \frac{1}{N^2} \log \mathcal{B}_{N_1,N_2, \zeta, M}
    \end{gathered}
\end{equation*}
Taking the limit as $R \to \infty$ and applying the monotone convergence theorem yields the desired result.
\end{proof}

\subsection{A.4: Proof of Lemma \ref{lemma4}}\label{appendix1:lemma4proof}
\lemmafour*

\begin{proof}
 Without loss of generality, we may assume that $\mu$ has a continuous density $f$ on $[0,1]$. Then, there exists a $\epsilon > 0$ such that $\epsilon \leq f(x) \leq \frac{1}{\epsilon}$ for $x\in[0,1]$. Next, for each $N$, define constants 
\[ 0 = s_{0,N} < r_{1,N} < s_{1,N} < \cdots < r_{M,N} < s_{M,N} = 1 \]
such that
\[ \int\limits_0^{r_{i,N}} f(x) dx = \frac{i - 1/2}{M} \hspace{1pc} \mbox{and} \hspace{1pc} \int\limits_0^{s_{i,N}} f(x) dx = \frac{i}{M} \]
Then we have
\[ \frac{\epsilon}{2M} \leq s_{i,N} - r_{i,M} \leq \frac{1}{2M\epsilon} \]
Now, define
\[ \Delta_N = \left\{ (t_1, \dots, t_M) \in \mathbb{R}^M : r_{i,N} \leq t_i \leq s_{i,N} \right\} \]
For any neighbourhood $G$ of $\mu$ we can choose $N$ large enough so that $\Delta_N \subset \tilde{G}$. Thus,
\begin{align*} 
&\mu_{\bm U_M}(G) = \nu(\tilde{G}) \geq \nu(\Delta_N)  \\
&= \mathcal{B}_{N_1,N_2, \zeta, M}^{-1} \int\cdots\int\limits_{\Delta_N} \prod_{i = 1}^M t_i^{\gamma}(1 - t_i)^{\gamma} \prod_{i < j} (t_i - t_j) dt_1\dots dt_M \\
& \geq \mathcal{B}_{N_1,N_2, \zeta, M}^{-1} \left( \frac{\zeta}{2M} \right)^M \prod_{i = 1}^M r_{i,N}^{\gamma} \prod_{i < j} 	d_{i,j,N}
\end{align*}
where $d_{i,j,N} = \min\{ |x-y| : r_{i,N} \leq x \leq s_{i,N}, r_{j,N} \leq y \leq s_{j,N}\}$. Now, we observe that
\begin{equation*}
    \begin{gathered}
        \lim_{N \to \infty} \sum_{i = 1}^M \frac{\gamma(N_1)}{N^2} \log r_{i,N} + \frac{\gamma(N_2)}{N^2} log(1 - r_{i,N}) = \\
        A \int \log x + \log(1-x) d\mu(x)
    \end{gathered}
\end{equation*}
and
\begin{equation*}
    \lim_{N \to \infty} \frac{1}{N^2} \sum_{i < j} \log(r_{j,N} - s_{i,N}) = a^2\zeta \int \log|x-y| d\mu(x)d\mu(y)
\end{equation*}
Thus,
\begin{equation*}
    \begin{gathered}
        \limsup_{N \to \infty} \frac{1}{N^2} \log \mu_{\bm U_M}(G) \geq \\
        - \iint F(x,y) d\mu(x)d\mu(y) - \liminf_{N \to \infty} \frac{1}{N^2} \log \mathcal{B}_{N_1,N_2, \zeta, M}
    \end{gathered}
\end{equation*}
After taking the infimum over $\mu$, this implies
\[ \liminf_{N \to \infty} \frac{1}{N^2} \log \mathcal{B}_{N_1,N_2, \zeta, M} \geq - \iint F(x,y) d\mu_0(x)d\mu_0(y) \]
Moreover, we have
\begin{equation*}
    \begin{gathered}
        \liminf_{N \to \infty} \log \mu_{\bm U_M}(G) \geq \\
        - \iint F(x,y)d\mu(x)d\mu(y) - \limsup_{N \to infty} \frac{1}{N^2} \log \mathcal{B}_{N_1,N_2, \zeta, M}
    \end{gathered}
\end{equation*}
\end{proof}

\end{appendices}

%
%
%

\clearpage
\bibliographystyle{plainnat}
\bibliography{Biblio}

\end{document}